\documentclass[11pt]{article}

\usepackage{amsmath,amssymb,amsfonts}
\usepackage{amsthm}
\usepackage{mathrsfs}
\usepackage{graphicx}
\usepackage{booktabs}
\usepackage{multirow}
\usepackage{natbib}
\usepackage{geometry}
\usepackage{setspace}
\usepackage{enumitem}
\usepackage{hyperref}
\usepackage{bm}

\usepackage[textwidth=8em,textsize=small]{todonotes}
\usepackage[plain,noend]{algorithm2e}
\usepackage{enumerate}
\usepackage{url} 
\usepackage{xcolor}

\theoremstyle{plain}
\newtheorem{theorem}{Theorem}
\newtheorem{lemma}{Lemma}

\newtheorem{corollary}{Corollary}

\theoremstyle{definition}

\newtheorem{assumption}{Assumption}

\theoremstyle{remark}
\newtheorem{remark}{Remark}

\def\bmu{{\boldsymbol \mu}}
\def\bEta{{\boldsymbol \eta}}
\def\blambda{{\boldsymbol \lambda}}
\def\bbeta{{\boldsymbol \beta}}
\def\bSigma{{\boldsymbol \Sigma}}
\def\bOmega{{\boldsymbol \Omega}}
\def\btheta{{\boldsymbol \theta}}
\def\bLambda{{\boldsymbol \Lambda}}
\def\0{{\bf 0}}
\def\B{{\bf B}}
\def\A{{\bf A}}
\def\z{{\bold z}}
\def\Y{{\bold Y}}
\def\y{{\bold y}}

\def\I{{\bold I}}

\newcommand{\argmax}{\mathop{\rm arg\max}}

\def\hat{\widehat}
\def\tilde{\widetilde}
\def\T{{ \mathrm{\scriptscriptstyle T} }}
\newtheorem{condition}{Condition}
\newtheorem{lem1}{Lemma S.}

\def\0{{\bf 0}}
\def\A{{\bf A}}

\def\B{{\bf B}}

\def\e{{\bf e}}

\def\H{{\bf H}}
\def\I{{\bf I}}

\def\T{{\bf T}}

\def\T{{\bf T}}
\def\Y{{\bf Y}}
\def\y{{\bf y}}

\def\z{{\bf z}}
\def\1{{\bf 1}}

\def\bmu{{\boldsymbol \mu}}
\def\bEta{{\boldsymbol \eta}}
\def\bbeta{{\boldsymbol \beta}}
\def\bSigma{{\boldsymbol \Sigma}}
\def\bOmega{{\boldsymbol \Omega}}
\def\btheta{{\boldsymbol \theta}}
\def\0{{\bf 0}}
\def\B{{\bf B}}
\def\A{{\bf A}}
\def\z{{\bold z}}
\def\Y{{\bold Y}}
\def\y{{\bold y}}

\title{\bf Distributed inference for heterogeneous mixture models using multi-site data}
\author{Xiaokang Liu\thanks{Department of Statistics and Data Science, University of Missouri, xiaokang.liu@missouri.edu} \and  Rui Duan \thanks{Department of Biostatistics,
       Harvard University, rduan@hsph.harvard.edu} \and Raymond J. Carroll\thanks{Department of Statistics,
       Texas A\&M University, carroll@stat.tamu.edu} \and Yang Ning\thanks{Department of Statistics and Data Science,
       Cornell University, yn265@cornell.edu} \and Yong Chen\thanks{Corresponding author. Department of Biostatistics, Epidemiology and Informatics,
       University of Pennsylvania, ychen123@upenn.edu}}
\date{}

\begin{document}
\maketitle

% -------------------- Abstract --------------------
\begin{abstract}
Mixture models postulate the overall population as a mixture of finite subpopulations with unobserved membership. Fitting mixture models usually requires large sample sizes and combining data from multiple sites can be beneficial. However, sharing individual participant data across sites is often less feasible due to various types of practical constraints, such as data privacy concerns. Moreover, substantial heterogeneity may exist across sites, and locally identified latent classes may not be comparable across sites. We propose a unified modeling framework where a common definition of the latent classes is shared across sites and heterogeneous mixing proportions of latent classes are allowed to account for between-site heterogeneity. To fit the heterogeneous mixture model on multi-site data, we propose a novel distributed Expectation-Maximization (EM) algorithm where at each iteration a density ratio tilted surrogate Q function is constructed to approximate the standard Q function of the EM algorithm as if the data from multiple sites could be pooled together.  Theoretical analysis shows that our estimator achieves the same contraction property as the estimators derived from the EM algorithm based on the pooled data. %We also applied our distributed EM algorithm to an analysis of long-COVID using data from 6,005 patients across 9 hospitals.  
\end{abstract}


\newpage

\section{Introduction}\label{sec:intro}

Mixture models analyze complex data sets by postulating the overall population as a mixture of finite subpopulations with unobserved membership, which are also referred to as latent classes. Mixture models allow characterizations of each subpopulation's distribution and mixing proportions \citep{Lindsay1995MixtureM}. Due to their  flexibility, mixture models have a wide range of applications in many fields including biomedical studies \citep{peel2000finite}. 
%%\citep{qin2005semiparametric, follmann2014estimating}.
%% instead of these two papers by Qin, I would cite the book below. Xiaokang, please add it
%% %{Xiaokang, add the citation to support this: Finite Mixture Models by McLachlan and Peel}
For example,     
%Mixture model and its utility in public health: e.g., cite Qin Jing's papers. 
%Importance of disease subtyping
disease subtyping plays a critical role in disentangling some syndromic diseases by identifying meaningful biological and clinical subphenotypes \citep{li2015identification, neff2021molecular}. A recent application is the work by \cite{su2021clinical} to derive subphenotypes for COVID-19, which is well-known for its variable host responses and clinical manifestations. Based on routinely collected clinical data from five health systems in New York City, four biologically different subphenotypes were identified which differed significantly in demographics, clinical variables, and chronic comorbidities, and were found to be predictive of patient mortality. These results are important in advancing our understanding of the varied biological disease mechanisms and facilitate subsequent pathophysiological studies on COVID-19.   
% (e.g., Su Chang 2021 NPJ digital medicine; COVID; Gaussian Mixture model).
%[4-7 sentences]

Fitting mixture models usually requires large sample sizes \citep{peel2000finite}. %{Xiaokang, add the citation to support this: Finite Mixture Models by McLachlan and Peel}
Combining data across multiple healthcare organizations provides the opportunities to obtain larger bodies of data from a more general population. 
As in COVID-19 subtyping to disentangle the complex clinical manifestations of COVID-19 by finding its subphenotypes,  
it is of great clinical importance to properly characterize the distinct distributions of latent classes by fully utilizing the data from the five healthcare systems, which can provide a better
understanding of the features manifested by each class and the determinants of differentiation between classes. Often, disease subtypes identified from a general population involving all the sites are of more clinical interest due to the potentially improved interpretability and generalizability compared to  locally identified site-specific disease subtypes \citep{calfee2014subphenotypes,sinha2018latent}. % XXXX (0.5 sentence; See Su Chang's paper). 
However, using data from multiple sites raises two major challenges: first, sharing individual participant data across clinical sites is sometimes %{\color{magenta}\bf give a reference? Dr. Chen, can you please add a reference here?} 
logistically prohibitive or practically infeasible due to privacy concerns \citep{ohno2012share}. Second, site-level data heterogeneity needs to be properly modeled and accounted for when applying mixture models across multiple data sets.
A naive approach is to apply mixture models locally within each site and then perform cross-site matching and combining based on certain similarity-based metrics. However, since locally identified latent classes may not be comparable across sites due to potential heterogeneity and  label-switching issues (classes are identifiable up to a permutation), matching local classes may introduce additional errors due to potential mismatching. 
%% \textcolor{violet}{Old: Applying latent class analysis locally within each site might be suboptimal as the classes may not be comparable across sites due to potential heterogeneity and the label-switching issue (class memberships are identifiable up to an arbitrary permutation). Although we can perform cross-site matching based on certain similarity measures between locally identified latent classes, errors of local estimation cannot be corrected when matching local classes, and it may introduce matching errors.} {\color{red}\bf It is not clear to me what you mean here. RJC}

Recently, there has been a growing interest in developing distributed algorithms  which allow jointly analyzing multiple data sets based on summary-level statistics. For example, \cite{wang2017efficient} and \cite{jordan2018communication} proposed the idea of a surrogate likelihood function to approximate the global likelihood function using local data, and \cite{battey2018distributed} and \cite{fan2019distributed} analyzed properties of an aggregation of locally obtained estimates, see also  \cite{zhang2013communication}, \cite{dobriban2020wonder} and \cite{dobriban2021distributed}.  In addition, \cite{chen2019quantile} and \cite{chen2020distributed} considered  approximating the distributions of Newton-type estimators, and the same idea is also considered in \cite{chen2021distributed, chen2021first}. Most of these methods were developed under a homogeneous assumption assuming all sites share the same distribution, 
% Fan2019distributed investigated distributed PCA and also considers heterogeneity. 
%see, e.g., \cite{battey2018distributed, dobriban2020wonder, dobriban2021distributed, fan2019communication,  jordan2018communication, lee2017communication, lian2017divide,   shamir2014communication, wang2017efficient, wang2019distributed, zhang2013communication}, 
which is less practical in real-world biomedical settings because multi-center data are likely to be heterogeneously distributed as they often represent different populations and may be exposed to different environments. %Heterogeneous settings are more appropriate in reflecting the intrinsic differences in patient populations across clinical sites. 
However, limited efforts have been devoted to address data heterogeneity issues in distributed learning and inference. More recently, \cite{cai2021individual} proposed a high-dimensional integrative regression that allows for heterogeneity in both the covariate distribution and model parameters. \cite{duan2019heterogeneity} developed a density ratio tilted efficient score function based approach to accommodate the site-specific nuisance parameter. To our knowledge, distributed unsupervised learning algorithms using heterogeneous mixture models have not been investigated.  
%On the other hand, most of the existing work focus on supervised learning (e.g., linear regression in \cite{dobriban2020wonder, dobriban2021distributed}, generalized linear model in \cite{fan2019communication} and classification in \cite{lian2017divide, wang2019distributed}) with few exceptions on the unsupervised problem (e.g., principal component analysis in \cite{chen2021distributed, fan2019distributed,  garber2017communication}), which is especially less on the heterogeneous case even though this kind of problem is not rare in reality. 

In this paper, we propose a unified modeling framework to allow a common characterization of the latent classes shared across different sites, yet we use heterogeneous mixing proportions of latent classes in mixture models to account for between-site heterogeneity. Such a formulation aligns with the practical needs of identifying  latent classes shared across sites.
%naturally avoids the class-matching issue. %Moreover, it models the  underlying latent classes in the overall population by combining data from all sites, while allowing the proportion of data from each latent class to vary across sites. 
For estimation, we propose a novel aggregated data based distributed Expectation-Maximization (EM) algorithm to allow joint analysis of multi-center data, which avoids the need for label matching of latent classes and addresses the challenge of sharing individual participant data.
%Different from the seminal work by  \cite{wang2017efficient} and \cite{jordan2018communication}, 
%As EM naturally involves a missing data formulation,  %instead of working with a surrogate likelihood function \citep{wang2017efficient, jordan2018communication}, 
We develop a novel construction of a surrogate Q function to approximate the standard Q function in the EM algorithm based on the pooled data which may not exist in practice because of privacy concerns. The surrogate Q function only requires sharing of aggregated data from sites. By construction, the gradient of surrogate Q function can match that of the standard Q function, and the high-order derivatives of the standard Q function are approximated using a density ratio tilting technique. Our theoretical analyses show that the resulting estimator from our heterogeneity-aware distributed EM algorithm retains the same contraction property as the estimator derived from the EM algorithm based on the pooled data and is consistent in estimating the unknown parameters.

%{\color{red} \bf Important to emphasize that the pooled data may not exist because of privacy concerns. RJC}

%defined as the expected value of the log-likelihood function, with respect to the current conditional distribution of the unobserved latent class variable given the observed data and the current estimates of the parameters.
%The construction of the surrogate Q function only . 

%We show 
 %We use the Gaussian mixture model as an example and validate out results %showed that the estimator from our distributed algorithm can achieve the same performance as the pooled EM-algorithm estimator
%from both theoretical and numerical analyses. 
%achieves the optimal convergence rate for the characterization of the distribution for the shared latent classes. Asymptotic results are also provided.

Our main contributions are summarized as follows: First, we propose a novel heterogeneity-aware distributed EM algorithm based on a unified formulation of multi-site mixture models. It characterizes between-site heterogeneity using site-specific class proportions and retains interpretability by setting the same definition of each latent class across sites. Second, via the novel construction of the density ratio tilted surrogate Q function, we obtain an aggregated data based estimator which achieves the same contraction property as the estimator derived from the EM algorithm based on the pooled data. Third, we use the Gaussian mixture model as an important example and validate our results %showed that the estimator from our distributed algorithm can achieve the same performance as the pooled EM-algorithm estimator
with theoretical analyses and numerical studies. %\textcolor{purple}
 Lastly, although our main motivation is from the practical needs of latent class analysis,  EM algorithms are widely applicable to a broader class of missing data problems due to their computational stability and theoretical foundations \citep{Dempster1977}, and our distributed EM algorithm can shed new lights to problems beyond mixture models. %Further, we established the asymptotic distribution (XXX theoretical property) of the estimator from our distributed algorithm, and showed that it shares the same asymptotic normal distribution as the pooled EM-algorithm estimator. %(Xiaokang to fill)

Our work is different from previous research in several significant aspects: 
\begin{itemize}
	\item To the best of our knowledge, our approach is among the first attempts to develop distributed learning strategies for unsupervised learning tasks. Specifically, we propose a novel heterogeneous mixture model to tackle the unique challenge of unsupervised learning, i.e., the class mismatching issue, which often impedes the application of commonly used average-type methods.
	\item Our work significantly departs from previous studies that primarily investigate the theoretical properties of the traditional EM algorithm on a single data set. We instead work on multi-site studies where data is stored at multiple locations that cannot be pooled together and which have heterogeneous characteristics. As a result, the traditional EM algorithm cannot be directly applied, and its distributed version requires a novel design that includes significant changes to the original algorithm. This fundamental difference makes the theoretical analysis more challenging as well. 
	\item As far as we know, we are the first group to work on Q functions, which involve unknown parameters and depend on current parameter estimates through conditional expectation. As a result, for the distributed EM algorithm to achieve the contraction property, the verification of some regularity conditions, particularly the smoothness conditions in terms of the current estimates of the parameters, requires much more effort. 
\end{itemize}

The rest of the paper is organized as follows. In Section 2, we introduce the problem setup and our distributed algorithm. Section 3 develops the theoretical properties for a general model and the results are then applied to a Gaussian mixture model in Section 4. In Section 5 we present some simulation results. Finally we conclude in Section 6.  

\section{Methodology}\label{sec:method}

In what follows, we introduce the problem setup and our approach.

\subsection{A Proposed Heterogeneous Mixture Model Under Multi-site Settings}\label{subsec:back}

We formulate the problem in a distributed learning setting with $K$ study sites, and assume that the observations within each site come from a mixture of $S$ distributions, and each distribution is characterized by a set of distribution-specific parameters $\bmu_c\in\mathbb{R}^d$ for $c\in [S]$ where we use $[S]$ to denote $\{1, \ldots, S\}$ for simplicity, and the same rule applies to the following contexts. To characterize the heterogeneity across sites, we assume that different sites have their own site-specific mixing proportion  vectors, i.e., $\blambda_{j}$ may or may not be the same as $\blambda_{\ell}$ for $j \neq \ell$ where $\blambda_j=(\lambda_{j1},\ldots,\lambda_{jS})^\T$ with $\lambda_{jc}\in (0,1)$ representing the proportion of the $c$th distribution in the site $j$, and $\sum_{c=1}^S \lambda_{jc}=1$. To convey our idea in its simplest form, without loss of generality, in the rest of this paper we consider mixture models with $S = 2$. {Therefore, in the following context, $\lambda_j$ reduces to a scalar between $(0,1)$ for all $j\in[K]$.} The extension to settings with $S > 2$ is algebraically tedious but conceptually straightforward.

%{\color{red}\bf Might mention what S is in the example. RJC} {\color{blue}I don't think it is necessary.}

For simplicity of notation, we assume that in  site $j\in [K]$ we collect $n$ independently and identically distributed observations $\{\y_{ij}\in\mathbb{R}^d; i \in [n]\}$. Let  $N=nK$  be the total sample size across $K$ sites. Our method applies to unequal sample size settings with some modifications of notations. Hereafter, we assume equal sample size across sites for simplicity in notations. The $i$th observation in the $j$th site follows a distribution 
\begin{align}
	%f(\y_{ij}; \bmu_0, \bmu_1, \lambda_j) = 
 f_j(\y_{ij};\btheta_j)=\lambda_j f(\y_{ij}; \bmu_1) +(1-\lambda_j) f(\y_{ij}; \bmu_0), \label{eq-main}
\end{align}
where $\btheta_j=(\lambda_j, \bmu^\T)^\T$, $\bmu= (\bmu_0^\T, \bmu_1^\T)^\T$, and $f(\y;\bmu)$ is a known function indexed by the parameter $\bmu$. This formulation explicitly postulates the same set of latent classes across sites, which allows identification and characterization of latent classes shared across sites. Different $\lambda_j$'s account for the between-site heterogeneity, which indicates that the proportion of individuals from a given latent class can vary across sites. %{\color{red}[sing the advantages of this formulation]} 
% Since the same set of latent classes are shared, it is expected that by integrating information from multiple sites the estimation efficiency of the shared parameters can be significantly improved.
 %{\color{red}As such, in this paper our target is to jointly estimate $\bmu$, while the $\lambda_j$'s are left as nuisance parameters. }
%In the following contexts, a distributed algorithm to solve \eqref{eq-main} will be elaborated, and the corresponding theoretical analysis will follow.
For convenience, we let $\btheta = (\bmu^\T,\bLambda^\T)^\T \in\Theta \subset\mathbb{R}^{2d+K} $  and $\bLambda=(\lambda_1, \dots,\lambda_K)^\T$, and the true parameter values are denoted by  $\btheta^* = (\bmu^{* \mathrm{\scriptscriptstyle T} },\bLambda^{* \mathrm{\scriptscriptstyle T}})^\T$.

\subsection{Standard EM Algorithm Based on the Pooled Data}\label{subsec:EM}

Model \eqref{eq-main} is essentially a missing data problem where each individual has an unobserved latent class membership denoted by  $Z_{ij} \sim Bernoulli(\lambda_j)$, and we have $$ \text{Pr}(\Y_{ij}=\y_{ij}\mid Z_{ij} = k) = f(\y_{ij}; \bmu_k), \text{ for } k=0,1.$$ To fit the latent class model using an EM algorithm,  we first consider an ideal situation where the data across all sites are available and could be pooled together. %{\color{red}\bf I interpret this to mean that all of the data across sites are available. Please correct me if needed, and all an explanation in the text what the term means. RJC}

Based on the pooled data, we have the global complete loglikelihood function %can then be written as 
\begin{equation*}
L_C = \frac{1}{nK}\sum_{j=1}^{K}\sum_{i=1}^{n}\left[Z_{ij}\log\left\{\lambda_jf(\y_{ij}; \bmu_1)\right\} +\left(1-Z_{ij}\right)\log\left\{\left(1-\lambda_j\right)f(\y_{ij}; \bmu_0)\right\} \right].
\end{equation*}
With a current parameter $\btheta^t$, the Q function in the expectation step is calculated as 
\begin{align}
& Q  \left(\btheta \mid \btheta^t\right)  = E\left(L_C \mid \btheta^t, \y\right) =\nonumber\\
 &  \frac{1}{Kn}\sum_{j=1}^{K}\sum_{i=1}^{n}\left[ w_{\btheta_j^t}^j (\y_{ij})  \log\left\{\lambda_jf(\y_{ij}; \bmu_1)\right\} %\nonumber  \\ & 
+\{1-w_{\btheta_j^t}^j \left(\y_{ij}\right)\}\log\left\{(1-\lambda_j)f(\y_{ij}; \bmu_0)\right\}\right],\label{eq:pooled-Q}
\end{align}
where the expectation is with respect to the unobservable variable $Z_{ij}$ conditional on the observed data and $\btheta^t$, and 
\begin{equation*}
w_{\btheta_j}^j (\y) = E\left(Z\mid \btheta_j, \y\right) = \frac{\lambda_j f(\y; \bmu_1)}{\lambda_j f(\y; \bmu_1) +(1-\lambda_j)f(\y; \bmu_0)}.%:=w_{\btheta_j}^j (\y)
\end{equation*}
Also, the local Q function only involving data from the $j$th site is denoted by 
\begin{align*}
Q_j(\btheta_j\mid \btheta_j^t) = \frac{1}{n}\sum_{i=1}^{n} \left[ w_{\btheta_j^t}^j (\y_{ij})   \log\{\lambda_j f(\y_{ij}; \bmu_1)\} +\{1-w_{\btheta_j^t}^j (\y_{ij}) \}\log\{(1-\lambda_j)f(\y_{ij}; \bmu_0)\} \right].
\end{align*}
The maximization step then updates the estimate by $$
\btheta^{t+1} = M_n(\btheta^t) = \argmax_{\btheta} Q(\btheta\mid \btheta^t). 
$$
By iteratively applying these two steps, the EM algorithm can provide a sequence of estimates $\{\btheta^t\}_{t \geq 0}$ that monotonically increase the likelihood function and converge to a consistent estimator under standard regularity conditions \citep{louis1982finding, balakrishnan2017statistical}. We denote the pooled EM estimator as $\widehat \btheta=(\widehat \bmu^\T, \widehat \bLambda^\T)^\T$.  
However, the above classical EM algorithm   requires individual participant data from all sites to construct the Q function in  \eqref{eq:pooled-Q} at each iteration, which cannot be applied directly in distributed data settings where individual participant data cannot be pooled together. % motivates our distributed EM algorithm.  

\subsection{An Aggregated Data Based Distributed EM Algorithm}\label{subsec:surrog}

Motivated by the common practice in many multi-center collaborations, we consider the setting where there is a lead site whose individual participant data are accessible  while only summary-level statistics are available from all other participating sites; for example, see \citet{duan2020learning} and \citet{luo2022odach}. In the  heterogeneous setting, data from the lead site  might not be representative of the overall population, which becomes a challenge in developing distributed algorithms. %As a consequence, the EM algorithm using the local data does not converge to the same mixing proportion parameter value as the EM algorithm  based on the pooled data. 
 To adjust for such heterogeneity, we propose a density ratio tilted surrogate Q function, which is constructed  using the individual participant data at the lead site and the gradients of the local Q function from the participating sites evaluated at the current estimates of the model parameters. {At each iteration, the participating sites only need to send the gradients to the lead site to construct the density ratio tilted surrogate Q function, and individual participant data from participating sites are not required to be shared in our distributed EM algorithm.}
%{\color{blue}\bf Yong: we may need to restate this at different places of the paper.}

Without loss of generality, we let site 1 be the lead site where individual participant data are accessible while all other sites can  only share summary-level statistics with the lead site.
%To solve the two problems simultaneously, 
With a current estimate $\btheta^t$, the density ratio tilted surrogate Q function is defined as
\begin{align}
\widetilde Q\left(\bmu\mid \btheta^t\right) = \check Q\left(\bmu\mid\btheta^t\right) + \left\langle \nabla_\bmu Q_{\bmu}\left(\bmu^t\mid\btheta^t\right)- \nabla_\bmu \check Q\left(\bmu^t\mid\btheta^t\right), \bmu \right\rangle, \label{DSQ-func}
\end{align}
where $Q_{\bmu} \left(\bmu \mid\btheta^t\right)$ includes only the terms  of the standard $Q$ function defined in (\ref{eq:pooled-Q}) that contain $\bmu$, i.e.,
%As $\lambda_j$ is only involved in the $j$th site, in the following we mainly focus on approximating the part of the Q function that is related to $\bmu$, i.e.,
\begin{align*}
Q_{\bmu}&\left(\bmu \mid\btheta^t\right) = \frac{1}{Kn}\sum_{j=1}^{K}\sum_{i=1}^{n}\left[w_{\btheta_j^t}^j (\y_{ij}) \log\{f(\y_{ij}; \bmu_1)\} +\{1-w_{\btheta_j^t}^j (\y_{ij})\}\log\{f(\y_{ij}; \bmu_0)\}\right],
\end{align*}
and $\check Q\left(\bmu\mid\btheta^t\right)$ is obtained using data only from the lead site, i.e.,
\begin{align*}
    \check Q&\left(\bmu\mid\btheta^t\right) \\ & =  \frac{1}{Kn}\sum_{j=1}^{K}\sum_{i=1}^{n}t\left(\y_{i1}, \bEta_j^t\right)
\left[ w_{\btheta_j^t}^j (\y_{i1})\log\{f(\y_{i1}; \bmu_1)\}  +
\{1-w_{\btheta_j^t}^j (\y_{i1})\}\log\{f(\y_{i1}; \bmu_0)\} \right],
\end{align*}
with the help of an density ratio term  defined as
\begin{equation}\label{eq_eta}
t\left(\y_{i1}, \bEta_j^t\right)=\frac{\lambda_j^tf\left(\y_{i1}; \bmu_1^t\right) +(1-\lambda_j^t)f\left(\y_{i1}; \bmu_0^t\right)}{\lambda_1^tf\left(\y_{i1}; \bmu_1^t\right) +\left(1-\lambda_1^t\right)f\left(\y_{i1}; \bmu_0^t\right)},\quad \bEta_j^t=\left(\bmu^{t \mathrm{\scriptscriptstyle T}},\lambda_1^t,\lambda_j^t\right)^\T. 
\end{equation}

The rationale behind the construction of the density ratio tilted surrogate Q function is to  approximate the global Q function  by  matching its gradient and the expectation of higher-order derivatives using only data from the lead site and the gradients $\nabla_{\bmu}Q_j(\btheta_j\mid\btheta_j^t)$ at $\btheta^t$ calculated from the rest of the sites. To see how the density ratio tilted surrogate Q function approximates the global Q function, it can be  verified that $\widetilde Q(\bmu\mid\btheta^t)$ has the same  gradient as $Q_{\bmu}(\bmu\mid\btheta^t)$ at $\bmu^t$. As for its second- and higher-order derivatives, we have
\begin{align}
\nabla_{\bmu}^p \widetilde Q\left(\bmu\mid\btheta^t\right) = \nabla_{\bmu}^p \check Q\left(\bmu\mid\btheta^t\right),\quad p \geq 2.\label{eq:gradient-match0}
\end{align}
A proof of equation \eqref{eq:gradient-match0} is provided in Supplementary Material S1.
In addition,  at the true values of the parameters $\btheta^*$, we have
\begin{align}
E_{\btheta_1^*} \{ \nabla_{\bmu}^p \check Q(\bmu\mid\btheta^*) \}= E_{\btheta^*} \{\nabla_{\bmu}^p Q_{\bmu}(\bmu\mid\btheta^*)\},\quad p \geq 1, \label{eq:gradient-match}
\end{align}
where $E_{\btheta_1^*}(\cdot)$ stands for the expectation with respect to the distribution $ f_1(\y;\btheta_1)$ and $E_{\btheta^*}(\cdot)$ represents the expectation with respect to the distribution of the pooled data of $K$ data sets. Equations \eqref{eq:gradient-match0} and \eqref{eq:gradient-match} imply that at the true $\btheta^*$, the density ratio tilted surrogate Q function have the same higher-order derivatives as the global Q function at the population level. Since the true coefficients $\btheta^*$ are unknown, we use the estimate $\btheta^t$ at the current iteration $t$ to approximate $\btheta^*$ and the resulting $\check Q(\bmu\mid\btheta^t)$ satisfies
\begin{align}
E_{\btheta_1^*} \{\nabla_{\bmu}^2 \check Q\left(\bmu\mid\btheta^t\right) \}= E_{\btheta^*} \{\nabla_{\bmu}^2 Q_{\bmu}\left(\bmu\mid\btheta^t\right)\} + o(1)%\quad {\color{red}p \geq 1} 
\label{eq:gradient-match-sample}
\end{align}
when $\nabla_{\bmu}^2 \check Q(\bmu\mid\btheta^t)$ and $\nabla_{\bmu}^2 Q_{\bmu}(\bmu\mid\btheta^t)$ satisfy standard %{\color{magenta}\bf ????} details added  
smoothness conditions in $\btheta^t$ (e.g., Lipschitz continuity%defined in Assumption 5 of \cite{duan2019heterogeneity}
) and $E(\|\bmu^t-\bmu^*\|^2_2) =o(1)$. Therefore, the density ratio tilted surrogate Q function $\widetilde Q(\bmu\mid\btheta^t)$ and the function $Q_{\bmu}(\bmu\mid\btheta^t)$ have the same gradient at $\bmu^t$, and their second-order derivatives have expectations with corresponding elements that only differ by $o(1)$, and these properties  %Also, note that \eqref{eq:gradient-match-sample} 
ensure
\begin{align}
    \widetilde Q\left(\bmu\mid\btheta^t\right) - Q_{\bmu}\left(\bmu\mid\btheta^t\right) = O\left\{ n^{-1/2} + o(1)\right\} \left(\| \bmu^t - \bmu \|_2^2 + \| \bmu^t - \bmu \|_2^3 \right).
\end{align}
Also, the construction of the density ratio tilted surrogate Q function only requires sharing the gradients $\nabla_{\bmu}Q_j(\btheta_j^t\mid\btheta_j^t)$ from the participating sites (i.e., $j > 1$). The communication cost is low since only the first-order terms are shared. 

It is worthwhile to compare the density ratio tilted surrogate Q function with the surrogate likelihood function proposed by \cite{jordan2018communication}. Specifically, instead of working with likelihood functions which are often non-convex in the case of a mixture model, here we focus on approximating the global Q function of the EM algorithm which has better computational stability. With the additional density ratio tilting component, we successfully adjust for the difference between the local data and the pooled data, which enables approximating the standard pooled $Q$ function using the local data with the presence of site-specific nuisance parameters. %$\lambda_j$, which is necessary to account for the site-level heterogeneity.

Based on the density ratio tilted surrogate Q function \eqref{DSQ-func}, we update the estimates as
\begin{align}
	\widetilde\bmu^{t+1} & %\widetilde M_n(\btheta^t)_{\bmu} 
	= \arg\max_{\bmu} \widetilde Q(\bmu\mid\btheta^t),\label{mu_upd} \\
	\widetilde\lambda_j^{t+1} & %\widetilde M_n(\btheta^t)_{\lambda_j}
% 	= \arg\max_{\lambda} Q_{n}^j(\lambda_j|\btheta_j^{t})
	=n^{-1}\sum_{i=1}^n w_{\btheta_j^{t}}^j(\y_{ij})\label{lam_upd}.
\end{align}
For simplicity, we let $ \widetilde\btheta^{t+1} = \widetilde M_n(\btheta^t) = (\widetilde\bmu^{t+1}, \widetilde\lambda_1^{t+1}, \ldots, \widetilde\lambda_K^{t+1}) %= (\widetilde M_n(\btheta^t)_{\bmu}^\T, \widetilde M_n(\btheta^t)_{\lambda_1},\cdots, \widetilde M_n(\btheta^t)_{\lambda_K})^\T 
% = \arg\max_{\btheta} \widetilde Q(\btheta|\btheta^t)$ where $\widetilde Q(\lambda_j|\btheta^t) = Q_{n}^j(\lambda_j|\btheta_j^{t})
$. The reason that we update $\bmu$ and the $\lambda_j$'s separately from different objective functions is due to the fact that only the data in the $j$th site  contain the information of $\lambda_j$ while data from all sites are informative to estimate $\bmu$. %We denote the distributed EM estimator as $\widetilde \btheta=(\widetilde \bmu^\T, \widetilde \bLambda^\T)^\T$. 

To summarize, at the $(t+1)$-st iteration, with the estimate $\btheta^t$ from the previous iteration, a density ratio tilted surrogate Q function $\widetilde Q(\bmu\mid\btheta^t)$ is built at the lead site based on $\{\y_{i1}\}_{i=1}^n$ and $\{\nabla_{\bmu}Q_j(\btheta_j^t\mid\btheta_j^t)\}_{j=1}^K$, and then the maximization step is implemented to update $\btheta$. The iteration continues until the algorithm reaches convergence or the prespecified iteration number, and we %name the final result as distributed EM estimator $\widetilde \btheta$.
denote the final estimator as the distributed EM estimator $\widetilde \btheta=(\widetilde \bmu^\T, \widetilde \bLambda^\T)^\T$. 
 Importantly, the whole procedure does not require sharing individual participant data from the participating sites and only aggregated data of the same dimension as $\bmu$ are transferred between sites. Therefore, this algorithm avoids communicating individual participant data, while accounting for between-site heterogeneity.  We summarize our algorithm below.
\begin{algorithm}
Algorithm 1: the distributed EM algorithm
	\begin{tabbing}
        \qquad \enspace 1.\enspace \enspace  Input: data $\{\y_{i1}\}_{i=1}^n$, initial estimates $\bmu^0$; \\
		\qquad \enspace 2.\enspace \enspace Initialize with $\widetilde \bmu^0 = \bmu^0$; \\
		\qquad \enspace 3.\enspace \enspace From $t=0$ iterate until converge: \\
		\qquad \enspace 4.  \enspace\quad\enspace In Site $j=1$ to $j=K$ \\
		\qquad \enspace 5.\enspace \qquad\quad Compute and transfer $\widetilde\lambda_j^t$ (by \eqref{lam_upd}) and  $\nabla_{\bmu} Q_j(\widetilde\btheta_j^t\mid\widetilde\btheta_j^t)$ to Site 1;\\
		\qquad \enspace 6. \enspace\quad\enspace  In Site 1 \\
     	\qquad \enspace 7.\enspace \qquad \qquad Construct $\widetilde Q(\bmu\mid\widetilde\btheta^t)$ using $\widetilde\btheta^t$ and $\{\nabla_{\bmu}Q_j(\widetilde\btheta_j^t\mid\widetilde\btheta_j^t)\}_{j=1}^K$;\\
	    \qquad \enspace 8.\enspace \qquad \qquad  Obtain $\widetilde\bmu^{t+1}$ by solving $\nabla_{\bmu} \widetilde Q(\bmu\mid\widetilde\btheta^t)=\0$ and broadcast $\widetilde\bmu^{t+1}$;\\	
	   	\qquad \enspace 9. \enspace Output: $\widetilde\btheta$  \\
	\end{tabbing}\label{alg1}
 \end{algorithm}
% %\vspace*{5pt}
% \begin{algorithm}
% \caption{The distributed EM algorithm}\label{alg1}
% \begin{algorithmic}[1]
% \STATE {\bf Input}: data $\{\y_{i1}\}_{i=1}^n$, initial estimates $\bmu^0$, tolerance level $\rho_{\textrm{tol}}$, maximum iteration number $T$;
% \STATE $\widetilde \bmu^0 = \bmu^0$, $\rho_{\textrm{dif}} = \rho_{\textrm{tol}} + 1$, $t=0$;
% \WHILE{$t<T$ and $\rho_{\textrm{dif}}> \rho_{\textrm{tol}} $}
% \FOR{$j$ in \{$1,2,...,K$\}}
% \STATE Compute and transfer $\widetilde\lambda_j^t$ (by \eqref{lam_upd}) and  $\nabla_{\bmu} Q_j(\widetilde\btheta_j^t\mid\widetilde\btheta_j^t)$ to Site 1;
% \ENDFOR
% \STATE In Site 1, 
% \STATE \quad Construct $\widetilde Q(\bmu\mid\widetilde\btheta^t)$ using $\widetilde\btheta^t$ and $\{\nabla_{\bmu}Q_j(\widetilde\btheta_j^t\mid\widetilde\btheta_j^t)\}_{j=1}^K$;
% \STATE \quad Obtain $\widetilde\bmu^{t+1}$ by solving $\nabla_{\bmu} \widetilde Q(\bmu\mid\widetilde\btheta^t)=\0$ and broadcast $\widetilde\bmu^{t+1}$;
% \STATE $\rho_{dif} = \| \widetilde\bmu^{t+1} - \widetilde\bmu^{t} \|_2$;
% \STATE $t = t + 1$;
% \ENDWHILE
% \STATE {\bf Output}: $\widetilde\btheta$.
% \end{algorithmic}
% \end{algorithm}
\begin{remark}
	A natural choice of the initial estimates $\bmu^0$ is from an EM algorithm fitted locally in the lead site. Theoretically, we show in the next section that an initial value has to fall in a neighborhood of true parameter values to guarantee the  convergence and consistency of our final estimator.  %{\color{magenta}\bf Last sentence is not clear} (modified) 
	Under mild regularity conditions, the local EM estimator satisfies the initialization condition  when the local sample size is not too small. In some distributed algorithms, the initial estimator can be obtained by averaging local estimates from all sites, which can further improve the accuracy  \citep{huang2019distributed}. However, in the mixture model setting, the locally identified latent classes need to be matched before obtaining  an average-type initial value, and the matching step may introduce errors especially when the estimation accuracy of some local estimators is low. Therefore, in practice, we suggest using a site with relatively large sample size for initialization.  Given $\bmu^0$, the initialization of $\lambda_j$'s is achieved from Equation  \eqref{lam_upd}  by plugging in $\btheta^t_j = (\bmu^0,  0.5)$ or from optimizing \eqref{eq-main} with the given $\bmu^0$. When obtaining  $\bmu^0$  from fitting an EM algorithm locally in the lead site,  initialization of the local EM need to be carefully chosen as discussed in \citet{biernacki2003choosing}.
	
	%There is no need to use multiple initializations to run our distributed EM algorithm since the initial value $\bmu^0$ is a local estimate obtained from multiple initializations and the distributed EM algorithm is to further improve the local estimates.  %{\color{red}Any other reasonable choice of $\lambda_j^0$ other than 0.5? use $\bmu^0$ to get $\lambda_j^0$ and use $\btheta^0=(\bmu^0,\Lambda^0)$ as initial value, and no need to use multiple initialization except in getting $\bmu^0$.}
\end{remark}

\begin{remark}
A classical EM algorithm often requires two layers of iteration: the outer loop to update the Q function and an inner loop to optimize the Q function unless there is a closed form solution. Our distributed EM algorithm enables the inner loop to be completely conducted at the lead site without further communications among sites, thanks to the construction of the density ratio tilted surrogate Q function. Although the outer loop is unavoidable, the reduction in the communication cost of the inner loop is substantial in a distributed analysis setting.
\end{remark}

\section{Theoretical Analysis}\label{sec:generalTheory}

In this section we investigate the theoretical properties of our distributed EM algorithm. The contraction property of an algorithm states the ability and the speed of an algorithm to shorten the distance between its iterates and the true parameter value through each iteration \citep{cai2019chime}. Under Assumptions \ref{loc_conc}--\ref{assump:init} described below, we show that %when the sample size $n$ is large enough to make the approximation error be ignorable, 
our distributed EM algorithm can achieve the same contraction property as the EM algorithm based on the pooled data. In other words, the distributed EM algorithm can achieve the best possible estimation performance to learn a mixture model without sharing individual participant data.  %It also shares the same asymptotic performance as the pooled EM estimator.     
%{\color{magenta}\bf What is a contraction property? Also, something is wrong with the notation in the next displayed equation} corrected

We first introduce some notations. We use  $\|\cdot\|_2$ to denote the $\ell_2$ norm when applied to a vector and it is the operator norm if applied to a matrix, i.e., the largest singular value of a matrix. To measure the distance between two estimates, for any $\btheta,\widetilde\btheta\in\Theta$, we define 
\begin{eqnarray*} d_2(\btheta,\widetilde\btheta)&=&\left(\hbox{$\sum_{j=1}^K$}|\lambda_j-\widetilde\lambda_j|^2\right)^{1/2}+\hbox{$\sum_{k=0}^1$}\|\bmu_k-\widetilde\bmu_k\|_2, \\
d_2(\btheta_j,\widetilde\btheta_j)&=&|\lambda_j-\widetilde\lambda_j|+\hbox{$\sum_{k=0}^1$}\|\bmu_k-\widetilde\bmu_k\|_2. 
\end{eqnarray*}
Also, we write $a_n \lesssim b_n$ for two sequences $\{a_n\}$ and $\{b_n\}$ when there exists a constant $c$ such that $a_n \leq c b_n$ for all $n$. %{\color{magenta}\br this is just that $a_n/b_n = O(1)}; right$}  % A sub-Gaussian random variable $X$ is a random variable that satisfies $P(|X|>t) \leq \exp(1-ct^2/\|X\|_{\psi_2}^2)$ for all $t \geq 0$ with $\|X\|_{\psi_2}^2 = \sup_{p \geq 1} p^{-1/2}(E|X|^p)^{1/p}$ \citep{vershynin2010introduction}. %{\color{red} sub-gaussian definition}
%$a \vee b = \max\{a, b\}$ and $a \wedge b=\min\{a, b\}$. {\color{red}The notation $\lesssim $ means ...}

Suppose the parameter space of $\btheta_j$, denoted by $\Theta_j$, is a compact and convex set and the true parameter is an interior point of $\Theta_j$. This implies that the parameter space of $\bmu$ is also a compact and convex set and the true parameter $\bmu^*$ is an interior point. %$\Theta$ is a compact and convex set, the true parameter $\btheta^*$ is an interior point of $\Theta$ and the $\ell_2$-radius is $R=\max_{\btheta\in\Theta}\|\btheta-\btheta^*\|_2$. 
We need the following assumptions to derive the contraction properties for our distributed EM algorithm estimator.

 \begin{assumption}[Local strong concavity]\label{loc_conc}
Let $\mathcal{Q}(\btheta\mid\btheta^t)=E\{Q(\btheta\mid\btheta^t)\}$ be the population objective function and $\mathcal{Q}^*(\btheta)=\mathcal{Q}(\btheta\mid\btheta^*)$. There exists some $\mu_+, \mu_- > 0$ s.t.
%\[
%q(\btheta') - q(\btheta'') - \left< \nabla q(\btheta''), \btheta' - \btheta''\right> \leq -\bmu_-/2 \|\btheta' - \btheta''\|_2^2
%\] 
%for all pairs $\btheta'$, $\btheta'' \in B_2(r\Delta; \btheta^*)$.
$-\mu_+ \I \preceq \nabla_{\btheta}^2 \mathcal{Q}^*(\btheta^*) \preceq -\mu_- \I $ where $\I$ is an identity matrix and $\A \preceq \B $ means $\B-\A$ is positive semidefinite.  
 \end{assumption}

{%\color{blue}
\begin{assumption}[Smoothness]\label{assump:smooth}
	For each $j \in [K]$, define 
	\begin{align*}
		h(\y;\bmu,\btheta_j') = w_{\btheta_j'}^j (\y)\log\{f(\y; \bmu_1)\} +\{1-w_{\btheta_j'}^j (\y)\}\log\{f(\y; \bmu_0)\}. 
        %& H(\bmu,\btheta_j'; y) = \nabla_{\bmu\bmu}^2 h(\y;\bmu, \btheta_j'),\quad \widetilde H(\bmu, \bEta_j'; y) 
        %=\nabla_{\bmu\bmu}^2 h(\y;\bmu, \btheta_j') t(y, \bEta_j'). 
	\end{align*}
	Let $U_{\btheta^*} (\rho) = \{\btheta';\|\btheta'-\btheta^*\|_2\le\rho\}$ be a neighborhood around $\btheta^*$ for some radius $\rho>0$, and $U_{\bmu^*}(\rho)$, $U_{\btheta_j^*}(\rho)$, and $U_{\bEta_j^*}(\rho)$ are defined in a similar way, where $\bEta_j$ is defined in (\ref{eq_eta}). There exist some functions $m_k (\cdot)$, $k=1,2,3,4$, %$m_2(\cdot)$, and $m_3(\cdot)$ where $m_1 (\y_{ij})$, $m_2(\y_{ij})$, and $m_3(\y_{ij})$ are sub-Gaussian distributed for all $j \in [K]$ and $i \in [n]$, 
	such that for %$\btheta_j$ and $\btheta_j'\in U_{\btheta_j}(\rho)$, we have
	any $\bmu$, $\bmu' \in U_{\bmu^*}(\rho)$, $\bar\btheta_j$, $\bar\btheta_j'\in U_{\btheta_j^*}(\rho)$, $\bar\bEta_j$, $\bar\bEta_j'\in U_{\bEta_j^*}(\rho)$ with any $j\in[K]$, we have
%	\[
%	\|H(\btheta_j; y) -H(\btheta_j'; y)\|_2 \le m_1(\y)\|\btheta-\btheta_j'\|_2.
%	\]
\begin{align*}
    & \|\nabla_{\bmu\bmu}^2 h(\y_{ij};\bmu, \bar\btheta_j) - \nabla_{\bmu\bmu}^2 h(\y_{ij};\bmu', \bar\btheta_j')\|_2 \le m_1(\y_{ij})(\|\bmu-\bmu'\|_2+\|\bar\btheta_j-\bar\btheta_j'\|_2),\\
    & \|t(\y_{i1},\bar\bEta_j)\nabla_{\bmu\bmu}^2 h(\y_{i1};\bmu, \bar\btheta_j)  -t(\y_{i1}, \bar\bEta_j')\nabla_{\bmu\bmu}^2 h(\y_{i1};\bmu', \bar\btheta_j')\|_2 \le m_2(\y_{i1})(\|\bmu-\bmu'\|_2+\|\bar\bEta_j-\bar\bEta_j'\|_2),\\ 
    & | w_{\bar\btheta_j}(\y_{ij}) - w_{\bar\btheta_j'}(\y_{ij}) | \leq m_3(\y_{ij})\|\bar\btheta_j- \bar\btheta_j'\|_2, \\
    & \|\nabla_{\bmu\btheta_j}^2 h(\y_{ij};\bmu, \bar\btheta_j) - \nabla_{\bmu\btheta_j}^2 h(\y_{ij};\bmu', \bar\btheta_j')\|_2 \le m_4(\y_{ij})(\|\bmu-\bmu'\|_2+\|\bar\btheta_j-\bar\btheta_j'\|_2),
\end{align*}
where $t(\y_{i1},\bar\bEta_j)$ is given in (\ref{eq_eta}). 
%	and
%	\[|w_{\bar\btheta_j}^j (Z) - w_{\bar\btheta_j'}^j (Z)|\leq m_3(Z)\|\bar\btheta_j-\bar\btheta_j'\|_2.\]
We require that $E\{m_k(\Y_{ij})^8\} \leq L^8$ and $E([m_k(\Y_{ij})- E \{ m_k(\Y_{ij})\} ]^8) \leq L^8$ with some finite constant $L$ for all $k$ and $j$. 
	Also, there are finite constants $G$, $H$, $J$ and $C$ such that the first and the second partial derivatives of $h$ exist and satisfy
	\begin{align*}
		& E\{\|\nabla_{\bmu} h(\y_{ij},\bmu^*,\btheta_j^*)\|_2^8\}\leq G^8, \\
		& E\left[\|\nabla_{\bmu\bmu}^2 h(\y_{ij},\bmu^*,\btheta_j^*)- E \{ \nabla_{\bmu\bmu}^2 h(\y_{ij},\bmu^*,\btheta_j^*)\}\|_2^8\right] \leq H^8, \\
		& E\left[\| t(\y_{i1},\bEta_j^*) \nabla_{\bmu\bmu}^2 h(\y_{i1},\bmu^*,\btheta_j^*)-E \{\nabla_{\bmu\bmu}^2 h(\y_{ij},\bmu^*,\btheta_j^*)\}\|_2^8\right]\leq J^8, \\
		& E\left[\|\nabla_{\bmu\btheta_j}^2 h(\y_{ij},\bmu^*,\btheta_j^*)- E \{\nabla_{\bmu\btheta_j}^2 h(\y_{ij},\bmu^*,\btheta_j^*)\}\|_2^8\right]\leq C^8.
	\end{align*}
\end{assumption}
}

\begin{assumption}[Initialization and pooled contraction]\label{assump:init} %With probability at least $1-1/n$, the initial estimator satisfies $d_2(\btheta_j^0, \btheta_j^*) \lesssim \sqrt{n^{-1}}$
Given an initial estimator that
satisfies $d_2(\btheta_j^0, \btheta_j^*) \lesssim \{\log(n)/n\}^{1/2}$, %{\color{magenta}\f something wrong here}corrected
with probability at least $1-K/n-1/(Kn)$, the EM algorithm iterates $\{ \btheta^t \}_{t \geq 1 }$ based on the pooled data satisfy $d_2(\btheta^t, \btheta^* ) \leq \kappa d_2(\btheta^{t-1}, \btheta^*) + O\{K\log(n)/n\}^{1/2}$ with $\kappa \in (0, 1)$. In particular, there is
$d_2(\btheta^t_j, \btheta_j^* ) = O_p\{\log(n)/n\}^{1/2}$ for any $j \in [K]$. 
\end{assumption}

Assumption \ref{loc_conc} requires the population objective function to be concave around $\btheta^*$ to induce consistency of the estimator. Assumption \ref{assump:smooth} requires the Hessian matrices and the function $w^j_{\btheta_j}(\y)$ to be smooth in the neighborhood of the optimal point $\btheta^*$, and is essential %{\color{magenta}\bf only for $K=2$?} K is the number of sites 
to control the approximation error between the distributed estimator and the pooled estimator \citep{zhang2013communication%, duan2019heterogeneity
}. The moment conditions control the tail of the gradient and Hessian of the Q functions.  Assumption \ref{assump:init} requires the EM algorithm iterates obtained from the pooled data are contracting towards $\btheta^*$ once the initial estimator is good enough, %{\color{red} It is confusing, maybe we should remove this sentence}.  
%Assumption \ref{assump:init} 
and it is a necessary basis since the proposed distributed EM algorithm approximates the pooled EM and the pooled EM needs to converge. Moreover, %Assumption \ref{assump:init} implies a bound on $d_2(\btheta^t,\btheta^*)$ as $d_2(\btheta^t,\btheta^*) \leq \sqrt{K}\max_{j\in[K]} d_2(\btheta_j^t,\btheta_j^*) \lesssim \sqrt{K/n}$. 
since $\sum_{k=0}^1\| \bmu_k^t - \bmu_k^* \|_2 \leq d_2(\btheta^t_j, \btheta_j^* )$, we also have $\sum_{k=0}^1\| \bmu_k^t - \bmu_k^* \|_2 \lesssim \{\log(n)/n\}^{1/2}$. % and this bound can be further improved to $(nK)^{-1/2}$ when the site-specific parameter $\lambda_j$'s effects on the estimation of $\bmu$ can be resolved \citep{duan2019heterogeneity}. 
We ignore the $\log(n)$ factor in the following contexts since it is a small term compared to $n$. The verification of Assumption \ref{assump:init} for Gaussian mixture models is shown in the next section. 
%for a specific model can follow the results derived by \cite{balakrishnan2017statistical}.   

With Assumptions \ref{loc_conc}--\ref{assump:init}, 
%Based on the above contraction results of the oracle estimator, we can develop similar contraction property for the \proposed\ estimator. 
we obtain contraction properties for the proposed estimator. % follow the following steps. 
Recall that  $\btheta^{t}=M_n(\btheta^{t-1}) = \argmax_{\btheta} Q(\btheta\mid\btheta^{t-1})$ is the EM algorithm estimator based on the pooled data and $\widetilde\btheta^{t}=\widetilde M_n(\widetilde\btheta^{t-1}) $ is the distributed EM algorithm estimator. %$ = \argmax_{\btheta} \widetilde Q(\btheta|\btheta^t)$.  
Since 
\begin{align*}
	d_2(\widetilde\btheta^{t}, \btheta^* ) 
	\leq d_2( \widetilde\btheta^{t}, \btheta^{t} ) + d_2(\btheta^{t},\btheta^*),
\end{align*}
and we bound $d_2(\btheta^{t},\btheta^*)$ by Assumption \ref{assump:init}, it remains to control the approximation error $d_2(\widetilde\btheta^{t}, \btheta^{t})$.

\begin{lemma}\label{lemma1}
Under Assumptions \ref{loc_conc}--\ref{assump:init}, %when $K^3 < n$,  
%we have \begin{align*}
%	E \|{\widetilde{M}}_n (\btheta) - M_n(\btheta)\|_2 \lesssim \frac{K}{n}.
%\end{align*} 
with probability at least $1- n^{-2/3} - K/n$, %for $\btheta \in \B_2(r\Delta;\btheta^*)$ 
$d_2(\widetilde\btheta_j^{t}, \btheta_j^{t}) = O(n^{-5/6})$.
\end{lemma}
 
A proof of Lemma \ref{lemma1} is provided in Supplementary Material S2. Lemma \ref{lemma1} indicates that for estimating parameters in a single site, the proposed estimator approximates the pooled EM estimator with an approximation error of order $n^{-5/6}$, which is smaller than  the estimation error of the pooled EM algorithm of order $n^{-1/2}$. Thus, combining all parameters across $K$ sites, $d_2(\widetilde\btheta^t,\btheta^t)$ will be no greater than $n^{-1/3}(K/n)^{1/2}$, which is ignorable compared to the estimation error occurred at each iteration of the pooled data-based EM algorithm.  Therefore, our distributed EM algorithm estimator achieves 
a similar contraction behavior as the pooled estimator at each iteration. We summarize the analysis results formally in the following theorem. 
 
\begin{theorem}\label{coro1}
Under Assumptions \ref{loc_conc}--\ref{assump:init}, with probability at least $1-n^{-2/3}-K/n$, we have $$ d_2(\widetilde\btheta^t, \btheta^*) \leq \kappa^t d_2(\btheta^0, \btheta^*) + O\{(K/n)^{1/2}\}.$$ Thus, when $t$ is large, we have 
$d_2(\widetilde\btheta^t, \btheta^*) = O\{(K/n)^{1/2}\}.
$
\end{theorem}
A proof of Theorem \ref{coro1} is in Supplementary Material S3. Theorem \ref{coro1} shows that the distributed estimator $\widetilde\btheta$ is a consistent estimator and achieves the same estimation accuracy as the pooled estimator.  In the next section, we further illustrate our theoretical results under  Gaussian mixture models.

\section{Distributed EM Algorithm for Heterogeneous Gaussian Mixture Models}\label{sec:gaussian}

%%{\color{magenta}\bf this is for $K=2$}

Within site $j\in[K]$, although our method is generally applicable to any finite mixture models, we assume that the $i$th subject independently and identically follows a two-component Gaussian mixture model (i.e., $S=2$), 
$ \lambda_j N_d(\bmu_1,\bSigma) +  (1-\lambda_j)N_d(\bmu_0, \bSigma)$,
where $\bmu_0,\bmu_1\in\mathbb{R}^d$ are the unknown mean vectors, $\lambda_j\in(0,1)$ is the unknown mixing proportion of site $j$. To simplify the theoretical analysis, we consider the case that the covariance matrix $\bSigma$ is known. We further require there are some positive constants $M$ to make $M^{-1}\leq \lambda_{min}(\bSigma)\leq \lambda_{max}(\bSigma)\leq M$ %for all $j\in[K]$ 
where $\lambda_{min}(\cdot)$ and $\lambda_{max}(\cdot)$ denote the smallest and largest eigenvalues, respectively. %{\color{red}We aim to jointly estimate the two mean vectors.} 
Following the steps introduced in Algorithm \ref{alg1}, we can obtain the distributed EM algorithm estimator. Next, we show that Assumptions 1-3 in Section \ref{sec:generalTheory} hold in Gaussian mixture models. 

Assumption \ref{loc_conc} can be easily verified as $\mathcal{Q}^*(\btheta)$ can be written as a summation of several strongly concave terms. The verification of Assumption \ref{assump:smooth} is relegated to Supplementary Material S7. To verify Assumption \ref{assump:init}, we need to investigate the theoretical performance of the EM algorithm estimator $\widehat\btheta$ on the pooled data. Specifically, we consider a parameter space $$\Theta=\{\btheta=(\bmu_0,\bmu_1,\lambda_1,\ldots,\lambda_K):\lambda_j\in(c_w,1-c_w),\ j\in[K],\ \bmu_0,\bmu_1\in \mathbb{R}^d \}$$ with $0<c_w<1/2$.
We also define $\Delta=\{(\bmu_0^*-\bmu_1^*)^\T\bSigma^{-1}(\bmu_0^*-\bmu_1^*)\}^{1/2}$ as the signal to noise ratio,  
%of site $j$ and we further require that there exist $\Delta_{min}, \Delta_{max}>0$, s.t., $ \Delta_{min}=\sup_{\Delta}\{\Delta \leq \Delta_j, \text{for all}\ j\in [K] \} $ and $ \Delta_{max}=\inf_{\Delta}\{\Delta \geq \Delta_j, \text{for all}\ j\in [K] \} $. Thus, for any $j\in [K]$ we have $\Delta_{min}\leq \Delta_j\leq\Delta_{max}$, and $\Delta_{min}$ can be viewed as the global signal to noise ratio in our multi-site learning setting. In addition, we require $\Delta_{max} = c_2 \Delta_{min}$ with a positive constant $c_2>1$.
and require the following condition to restrict
%This condition requires 
the initial estimator be within a small neighborhood around the true parameter at a radius proportional to $\Delta$. % noise ratio of the problem.
\begin{condition}\label{cond2}The initial estimator $\btheta^0$ satisfies $d_2(\btheta^0,\btheta^*)\leq r\Delta$, with %{\color{magenta} what is $\wedge$. Next should be 2 lines} done
\begin{align*}
    r<\frac{M^{3/2}}{4}\wedge\frac{|c_0-c_w|}{\Delta} & \wedge\left\{\left(\frac{2c_1-1}{M}+\frac{4}{M}\right)^{1/2}- \frac{2}{M^{1/2}} \right\} \\
    & \wedge\left[\left\{ \frac{c_1}{M} + \frac{1}{4}\left(M + \frac{1}{M} +2\right) \right\}^{1/2}  - \frac{1}{2}(M^{1/2} + M^{-1/2})\right],
\end{align*}
where $M$ is the upper bound of  $\lambda_{max}(\bSigma)$, $c_0$ and $c_1$ are constants which satisfy $0 < c_0 \leq c_w < 1/2$ and $1/2 < c_1<1$, and $\wedge$ is a binary operator to take the smaller one between two items. 
\end{condition}
\noindent This kind of initialization requirement is commonly seen in non-convex problems \citep{loh2013regularized}.  %{\color{red} what is M?}
%This condition is %sufficient to guarantees the initial estimator $\btheta^0$ falls into a contraction region $B(\btheta^*;c_0,c_1)$,
%Furthermore, we require the initial values to fall into a contraction region $B(\btheta^*;c_0,c_1)$,
Under some mild regularity conditions on the sample size, we can verify that for any $\btheta^0 \in \Theta$, Condition \ref{cond2} guarantees $\btheta^0$ and all the subsequent EM algorithm iterates $\btheta^t$ for $t \geq 1$ are within a contraction region $\B(\btheta^*;c_0,c_1)$ whose exact form is in Supplementary Material S6. The contraction region is crucial to guarantee the EM algorithm to find a global solution \citep{balakrishnan2017statistical, cai2019chime}. When deriving the contraction property of $\btheta^t$, $\btheta^t \in \B(\btheta^*;c_0,c_1)$ for $t \geq 0$ is a prerequisite to use a uniform upper-bound on the difference between the sample-based and population-level EM updates on the contraction region.

The next result formally characterizes the contraction property of $\btheta^t$ based on the pooled data, and the proof is deferred to Supplementary Material S4 where the exact forms of constants $\kappa,\ \kappa'$, and $\kappa''$ can be found.

\begin{theorem}\label{overall_contraction} Consider the Gaussian mixture model over the parameter space $\Theta$, under Condition \ref{cond2} and assume  $\Delta > C(c_0,c_1,M,K)$ with $C(c_0,c_1,M,K)$ being a positive quantity that depends on constants $c_0$, $c_1$, $M$,  
and grows with $K$ with the rate of $\log(K)^{1/2}$, then there exist constants $\kappa, \kappa', \kappa'' %\lesssim \exp(-\Delta_{min}^2) 
\in(0,1)$ %{\color{red}$(\kappa=\sqrt{K}\kappa_1+\kappa_2)$} 
such that when $n$ is large enough to make $T(n,K) =  \{K\log(n)/n\}^{1/2} \leq (1-\kappa)r\Delta$,  %$o(1)$, 
we have with probability at least $1-Kn^{-1}-(nK)^{-1}$, 
\begin{align}
    d_2(\btheta^{t+1},\btheta^*)\leq \kappa^{t+1} d_2(\btheta^*,\btheta^0) + \frac{1-\kappa^{t+1}}{1-\kappa}T(n,K). \label{theta}
\end{align}
In particular, at the $(t+1$)st iteration, with probability at least $1- n^{-1} - (nK)^{-1}$,
\begin{align}
    d_2(\btheta_j^{t+1},\btheta_j^*) \leq \kappa' K^{-1}\sum_{m=1}^K d_2(\btheta_m^{t},\btheta_m^*) + \kappa'' d_2(\btheta_j^{t},\btheta_j^*) + O\{\log(n)/n\}^{1/2}. \label{theta_j} % + O_p(\sqrt{\frac{\log(nK)}{nK}}).
\end{align}
\end{theorem}
The first term in the right hand side of inequality (11) decreases geometrically in the iteration number $t$, and the latter term measures the estimation error accumulated along the iterations. When $t$ is large enough to make the former term to be dominated by the latter term, the iterates $\btheta^t$ fall in a ball of radius $O\{T(n,K)\}$ centered at the true parameter vector $\btheta^*$.  The bounds in %\eqref{mu} and 
\eqref{theta_j} further depicts the contraction of each parameter component at each site.  The involvement of $K$ nuisance parameters $\lambda_j$'s induces the term $K^{1/2}$ in $T(n,K)$, which does not exist in %\eqref{mu} and 
\eqref{theta_j}. The quantity $C(c_0,c_1,M,K)$ relies on $K$ mainly through the term $\log(K)^{1/2}$, which is a small term since the number of sites $K$ is usually not too big in practice.  This requirement on the signal to noise ratio is due to the increasing number of nuisance parameters when $K$ increases, thus requiring a larger signal to noise ratio to guarantee the contraction parameter $\kappa$ to be less than 1.         
 
Theorem \ref{overall_contraction} successfully verifies Assumption \ref{assump:init} in two-ways: first, when $K=1$, it includes the local estimator as a special case and shows that the local estimator achieves an estimation error rate of $O_p(n^{-1/2})$. %\citep{cai2019chime}. 
Therefore, the local estimator satisfies Condition \ref{cond2} when $n$ is large enough and it is eligible to serve as an initial value of our distributed algorithm. Secondly, the EM algorithm applied to the pooled data produces contractive iterates $\{\btheta^t\}_{t\geq 1}$. In particular, \eqref{theta_j} together with $d_2(\btheta_j^0, \btheta_j^*)=O_p(n^{-1/2})$ ensures $d_2(\btheta_j^t, \btheta_j^*)=O_p(n^{-1/2})$ for all $t \geq 1$. %{\color{red} what happened to the $\log n$ factor?}
 %{\color{red} the discussion is not clear. You should talk about the three results in above theorem one by one.}
Therefore, we have now verified all the Assumptions \ref{loc_conc}--\ref{assump:init}, and we get the following result on the contraction behavior of the distributed EM estimator under heterogeneous Gaussian mixture models.

\begin{corollary}\label{coro2}
%{\color{red}add a statement as: from (12) we can obtain $d_2(\btheta_j^t,\btheta_j^*)=n^{-1/2}$ which matches with the requirement stated in Assumption \ref{assump:init}, therefore we have the following results...} 
Under the conditions in Theorem \ref{overall_contraction}, our distributed EM estimator satisfies  
$ d_2(\widetilde\btheta^t,\btheta^*) \leq \kappa^t d_2(\btheta^0,\btheta^*) +  O_p\{(K/n)^{1/2}\}. $
In particular, when $t$ is large enough, we have
$d_2(\widetilde\btheta^t,\btheta^*) = O_p\{(K/n)^{1/2}\}. $
\end{corollary}

%{\color{red} $O(\exp(\Delta_{min}^2)n^{-1/2})$ be dominated by $O((nK)^{-1/2})$? How about $\exp(\Delta_{min}^2)\lesssim K^{-1/2}$? $o(1)$? $K^3 < n$?}

\section{Simulation Study}\label{sec:simu}
To illustrate the empirical performance of our method, we conducted simulation studies with data generated from the heterogeneous Gaussian mixture model,
\[
\Y_{ij}\sim \lambda_j N_d(\bmu_1 ,\bSigma) + (1-\lambda_j) N_d(\bmu_0, \bSigma), \ j\in[K],\ i\in[n].
\]
The shared parameters to be estimated are $(\bmu_1,\bmu_0)$. The nuisance parameter $(\lambda_1, \ldots, \lambda_K)$ is generated from $U(0.5 - a, 0.5 + a)$ where $a = 0.1$ or $0.3$ represents respectively a low or high level of heterogeneity across sites. We let $d  = 5$, $\bmu_1^* = (5, \ldots, 5)^\T \in \mathbb{R}^d$ and $\bmu_0^* = (4, \ldots, 4)^\T \in \mathbb{R}^d$, and the variance-covariance matrix $\bSigma = \sigma^2 \I_d$ with $\sigma^2 = 2.5$ or 5 representing respectively a high or low level of signal to noise ratio (i.e., $\{ (\bmu_1^* - \bmu_0^*)^\T \bSigma^{-1}(\bmu_1^* - \bmu_0^*) \}^{1/2} $) of the problem. Moreover, for each of these scenarios, we let $K \in \{10, 30\}$ and $n\in \{1000, 3000\}$ to see the impacts of the number of sites and sample size, respectively, on the performance of the distributed EM algorithm.

The methods under comparison are: (1) our distributed EM algorithm estimator;
	(2) the pooled estimator obtained by applying the EM algorithm to the pooled data, which we call the gold standard;
	(3) the average estimator $(\hat\bmu_{avg,1},\hat\bmu_{avg,0})$.
	The average estimator is obtained by first fitting a local model in each site separately to get      
$(\hat\bmu_{j1},\hat\bmu_{j0})$ for all $j\in [K]$. Then we use the lead site estimator as an anchor, and match the classes identified from other sites based on the distance to the estimated class centriods. Specifically, for each site $j\in\{2,\ldots,K\}$ 
	we compute 
	$a_{j}^1 = \|\hat\bmu_{j1} - \hat\bmu_{11}\|_2 + \|\hat\bmu_{j0} - \hat\bmu_{10}\|_2 ,$ and $a_{j}^2 = \|\hat\bmu_{j1} - \hat\bmu_{10}\|_2 + \|\hat\bmu_{j0} - \hat\bmu_{11}\|_2 .$
	Then the average estimator is %{\color{magenta}\bf do not use this acronym} has removed all of them  
	calculated as $\hat\bmu_{avg,1}=w_1\hat\bmu_{11} + \sum_{j=2}^K w_j \{ 1(a_j^1<a_j^2) \hat\bmu_{j1} + 1(a_j^1 \geq a_j^2) \hat\bmu_{j0} \},$ and
	$\hat\bmu_{avg,0}=w_1\hat\bmu_{10} + \sum_{j=2}^K w_j \{ 1(a_j^1<a_j^2) \hat\bmu_{j0} + 1(a_j^1 \geq a_j^2) \hat\bmu_{j1} \}$ with $w_j=1/K$.

For the pooled estimator and the distributed EM algorithm  estimator, we use the same local estimator from site 1 as the initial estimator. % (results of some apparent outlying initial values are removed). 
Specifically, the local estimator is initialized by K-means clustering with five different initializations to ensure convergence to a global maximizer. To see the approximation performance of the distributed EM algorithm  estimator, we calculate the approximation error, measured by the relative distance $\|\widetilde\bmu^t - \bmu^t \|_2/\|\bmu^t\|_2$ of the distributed EM algorithm  estimator to the pooled estimator along the iteration path $\{\widetilde\bmu^t\}_{t \geq 1}$. We also compare the estimation error and the bias of the local estimator, average estimator, pooled estimator and distributed EM algorithm  estimator, e.g., for the pooled estimator we calculate $\|\widehat\bmu - \bmu^* \|_2/\sqrt{2d}$ and $\widehat\mu_{01} - \mu_{01}^*$. The simulation is replicated 200 times for each setting.

We first investigate the approximation error of the distributed EM algorithm  estimator relative to the pooled estimator along the iteration path. Since the approximation error decays rapidly to a small value at the early stage of the iteration path, we only show the approximation error over first 50 iterations.
%{\color{blue}Since the pooled estimator often has a smaller iteration number than the distributed EM algorithm  estimator}, for the parts {\color{blue}in iterations} where the distributed EM algorithm  has not converged {\color{blue}yet} {\color{red}\bf not converged?? I do not understand RJC} while the pooled estimator has achieved convergence, we keep using the converged pooled estimates as a reference to calculate the relative distance. 
Figure \ref{fig0} displays randomly selected examples from 200 replications under simulation settings with number of sites $K=10$. After 50 iterations, it is shown that the approximation errors for all four examples are smaller than $10^{-4}$. Under all settings considered, these paths confirm that when initialized with a local estimator, the relative distance starts from a small value and then decays to zero rapidly.

\begin{figure}
	\centering
	\includegraphics[width = 1\linewidth]{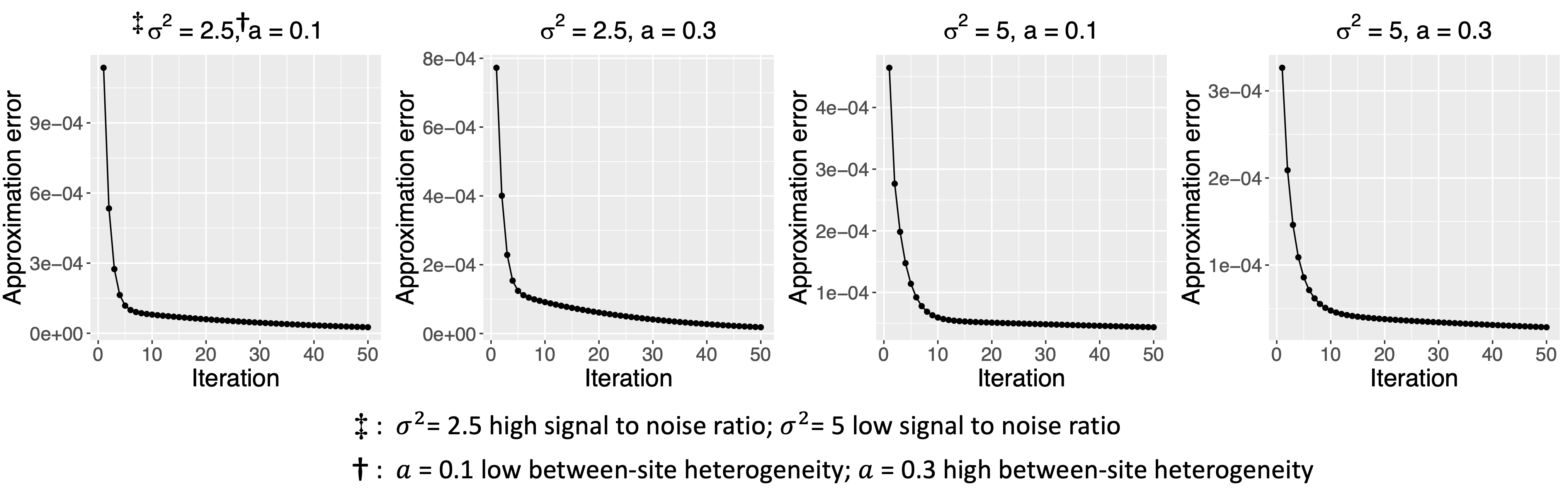}
	\caption{\baselineskip=12pt 
	Approximation error of the distributed EM algorithm estimator to the pooled estimator when $n$ = $1,000$ at the first 50 iterations of the EM algorithm under different settings of signal to noise ratio $\sigma^2$ and heterogeneity level $a$ with number of sites $K=10$. 
	}
\label{fig0}
\end{figure}

%\begin{figure}
%\figuresize{.2}
%\figurebox{20pc}{25pc}{}[appro_path_n1000_ite50_withSymbol.png]
%\caption{Approximation error of the distributed EM algorithm estimator to the pooled estimator when $n$ = $1,000$ at the first 50 iterations of the EM algorithm under different settings of signal to noise ratio $\sigma^2$ and heterogeneity level $a$ with number of sites $K=10$.}
%\label{fig0}
%\end{figure}

We now consider the relative performance in estimation of different methods. Figure \ref{fig1} presents the empirical bias and variances of estimates of $\mu_{01}$ when $n=1,000$. The results for $n=3,000$ are similar and are deferred to Supplementary Material S10. Overall, we found that for the average estimator, either a low signal to noise ratio (i.e., $\sigma^2=5$ as opposed to $\sigma^2=2.5$) or larger between-site heterogeneity (i.e., $a=0.3$ as opposed to $a=0.1$) led to increased bias and larger variance. A larger number of sites (i.e., $K=30$ as opposed to $K=10$), corresponding to a larger total sample size, led to smaller variance yet the larger bias in the average estimator remains. On the other hand, for both the pooled estimator and our distributed EM estimator, the bias under all settings was small, and, as expected, both estimators have similar bias and variance. Further, we found that, similar to the average estimator, a smaller signal to noise ratio led to larger variance in both estimators. However, different from the average estimator, larger between-site heterogeneity (i.e., $a=0.3$ as opposed to $a=0.1$) has little impact on the relative variance of the pooled estimator and the distributed EM estimator. %This is consistent with our theoretical results in Theorem 1--2. (?)

Figure \ref{fig3} presents the mean squared error of different estimates of the parameter $\bmu$ when $n=1,000$. The results for $n=3,000$ are similar and are deferred to Supplementary Material S10. Similar to the findings in Figure \ref{fig1}, for the average estimator, either the signal to noise ratio or between-site heterogeneity had a sizable impact on the mean squared error. The pooled estimator and the distributed EM estimator had a similar mean squared error, which is impacted by the signal to noise ratio, but is relatively robust to the level of between-site heterogeneity.

In summary, the simulation study confirmed that although the average method is simple to implement, its performance, in terms of estimation bias, variance and mean squared error, is sensitive to the signal to noise ratio, level of heterogeneity, and the number of sites.  It has a larger bias compared to the proposed estimator across all settings considered. On the other hand, the distributed EM algorithm estimator provides an excellent approximation to the pooled estimator with a small bias and nearly identical variance. It can successfully handle different levels of between-site heterogeneity, and its variance is robust to the level of between-site heterogeneity. %When there are more sites participating in a collaborative study or each site has a larger sample size, the estimation accuracy of the shared parameters is improved. 

\begin{figure}
	\centering
	\includegraphics[width = 0.9\linewidth]{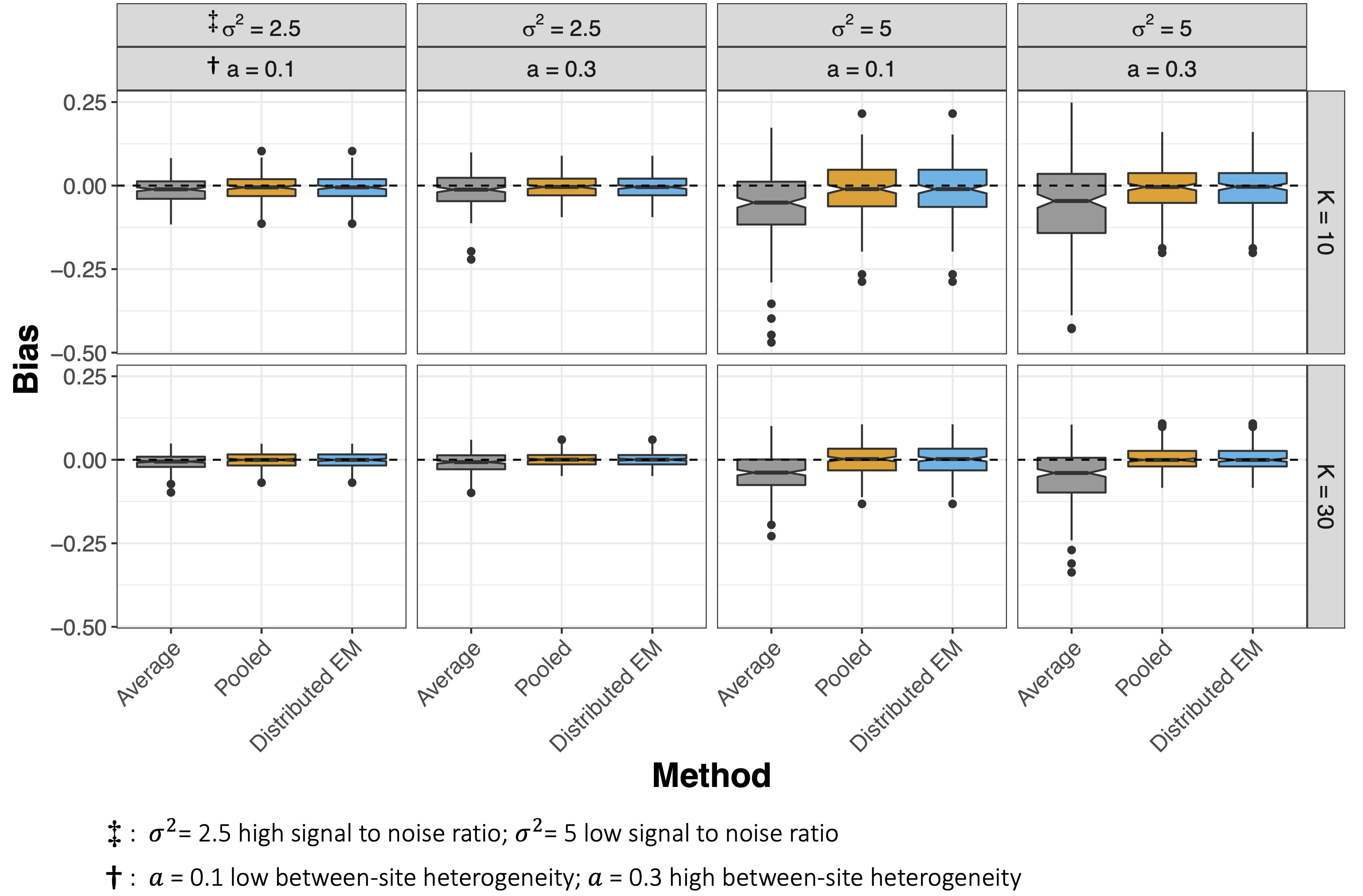}
%\figuresize{.24}
%\figurebox{20pc}{25pc}{}[BIAS_3_C2_mu01_n1000_withSymbol.png]
	\caption{\baselineskip=12pt Empirical bias and variances of estimates of $\mu_{01}$ from the average estimator, the pooled estimator, and our distributed EM estimator, when $n=1,000$ under different settings of number of sites ($K$), signal to noise ratio ($\sigma^2$) and heterogeneity level ($a$).}
\label{fig1}
\end{figure}

\begin{figure}
	\centering
	\includegraphics[width = 0.9\linewidth]{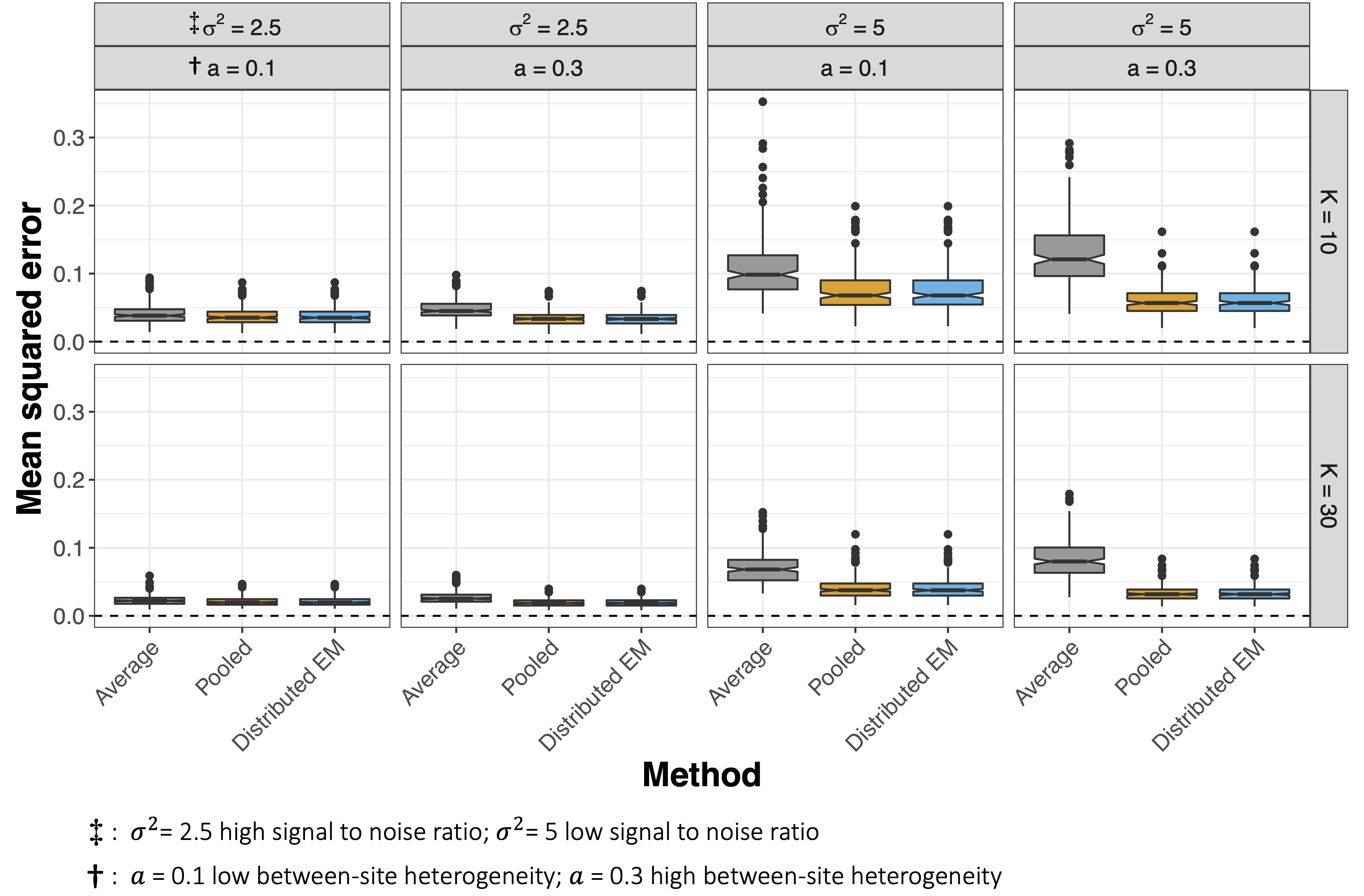}
%\figuresize{.24}
%\figurebox{20pc}{25pc}{}[MSE_3_C2_n1000_withSymbol.png]
	\caption{\baselineskip=12pt Mean squared error of estimates of $\bmu$ from the average estimator, the pooled estimator, and our distributed EM estimator, when $n=1,000$ under different settings of number of sites ($K$), signal to noise ratio ($\sigma^2$) and heterogeneity level ($a$).}
\label{fig3}
\end{figure}

\section{Discussion}\label{sec:dis}

%{\color{red}\bf The title of Section 6 is all in caps, but here not.}
We have developed a distributed learning framework for latent class models, which is distinct from most existing work that focuses on distributed supervised learning. Specifically, our approach investigates federated unsupervised learning and enables collaborative identification of shared latent classes across sites while allowing for heterogeneous proportions of latent classes. Our approach enables jointly fitting EM algorithms across multiple heterogeneous local data sets without sharing individual participant data, and this novel setting also makes our work distinguishable from traditional single site EM algorithm analysis. Additionally, our algorithm is derived from a novel construction of a surrogate Q function, which adopts a density ratio weighting approach to approximate the pooled population. We demonstrated that our proposed method achieves nearly identical performance to the pooled EM algorithm, both theoretically and numerically. %We further illustrated the validity and feasibility of implementing the  algorithm in large-scale multi-center studies for subphenotyping of a long COVID pediatric patient cohort. 
Overall, our work presents a novel approach to federated unsupervised learning in latent class models that can be applied to large-scale multi-site studies.

Similar to existing work on the EM algorithm \citep{wang2015high, %balakrishnan2017statistical,
cai2019chime}, our theoretical analysis is based on the case where we know the number of clusters, i.e., $S = 2$. Even for a single-site study, choosing the number of latent classes is a nonregular problem and has undergone extensive investigation. Specifically, in mixture models, testing for $S$ versus $S+1$ latent classes involves testing with the mixture proportion parameter lying on the boundary of the parameter space. What is more, under the null hypothesis, the relevant class-specific parameter is not identifiable \citep{davies1977hypothesis, davies1987hypothesis}. \citet{ning2015class} and \citet{hong2017plemt} have developed pseudolikelihood-based methods to test one versus two-class problems in exponential tilt mixture models. Extension of these methods to handle problems involving more than two classes  while accounting for between-site heterogeneity, data-sharing prohibition, and communication cost requires more investigation.

Our model specification assumes the same number of latent classes across sites. An interesting future direction is to consider the case where some sites may only contain a subset of the $S$ latent classes. Some recent theoretical work on over-specified class numbers might be helpful to understand the behavior of the distributed EM algorithm under this case \citep{dwivedi2020sharp,dwivedi2020singularity}.  In addition, when there are intrinsic differences in the latent class characterizations, i.e., each site may have site-specific parameters $\bmu$, it is also interesting to study whether the latent class characterizations are similar between sites such that data from one site can be used to refine the latent class analysis in the other sites. Some of these extensions are currently under investigation, and will be reported in the future.

\bibliographystyle{apalike}
\bibliography{LCM}

% ==================================================
% Supplementary Materials (AFTER references)
% ==================================================
\newpage
\appendix

\centerline{\Large\bf Supplementary Materials}
\vspace{0.3in}

% ---------- Reset counters and numbering ----------
\setcounter{section}{0}
\setcounter{theorem}{0}
\setcounter{figure}{0}
\setcounter{table}{0}

\renewcommand{\thesection}{S\arabic{section}}
\renewcommand{\thetheorem}{S\arabic{theorem}}
\renewcommand{\thelemma}{S\arabic{lemma}}
\renewcommand{\thefigure}{S\arabic{figure}}
\renewcommand{\thetable}{S\arabic{table}}

% ---------- Supplement theorem environments ----------
\theoremstyle{plain}
\newtheorem{theoremS}{Theorem}[section]
\newtheorem{lemmaS}[theoremS]{Lemma}
\newtheorem{propositionS}[theoremS]{Proposition}

\theoremstyle{remark}
\newtheorem{remarkS}[theoremS]{Remark}

\noindent
Note: we use $c_1,\ c_2,\ \ldots$ to represent positive constants whose exact values may change line to line. 

\section{Proof of equation (5)}
\begin{proof}%[Proof of equation (5)]
Recall the definition of the density ratio tilted surrogate Q function:
\begin{align*}
\widetilde Q\left(\bmu\mid \btheta^t\right) = \check Q\left(\bmu\mid\btheta^t\right) + \left\langle \nabla_\bmu Q_{\bmu}\left(\bmu^t\mid\btheta^t\right)- \nabla_\bmu \check Q\left(\bmu^t\mid\btheta^t\right), \bmu \right\rangle.
\end{align*}
Take the gradient of $\widetilde Q\left(\bmu\mid \btheta^t\right)$ at $\bmu^t$, we have
\begin{align*}
    \nabla_{\bmu} \widetilde Q\left(\bmu^t \mid \btheta^t\right) & = \nabla_{\bmu} \check Q\left(\bmu^t\mid\btheta^t\right) + \nabla_{\bmu} Q_{\bmu}(\bmu^t \mid\btheta^t) - \nabla_{\bmu} \check Q\left(\bmu^t\mid\btheta^t\right)\\
    & = \nabla_{\bmu} Q_{\bmu}(\bmu^t \mid\btheta^t)
    .
\end{align*}
Take the second- and higher-order derivatives of $\widetilde Q\left(\bmu\mid \btheta^t\right)$, we directly get 
\begin{align*}
\nabla_{\bmu}^p \widetilde Q\left(\bmu\mid\btheta^t\right) = \nabla_{\bmu}^p \check Q\left(\bmu\mid\btheta^t\right),\quad p \geq 2.
\end{align*}
This completes the proof.
\end{proof}

\section{Proof of Lemma \ref{lemma1} }

In the proof we ignore the logarithmic factors, e.g., $\log(n)$, when they are dominated by other terms. We use $c_1,\ c_2,\ \ldots$ to represent positive constants whose exact values change line to line. We need the following technical lemma.

\begin{lem1}
	Under Assumption \ref{assump:smooth}, % \ref{assump:smooth}, 
	we have for $2 \leq \kappa_1, \kappa_2, \kappa_3, \kappa_4 \leq 8$, %there exist constants $c_1$ and $c_2$ that only relate to the moments $\kappa_1$ and $\kappa_2$, respectively, such that
	\begin{align}
		& E[\| \frac{1}{nK} \sum_{j=1}^K\sum_{i=1}^n \{\nabla_{\bmu\bmu}^2 t(\y_{i1},\bEta_j^*)h(\y_{i1},\bmu^*,\btheta_j^*) - E\nabla_{\bmu\bmu}^2 h(\y_{ij},\bmu^*,\btheta_j^*)\} \|_2^{\kappa_1}] = O(n^{-\kappa_1/2}) %\leq c_1\frac{\log^{\kappa_1/2}(2d)H^{\kappa_1}}{n^{\kappa_1/2}},
		\label{s1_1} 
		\\
		& E[\| \frac{1}{nK} \sum_{j=1}^K\sum_{i=1}^n \{\nabla_{\bmu\bmu}^2 h(\y_{ij},\bmu^*,\btheta_j^*) - E\nabla_{\bmu\bmu}^2 h(\y_{ij},\bmu^*,\btheta_j^*)\} \|_2^{\kappa_2}] = O((nK)^{-\kappa_2/2}) %\leq c_2\frac{\log^{\kappa_2/2}(2d)H^{\kappa_2}}{(nK)^{\kappa_2/2}}.
		\label{s1_2}\\
		& E[\| \frac{1}{nK} \sum_{j=1}^K\sum_{i=1}^n \{\nabla_{\bmu\btheta_j}^2 h(\y_{ij},\bmu^*,\btheta_j^*) - E\nabla_{\bmu\btheta_j}^2 h(\y_{ij},\bmu^*,\btheta_j^*)\} \|_2^{\kappa_3}] = O((nK)^{-\kappa_3/2}) %\leq c_2\frac{\log^{\kappa_2/2}(2d)H^{\kappa_2}}{(nK)^{\kappa_2/2}}.
		\label{s1_3} \\
		& E[| \frac{1}{n} \sum_{i=1}^n m_k(\y_{ij}) - E m_k(\y_{ij})|^{\kappa_4}] = O(n^{-\kappa_4/2}), \ k=1,2,3,4, \ j\in [K]. 
		\label{s1_4}
	\end{align}
\end{lem1}

The left hand side of \eqref{s1_1} only involves independently and identically distributed (i.i.d.) samples $\{\y_{i1}\}_{i=1}^n$ from the leading site, therefore, \eqref{s1_1} is a direct application of Lemma 7 in \cite{zhang2013communication} based on Assumption \ref{assump:smooth}. % \ref{assump:smooth}. 
The same argument goes for \eqref{s1_4}.
As for \eqref{s1_2} and \eqref{s1_3}, since the observations $\{\y_{ij}\}$ are not i.i.d. across sites, the proof needs some minor modifications to the proof of Lemma 7 in \cite{zhang2013communication} to change from i.i.d. summands to independent but not identically distributed summands. These modifications are trivial and we ignore the proof here.    

\begin{lem1} Define $\bar\bmu^{t+1}$ as the root of $\nabla_{\bmu}Q_{\bmu} ( \bmu|\tilde\btheta^t)=\0$, given that $d_2( \tilde\btheta_j^t, \btheta_j^t )=O(n^{-5/6})$ with probability at least $1-K/n-n^{-2/3}$, then we have
\begin{align*}
	\| \bar\bmu^{t+1} - \bmu^{t+1} \|_2 = O(n^{-5/6})
\end{align*} 
	with probability at least $1-K/n-n^{-2/3}$.
\end{lem1}

\begin{proof}[Proof of Lemma \ref{lemma1}]
%The proposed method updates $\bmu$ and $\bLambda=(\lambda_1,\ldots,\lambda_K)$ from different objective functions. %, i.e., for the shared parameter $\bmu$ we integrate information from all sites to update it while for the site-specific $\lambda_j$ we only use the local data to estimate it. Thus, 
We first focus on the updating from $\btheta^t$ to $\btheta^{t+1}$ and $\tilde\btheta^{t+1}$. When updating from the same $\btheta^t$, we have $d_2(\btheta^{t+1}, \tilde\btheta^{t+1})=\sum_{k=0}^1 \| \bmu_k^{t+1} - \tilde\bmu_k^{t+1} \|_2$ since $ \Lambda^{t+1} = \tilde \Lambda^{t+1}$.
%$
%\|\widetilde \btheta - \widehat\btheta \|_2 = \|\widetilde \bmu - \bmu^{t+1}  \|_2 
%$ (for simplicity, we let $\widehat\btheta = M_n(\btheta^t)$ and $\widetilde\btheta = \widetilde M_n(\btheta^t)$)
%as $\widetilde\bLambda - \widehat\bLambda=\0$. 
Therefore, we only need to bound $\sum_{k=0}^1 \| \bmu_k^{t+1} - \tilde\bmu_k^{t+1} \|_2$, and the following three events are required
\begin{align*}
	%\epsilon_0 &:= \{ \frac{1}{nK} \sum_{i=1}^n \sum_{j=1}^K  m_j(\y_{ij}) \leq 2L\}, j=1,2 \\
	&\epsilon_{0j} := \{ \frac{1}{n} \sum_{i=1}^n  m_k(\y_{ij}) \leq 2L\}, k=1,2;\ j\in [K] \\
	&\epsilon_1 := \{ \| \nabla_{\bmu\bmu}^2 {\widetilde{Q}} ( \bmu^{t+1} | \btheta^t) - \nabla_{\bmu\bmu}^2 Q_{\bmu}(\bmu^{t+1} | \btheta^t)  \|_2 \leq C_1 \} \\
	&\epsilon_2 := \{ \| \nabla_{\bmu} {\widetilde{Q}} ( \bmu^{t+1} | \btheta^t) \|_2 \leq C_2 \},
\end{align*}
where %$E(m_k(\y_{ij}))<M$ for all $k=1,2$ {(for $k=1$ only require $\epsilon_{01}$ holds)}, $i=1,\ldots,n$ and $j=1,\ldots, K$, 
$C_1\leq \rho\mu_-/2$ and $C_2 \leq (1-\rho)\mu_-\delta_{\rho}/2$ with $\delta_{\rho}=\min\{\rho, \rho\mu_-/(4L)\}$. 
Then, by using Lemma 6 in \cite{zhang2013communication}, under $\epsilon_0 \cap \epsilon_1 \cap\epsilon_2 $ where $\epsilon_0= \cap_j \epsilon_{0j}$ we have $\|\tilde \bmu^{t+1} - \bmu^{t+1}  \|_2 \leq C \| \nabla_{\bmu} {\widetilde{Q}} ( \bmu^{t+1} | \btheta^t) \|_2$. Since $\|\tilde \bmu_k^{t+1} - \bmu_k^{t+1}  \|_2 \leq \|\tilde \bmu^{t+1} - \bmu^{t+1}  \|_2$, next we only need to control $\| \nabla_{\bmu} {\widetilde{Q}} ( \bmu^{t+1} | \btheta^t) \|_2$. 

%Recall that for each $m \in [K]$, we write
%	\begin{align*}
%		h(\y;\btheta_j,\btheta_j') = w_{\btheta_j'}^m (\y)[\log\{f(\y; \bmu_1)\}+ & \log(\lambda_m)] + \\ & \{1-w_{\btheta_j'}^m (\y)\}[\log\{f(\y; \bmu_0)\}+\log(1-\lambda_m)], 
%	\end{align*}
%	where $w_{\btheta_j}^m (\y) = E(Z|\btheta_j, \y)$. Also, the density ratio tilting component between site $m$ and site $1$ with observation $\y_{i1}$ is written as $t(\y_{i1},\bEta_j^t)$. 
Due to the fact that $\nabla_{\bmu} Q_{\bmu} ( \bmu^{t+1} | \btheta^t)= \0 $, with $\bmu'=\alpha\bmu^t + (1-\alpha)\bmu^{t+1}$ where $\alpha\in(0,1)$ we have
\begin{align}
	\nabla_{\bmu} & {\widetilde{Q}} (\bmu^{t+1}|\btheta^t) = \nabla_{\bmu}{\check{Q}} (\bmu^{t+1}|\btheta^t) +\nabla_{\bmu} Q_{\bmu}(\bmu^t|\btheta^t) -\nabla_{\bmu}{\check{Q}} (\bmu^t|\btheta^t) - \nabla_{\bmu} Q_{\bmu} ( \bmu^{t+1} | \btheta^t) \nonumber \\
	& = \frac{1}{Kn}\sum_{i=1}^n\sum_{j=1}^K t(\y_{i1},\bEta_j^t)\{ \nabla_{\bmu} h(\y_{i1},\bmu^{t+1},\btheta_j^t) - \nabla_{\bmu} h(\y_{i1}, \bmu^t, \btheta_j^t) \} \nonumber \\
	& + \frac{1}{Kn}\sum_{i=1}^n\sum_{j=1}^K \{ \nabla_{\bmu} h(\y_{ij},\bmu^t,\btheta_j^t) - \nabla_{\bmu} h(\y_{ij}, \bmu^{t+1}, \btheta_j^t) \}\nonumber \\
	& \leq \frac{1}{Kn}\sum_{i=1}^n\sum_{j=1}^K t(\y_{i1},\bEta_j^t)\{ \nabla_{\bmu\bmu}^2h(\y_{i1},\bmu',\btheta_j^t)(\bmu^{t+1} - \bmu^t) \} \nonumber\\
	& -  \frac{1}{Kn}\sum_{i=1}^n\sum_{j=1}^K \{ \nabla_{\bmu\bmu}^2h(\y_{ij},\bmu',\btheta_j^t)(\bmu^{t+1} - \bmu^t) \}\nonumber \\
%	\end{align}
%	
%	\begin{align}
	& \leq \frac{1}{Kn}\sum_{i=1}^n\sum_{j=1}^K  \nabla_{\bmu\bmu}^2\{t(\y_{i1},\bEta_j^t) h(\y_{i1},\bmu',\btheta_j^t) - t(\y_{i1},\bEta^*_j) h(\y_{i1},\bmu^*,\btheta_j^*)\}(\bmu^{t+1} - \bmu^t) \label{1st_1} \\
	& -  \frac{1}{Kn}\sum_{i=1}^n\sum_{j=1}^K \nabla_{\bmu\bmu}^2\{  h(\y_{ij},\bmu',\btheta_j^t) - h(\y_{ij},\bmu^*,\btheta_j^*) \}(\bmu^{t+1} - \bmu^t) \label{1st_2} \\
	&  + \frac{1}{Kn}\sum_{i=1}^n\sum_{j=1}^K \{\nabla_{\bmu\bmu}^2t(\y_{i1},\bEta^*_j) h(\y_{i1},\bmu^*,\btheta_j^*) - \nabla_{\bmu\bmu}^2 h(\y_{ij},\bmu*,\btheta_j^*) \}(\bmu^{t+1} - \bmu^t). \label{1st_3}
\end{align}
As for \eqref{1st_1}, based on $\epsilon_{01}$ and Assumption \ref{assump:smooth} 
we have
\begin{align*}
	\| \frac{1}{Kn}\sum_{i=1}^n\sum_{j=1}^K  & \nabla_{\bmu\bmu}^2\{t(\y_{i1},\bEta_j^t) h(\y_{i1},\bmu',\btheta_j^t) - t(\y_{i1},\bEta^*_j) h(\y_{i1},\bmu^*,\btheta_j^*)\}(\bmu^{t+1} - \bmu^t) \|_2 \\
	& \leq 2L(\|\bmu' - \bmu^* \|_2 + \frac{1}{K} \sum_{j=1}^K \|\bEta_j^t - \bEta_j^*\|_2)\| \bmu^{t+1} - \bmu^t \|_2.
\end{align*}
Similarly, based on $\epsilon_{0}$ and Assumption \ref{assump:smooth} 
we can bound \eqref{1st_2} with
\begin{align*}
	\| \frac{1}{Kn}\sum_{i=1}^n\sum_{j=1}^K & \nabla_{\bmu\bmu}^2\{  h(\y_{ij},\bmu',\btheta_j^t) - h(\y_{ij},\bmu^*,\btheta_j^*) \}(\bmu^{t+1} - \bmu^t) \|_2 \\
	& \leq 2L(\|\bmu' - \bmu^* \|_2 + \frac{1}{K} \sum_{j=1}^K \|\btheta_j^t - \btheta_j^*\|_2)\| \bmu^{t+1} - \bmu^t \|_2.
\end{align*}
As for \eqref{1st_3}, we can further decompose it  as
\begin{align*}
	\frac{1}{Kn}\sum_{i=1}^n\sum_{j=1}^K & \{\nabla_{\bmu\bmu}^2t(\y_{i1},\bEta^*_j) h(\y_{i1},\bmu^*,\btheta_j^*) - \nabla_{\bmu\bmu}^2 h(\y_{ij},\bmu^*,\btheta_j^*) \}  \\
	& = \frac{1}{Kn}\sum_{i=1}^n\sum_{j=1}^K \{\nabla_{\bmu\bmu}^2t(\y_{i1},\bEta^*_j) h(\y_{i1},\bmu^*,\btheta_j^*) -E\nabla_{\bmu\bmu}^2  h(\y_{ij},\bmu^*,\btheta_j^*) \\ & +  E\nabla_{\bmu\bmu}^2 h(\y_{ij},\bmu^*,\btheta_j^*)  - \nabla_{\bmu\bmu}^2 h(\y_{ij},\bmu^*,\btheta_j^*) \},
\end{align*}
which leads to 
 \begin{align*}
	\|\frac{1}{Kn}\sum_{i=1}^n\sum_{j=1}^K & \{\nabla_{\bmu\bmu}^2t(\y_{i1},\bEta^*_j) h(\y_{i1},\bmu^*,\btheta_j^*) - \nabla_{\bmu\bmu}^2 h(\y_{ij},\bmu^*,\btheta_j^*) \}\|_2  \\
	& \leq \|\frac{1}{Kn}\sum_{i=1}^n\sum_{j=1}^K \{\nabla_{\bmu\bmu}^2t(\y_{i1},\bEta^*_j) h(\y_{i1},\bmu^*,\btheta_j^*) -E\nabla_{\bmu\bmu}^2 h(\y_{ij},\bmu^*,\btheta_j^*)\|_2 \\ & +  \|\frac{1}{Kn}\sum_{i=1}^n\sum_{j=1}^K \{E\nabla_{\bmu\bmu}^2 h(\y_{ij},\bmu^*,\btheta_j^*)  - \nabla_{\bmu\bmu}^2 h(\y_{ij},\bmu^*,\btheta_j^*) \}\|_2.
\end{align*}
Note that
$$E(\nabla_{\bmu\bmu}^2t(\y_{i1},\bEta^*_j) h(\y_{i1},\bmu^*,\btheta_j^*) -E\nabla_{\bmu\bmu}^2 h(\y_{ij},\bmu^*,\btheta_j^*))=\0$$ and 
$$E(\nabla_{\bmu\bmu}^2 h(\y_{ij},\bmu^*,\btheta_j^*)  - E\nabla_{\bmu\bmu}^2 h(\y_{ij},\bmu^*,\btheta_j^*))=\0, $$
then based on Assumption \ref{assump:smooth} 
and Lemma S.1, with probability at least $1-n^{-2/3}$, we have 
\begin{align*}
	\|\frac{1}{Kn}\sum_{i=1}^n\sum_{j=1}^K \{\nabla_{\bmu\bmu}^2t(\y_{i1},\bEta^*_j) h(\y_{i1},\bmu^*,\btheta_j^*) -E\nabla_{\bmu\bmu}^2 h(\y_{ij},\bmu^*,\btheta_j^*)\|_2 \lesssim n^{-1/3}
\end{align*}
and with probability at least $1-(nK)^{-2/3}$, we have
\begin{align*}
	\|\frac{1}{Kn}\sum_{i=1}^n\sum_{j=1}^K \{E\nabla_{\bmu\bmu}^2 h(\y_{ij},\bmu^*,\btheta_j^*)  - \nabla_{\bmu\bmu}^2 h(\y_{ij},\bmu^*,\btheta_j^*) \}\|_2 = O((nK)^{-1/3}).
\end{align*} 
Therefore, with probability at least $1-n^{-2/3}-(nK)^{-2/3}$, we have 
$$ \|\frac{1}{Kn}\sum_{i=1}^n\sum_{j=1}^K \{\nabla_{\bmu\bmu}^2t(\y_{i1},\bEta^*_j) h(\y_{i1},\bmu^*,\btheta_j^*) - \nabla_{\bmu\bmu}^2 h(\y_{ij},\bmu^*,\btheta_j^*) \}\|_2 = O(n^{-1/3}).$$
To summarize, with $c_1$ and $c_2$ being some constants, we have 
\begin{align*}
	\|\nabla_{\bmu}{\widetilde{Q}} (\bmu^{t+1}|\btheta^t)\|_2 \leq \{c_1\|\bmu'-\bmu^*\|_2 +\frac{c_2}{K}\sum_{j=1}^K\|\bEta_j^t - \bEta_j^*\|_2 + O_p(n^{-1/3})\}\|\bmu^t-\bmu^{t+1}\|_2.
\end{align*}

Next, let's control $P(\epsilon^c)=P((\epsilon_0 \cap \epsilon_1 \cap \epsilon_2)^c)$ where $\epsilon_0= \cap_j \epsilon_{0j}$. 
By Assumption \ref{assump:smooth} and Lemma S.1, %Proposition 5.10 in Vershynin \cite{vershynin2010introduction}, 
we have $P(\epsilon_{0j}^c)\lesssim 1/n$ which leads to $P(\epsilon_0^c)\lesssim K/n$. To control $P(\epsilon_1^c)$, follow the same steps as we bound $\|\nabla_{\bmu}{\widetilde{Q}} (\bmu^{t+1}|\btheta^t)\|_2$, we have
\begin{align*}
	\| \nabla_{\bmu\bmu}^2 & {\widetilde{Q}} ( \bmu^{t+1} | \btheta^t) - \nabla_{\bmu\bmu}^2 Q_{\bmu}(\bmu^{t+1} | \btheta^t)\|_2 \leq c_1\|\bmu' - \bmu^*\|_2 + \frac{c_2}{K}\sum_{j=1}^K\|\bEta_j^t-\bEta_j^*\|_2 \\
	& + \|\frac{1}{Kn}\sum_{i=1}^n\sum_{j=1}^K \{\nabla_{\bmu\bmu}^2t(\y_{i1},\bEta^*_j) h(\y_{i1},\bmu^*,\btheta_j^*) - \nabla_{\bmu\bmu}^2h(\y_{ij},\bmu^*,\btheta_j^*) \}\|_2.
\end{align*}
Since $ \| \bmu' - \bmu^* \|_2 \leq \| \bmu^t - \bmu^* \|_2 + \| \bmu^{t+1} - \bmu^* \|_2 $ and $ \| \bEta_j^t - \bEta_j^* \|_2 \leq d_2(\btheta_1^t, \btheta_1^*) + d_2(\btheta_j^t, \btheta_j^*) $, based on Assumption \ref{assump:init} we have 
$E(\| \bmu' - \bmu^* \|^2_2) = O(K/n)$ and $E(\| \bEta_j^t - \bEta_j^* \|^2_2) = O(K/n)$. Therefore, we obtain
\begin{align*}
	& P(\| \nabla_{\bmu\bmu}^2 {\widetilde{Q}} ( \bmu^{t+1} | \btheta^t) - \nabla_{\bmu\bmu}^2 Q_{\bmu}(\bmu^{t+1} | \btheta^t)\|_2 > C_1 ) \\
	& \leq P( c_1\|\bmu' - \bmu^*\|_2> C_1/3)  + P(\frac{c_2}{K}\sum_{j=1}^K\|\bEta_j^t - \bEta_j^*\|_2 > C_1/3 ) \\
	& + P(\|\frac{1}{Kn}\sum_{i=1}^n\sum_{j=1}^K  \{\nabla_{\bmu\bmu}^2t(\y_{i1},\bEta^*_j) h(\y_{i1},\bmu^*,\btheta_j^*) - \nabla_{\bmu\bmu}h(\y_{ij},\bmu^*,\btheta_j^*) \}\|_2 >C_1/3)\\
& \leq c_1 E\|\bmu' - \bmu^*\|_2^2 + E(\frac{c_2}{K}\sum_{j=1}^K\|\bEta_j^t - \bEta_j^*\|_2)^2    \\
& + c_3 E\|\frac{1}{Kn}\sum_{i=1}^n\sum_{j=1}^K  \{\nabla_{\bmu\bmu}^2t(\y_{i1},\bEta^*_j) h(\y_{i1},\bmu^*,\btheta_j^*) - \nabla_{\bmu\bmu}h(\y_{ij},\bmu^*,\btheta_j^*) \}\|^2_2  \lesssim \frac{K}{n} 
\end{align*}
where in the second inequality we use the Markov inequality, and in the last line we use Lemma S.1, Jensen's inequality and Assumption \ref{assump:init}. 

As for $P(\epsilon_2^c)$, with $$ \| \bmu' - \bmu^* \|_2 \leq \| \bmu^t - \bmu^* \|_2 + \| \bmu^{t+1} - \bmu^* \|_2 $$ and 
$$
{\|\bmu^t-\bmu^{t+1}\|_2 \leq \|\bmu^{t+1}-\bmu^*\|_2 + \|\bmu^t - \bmu^*\|_2}, %\leq \|\bmu^{t+1} - \bmu^*\|_2 + \frac{1}{K}\sum_{j=1}^K \| \bEta_j - \bEta^*_j \|_2 
$$
we have 
\begin{align*}
	& P( \| \nabla_{\bmu} {\widetilde{Q}} ( \bmu^{t+1} | \btheta^t) \|_2 > C_2 ) \\
	& \leq P( [ c_1\|\bmu' - \bmu^*\|_2 + \frac{c_2}{K}\sum_{j=1}^K\|\bEta_j^t - \bEta_j^*\|_2 + \|\frac{1}{Kn}\sum_{i=1}^n\sum_{j=1}^K \{\nabla_{\bmu\bmu}^2t(\y_{i1},\bEta^*_j) h(\y_{i1},\bmu^*,\btheta_j^*) - \\ 
	&  \nabla_{\bmu\bmu}^2h(\y_{ij},\bmu^*,\btheta_j^*) \}\|_2 ] \|\bmu^{t+1}-\bmu^t\|_2 > C_2 ) \\ 
	& \lesssim [E\|\bmu' - \bmu^*\|^2_2 + \frac{1}{K}\sum_{j=1}^K E \|\bEta_j^t - \bEta_j^*\|^2_2 + E\|\frac{1}{Kn}\sum_{i=1}^n\sum_{j=1}^K  \{\nabla_{\bmu\bmu}^2t(\y_{i1},\bEta^*_j) h(\y_{i1},\bmu^*,\btheta_j^*) - \\
	&  \nabla_{\bmu\bmu}^2h(\y_{ij},\bmu^*,\btheta_j^*) \}\|_2^2]^{1/2} \cdot[E \|\bmu^{t+1}-\bmu^*\|^2_2 + E\|\bmu^t - \bmu^*\|^2_2 ]^{1/2}  \lesssim \frac{K}{n}
\end{align*}
by the Markov inequality, Holder inequality, Lemma S.1 and Assumption \ref{assump:init}. 
Therefore, by combining the above results, with probability at least $1-Kn^{-1}-n^{-2/3}$ we have 
\begin{align*}
	d_2(\tilde \btheta^{t+1}, \btheta^{t+1} )& = \sum_{k=0}^1 \|\tilde \bmu_k^{t+1} - \bmu_k^{t+1} \|_2 \\
	&\leq [c_1\|\bmu^{t+1} - \bmu^*\|_2+\frac{c_2}{K}\sum_{j=1}^K\|\bEta_j^t - \bEta_j^*\|_2 + O(n^{-1/3})]\|\bmu^t - \bmu^{t+1}\|_2\\
	& \lesssim n^{-5/6}%n^{-1/3}h(n,K) 
\end{align*}
by Assumption \ref{assump:init}. By letting $t=0$, we have $d_2(\tilde \btheta_j^{1}, \btheta_j^{1} ) \lesssim n^{-5/6}$. % and we denote it by $P(n,K)=n^{-1/3}h(n,K)$ as the approximation error. Therefore, we have $d_2(\tilde \btheta_j^{1}, \btheta_j^{1} ) = O_p( n^{-1/3}h(n,K))$.

Next, we prove that the distance between $\tilde \btheta_j^{t+1}$ and $\btheta_j^{t+1}$ is also $O(n^{-5/6})$ for $t>0$, in which case the pooled EM algorithm and the distributed EM algorithm update with different estimates. Specifically, at the $t$-th iteration, the pooled EM algorithm updates $\btheta^t$ to $\btheta^{t+1}$ and the distributed EM algorithm updates $\tilde\btheta^t$ to $\tilde\btheta^{t+1}$. Suppose we have $d_2(\btheta_j^t,\tilde\btheta_j^t)=O(n^{-5/6})$. Recall that $d_2(\btheta_j^{t+1},\tilde\btheta_j^{t+1})= | \lambda_j^{t+1} - \tilde\lambda_j^{t+1} | + \sum_{k=0}^1 \| \bmu_k^{t+1} - \tilde\bmu_k^{t+1} \|_2 $, and we first control $\sum_{k=0}^1 \| \bmu_k^{t+1} - \tilde\bmu_k^{t+1} \|_2$ using $\sqrt{2}\| \bmu^{t+1} - \tilde\bmu^{t+1} \|_2$.  

Again, let's assume the following three events.
\begin{align*}
	&\epsilon_{0j} := \{ \frac{1}{n} \sum_{i=1}^n  m_k(\y_{ij}) \leq 2L\}, k=1,2,3,4;\ j\in [K] \\
	&\epsilon_1 := \{ \| \nabla_{\bmu\bmu}^2 {\widetilde{Q}} ( \bmu^{t+1} | \tilde\btheta^t) - \nabla_{\bmu\bmu}^2 Q_{\bmu}(\bmu^{t+1} | \btheta^t)  \|_2 \leq C_1 \} \\
	&\epsilon_2 := \{ \| \nabla_{\bmu} {\widetilde{Q}} ( \bmu^{t+1} | \tilde\btheta^t) \|_2 \leq C_2 \},
\end{align*}
where %$E(m_k(\y_{ij}))<M$ for all $k=1,2,3,4$, $i=1,\ldots,n$ and $j=1,\ldots, K$, 
$C_1\leq \rho\mu_-/2$ and $C_2 \leq (1-\rho)\mu_-\delta_{\rho}/2$ with $\delta_{\rho}=\min\{\rho, \rho\mu_-/(4L)\}$. 
Then, by using Lemma 6 in \cite{zhang2013communication}, under $\epsilon_0 \cap \epsilon_1 \cap\epsilon_2 $ where $\epsilon_0= \cap_j \epsilon_{0j}$ we have $\|\tilde \bmu^{t+1} - \bmu^{t+1}  \|_2 \leq C \| \nabla_{\bmu} {\widetilde{Q}} ( \bmu^{t+1} | \tilde\btheta^t) \|_2$. %Since $\|\tilde \bmu_k^{t+1} - \bmu_k^{t+1}  \|_2 \leq \|\tilde \bmu^{t+1} - \bmu^{t+1}  \|_2$, 
Next we control $\| \nabla_{\bmu} {\widetilde{Q}} ( \bmu^{t+1} | \tilde\btheta^t) \|_2$. 

Let's define $\bar \bmu^{t+1}$ as the root of $\nabla_{\bmu} Q_{\bmu} ( \bmu|\tilde\btheta^t)=\0$. Then with $\bmu'=\alpha_1\tilde\bmu^t + (1-\alpha_1)\bmu^{t+1}$ and $\bmu''=\alpha_2\tilde\bmu^t + (1-\alpha_2)\bar\bmu^{t+1}$ where $\alpha_1, \ \alpha_2\in(0,1)$, we have
\begin{align}
	\nabla_{\bmu} & {\widetilde{Q}} (\bmu^{t+1}|\tilde\btheta^t) = \nabla_{\bmu}{\check{Q}} (\bmu^{t+1}|\tilde\btheta^t) +\nabla_{\bmu} Q_{\bmu} (\tilde\bmu^t|\tilde\btheta^t) -\nabla_{\bmu}{\check{Q}} (\tilde\bmu^t|\tilde\btheta^t) - \nabla_{\bmu} Q_{\bmu} ( \bar\bmu^{t+1} | \tilde\btheta^t) \nonumber \\
	& = \frac{1}{Kn}\sum_{i=1}^n\sum_{j=1}^K t(\y_{i1},\tilde\bEta_j^t)\{ \nabla_{\bmu} h(\y_{i1},\bmu^{t+1},\tilde\btheta_j^t) - \nabla_{\bmu} h(\y_{i1}, \tilde\bmu^t, \tilde\btheta_j^t) \} \nonumber \\
	& + \frac{1}{Kn}\sum_{i=1}^n\sum_{j=1}^K \{ \nabla_{\bmu} h(\y_{ij},\tilde\bmu^t,\tilde\btheta_j^t) - \nabla_{\bmu} h(\y_{ij}, \bar\bmu^{t+1}, \tilde\btheta_j^t) \}\nonumber \\
	& \leq \frac{1}{Kn}\sum_{i=1}^n\sum_{j=1}^K \{ \nabla_{\bmu\bmu}^2  t(\y_{i1},\tilde\bEta_j^t) h(\y_{i1},\bmu',\tilde\btheta_j^t) - \nabla_{\bmu\bmu}^2t(\y_{i1},\bEta^*_j) h(\y_{i1},\bmu^*,\btheta_j^*)  \nonumber\\ 
	& + \nabla_{\bmu\bmu}^2t(\y_{i1},\bEta^*_j) h(\y_{i1},\bmu^*,\btheta_j^*)  \} (\bmu^{t+1} - \tilde\bmu^t) -  \frac{1}{Kn}\sum_{i=1}^n\sum_{j=1}^K \{ \nabla_{\bmu\bmu}^2h(\y_{ij},\bmu'',\tilde\btheta_j^t) - \nabla_{\bmu\bmu}^2 h(\y_{ij},\bmu^*,\btheta_j^*) \nonumber\\ 
	&+ \nabla_{\bmu\bmu}^2 h(\y_{ij},\bmu^*,\btheta_j^*) - E \nabla_{\bmu\bmu}^2 h(\y_{ij},\bmu^*,\btheta_j^*) + E \nabla_{\bmu\bmu}^2 h(\y_{ij},\bmu^*,\btheta_j^*) \} (\bar\bmu^{t+1} - \tilde\bmu^t).\nonumber
%	& \leq \frac{1}{Kn}\sum_{i=1}^n\sum_{j=1}^K  \nabla_{\bmu\bmu}^2\{t(\y_{i1},\bEta_j^t) h(\y_{i1},\bmu',\btheta_j^t) - t(\y_{i1},\bEta^*_j) h(\y_{i1},\bmu^*,\btheta_j^*)\}(\bmu^{t+1} - \bmu^t) \label{1st_1} \\
%	& -  \frac{1}{Kn}\sum_{i=1}^n\sum_{j=1}^K \nabla_{\bmu\bmu}^2\{  h(\y_{ij},\bmu',\btheta_j^t) - h(\y_{ij},\bmu^*,\btheta_j^*) \}(\bmu^{t+1} - \bmu^t) \label{1st_2} \\
%	&  + \frac{1}{Kn}\sum_{i=1}^n\sum_{j=1}^K \{\nabla_{\bmu\bmu}^2t(\y_{i1},\bEta^*_j) h(\y_{i1},\bmu^*,\btheta_j^*) - \nabla_{\bmu\bmu}^2 h(\y_{ij},\bmu*,\btheta_j^*) \}(\bmu^{t+1} - \bmu^t). \label{1st_3}
\end{align}
Then, using the same way we used before, based on $\epsilon_{0}$ and Assumption \ref{assump:smooth} 
we have
\begin{align}
	 \| \nabla_{\bmu} & {\widetilde{Q}} ( \bmu^{t+1} | \tilde\btheta^t) \|_2 \nonumber\\ 
	& \leq 2L(\|\bmu' - \bmu^* \|_2 + \frac{1}{K} \sum_{j=1}^K \|\tilde\bEta_j^t - \bEta_j^*\|_2)\| \bmu^{t+1} - \tilde\bmu^t \|_2 \label{mu1}\\ 
	& + 2M(\|\bmu'' - \bmu^* \|_2 + \frac{1}{K} \sum_{j=1}^K \|\tilde\btheta_j^t - \btheta_j^*\|_2)\| \bar\bmu^{t+1} - \tilde\bmu^t \|_2 \label{mu2}\\ 
	& + \|\frac{1}{Kn}\sum_{i=1}^n\sum_{j=1}^K \{\nabla_{\bmu\bmu}^2t(\y_{i1},\bEta^*_j) h(\y_{i1},\bmu^*,\btheta_j^*) - \nabla_{\bmu\bmu}^2 h(\y_{ij},\bmu^*,\btheta_j^*) \}\|_2\|\bmu^{t+1} - \tilde\bmu^t\|_2 \label{mu3}\\ 
	& + \|\frac{1}{Kn}\sum_{i=1}^n\sum_{j=1}^K \{\nabla_{\bmu\bmu}^2 h(\y_{ij},\bmu^*,\btheta_j^*) - E\nabla_{\bmu\bmu}^2 h(\y_{ij},\bmu^*,\btheta_j^*) \}\|_2\|\bar\bmu^{t+1} - \bmu^{t+1}\|_2 \label{mu4}\\ 
	& + \|\frac{1}{Kn}\sum_{i=1}^n\sum_{j=1}^K E\nabla_{\bmu\bmu}^2 h(\y_{ij},\bmu^*,\btheta_j^*)\|_2\|\bar\bmu^{t+1} - \bmu^{t+1}\|_2. \label{mu5}
\end{align}
As for \eqref{mu4}, based on Assumption \ref{assump:smooth} 
and Lemma S.1, with probability at least $1-(nK)^{-2/3}$, we have
\begin{align*}
	\|\frac{1}{Kn}\sum_{i=1}^n\sum_{j=1}^K \{E\nabla_{\bmu\bmu}^2 h(\y_{ij},\bmu^*,\btheta_j^*)  - \nabla_{\bmu\bmu}^2 h(\y_{ij},\bmu^*,\btheta_j^*) \}\|_2 = O((nK)^{-1/3}).
\end{align*} 
As for \eqref{mu3}, similarly, with probability at least $1-n^{-2/3}-(nK)^{-2/3}$, we have 
$$ \|\frac{1}{Kn}\sum_{i=1}^n\sum_{j=1}^K \{\nabla_{\bmu\bmu}^2t(\y_{i1},\bEta^*_j) h(\y_{i1},\bmu^*,\btheta_j^*) - \nabla_{\bmu\bmu}^2 h(\y_{ij},\bmu^*,\btheta_j^*) \}\|_2 = O(n^{-1/3}).$$
Also, by Assumption \ref{loc_conc} we can control $\|\frac{1}{Kn}\sum_{i=1}^n\sum_{j=1}^K E\nabla_{\bmu\bmu}^2 h(\y_{ij},\bmu*,\btheta_j^*)\|_2$ by a constant. Therefore, what remains is to control the terms $\|\bmu' - \bmu^* \|_2$, $\|\tilde\bEta_j^t - \bEta_j^* \|_2$, $\|\bmu^{t+1} -\tilde \bmu^t \|_2$, $\|\bmu'' - \bmu^* \|_2$, $\|\tilde\btheta_j^t - \btheta_j^* \|_2$, $\|\bar\bmu^{t+1} -\tilde \bmu^t \|_2$ and $\|\bar\bmu^{t+1} - \bmu^{t+1} \|_2$. 

We have the following relationship:
\begin{align*}
	\|\bmu' - \bmu^* \|_2 & \leq \|\tilde\bmu^t - \bmu^* \|_2 + \|\bmu^{t+1} - \bmu^* \|_2 \\
	\|\bmu'' - \bmu^* \|_2 & \leq \|\tilde\bmu^t - \bmu^* \|_2 + \|\bar\bmu^{t+1} - \bmu^{t+1} \|_2 + \|\bmu^{t+1} - \bmu^* \|_2 \\
	\|\bmu^{t+1} - \tilde\bmu^t \|_2 & \leq \|\tilde\bmu^t - \bmu^* \|_2 + \|\bmu^{t+1} - \bmu^* \|_2 \\
	\|\bar\bmu^{t+1} - \tilde\bmu^t \|_2 & \leq \|\bmu^{t+1} - \bar\bmu^{t+1} \|_2 + \|\bmu^{t+1} - \bmu^* \|_2 + \|\bmu^* - \tilde\bmu^t \|_2.
\end{align*}
Thus, we only need to control the terms on the right hand side of the above formula, i.e., $\|\tilde\bmu^t - \bmu^* \|_2$, $\|\bmu^{t+1} - \bmu^* \|_2$, and $\|\bar\bmu^{t+1} - \bmu^{t+1} \|_2$. % and $\|\bmu^{t+1} - \tilde\bmu^t \|_2$. 

For $\|\tilde\bmu^t - \bmu^* \|_2$, we have 
\begin{align*}
	\|\tilde\bmu^t - \bmu^* \|_2  & \leq \|\tilde\bmu^t - \bmu^t \|_2 + \|\bmu^t - \bmu^* \|_2  \\
	& \leq d_2(\btheta^t_j, \tilde\btheta_j^t) + d_2(\btheta^t_j, \btheta^*_j) \\
	& = O_p(n^{-5/6}) + O_p(n^{-1/2}),
\end{align*}
since $d_2(\btheta^t_j, \tilde\btheta_j^t) = O(n^{-5/6})$ with probability $1-K/n-n^{-2/3}$ and $d_2(\btheta^t_j, \btheta^*_j) = O(n^{-1/2}) $ with probability $1-K/n-(nK)^{-1}$ by Assumption \ref{assump:init}. It implies 
\begin{align*}
	E\|\tilde\bmu^t - \bmu^* \|_2^2 & \leq Ed_2^2(\btheta^t_j, \tilde\btheta_j^t) + Ed_2^2(\btheta^t_j, \btheta^*_j) \\
	& = O(K/n)+O(n^{-2/3}).
\end{align*}

For $\|\bmu^{t+1} - \bmu^* \|_2$, by Assumption \ref{assump:init} we have with probability $1-K/n-(nK)^{-1}$
\begin{align*}
	\|\bmu^{t+1} - \bmu^* \|_2 \leq d_2(\btheta^{t+1}_j, \btheta^*_j) = O(n^{-1/2})
\end{align*}
and it implies $E\| \bmu^{t+1} - \bmu^* \|_2^2 = O(K/n)$. 

As for $\|\bar\bmu^{t+1} - \bmu^{t+1} \|_2$, by {Lemma S.2}, we have $\|\bar\bmu^{t+1} - \bmu^{t+1} \|_2 = O(n^{-5/6})$ with probability at least $1-K/n-n^{-2/3}$ and it leads to $E(\|\bar\bmu^{t+1} - \bmu^{t+1} \|_2^2)=O(K/n) + O(n^{-2/3})$.
%For $\|\bar\bmu^{t+1} - \bmu^* \|_2$, since $\btheta^t$ is always in the contraction region based on Assumption \ref{assump:init}, we also have $\tilde\btheta^t$ is in the contraction region once $n$ is large to make $d_2(\btheta_j^t,\tilde\btheta_j^t)=O(n^{-5/6})$ be small. Therefore, by Assumption \ref{assump:init} we have %with probability $1-K/n-(nK)^{-1}$
%\begin{align*}
%	\|\bar\bmu^{t+1} - \bmu^* \|_2 \leq d_2(\bar\btheta^{t+1}_j, \btheta^*_j)  & \leq \kappa d_2(\tilde\btheta_j^{t}, \btheta_j^*) + O(n^{-1/2})  \\
%	& \leq \kappa d_2(\tilde\btheta_j^{t}, \btheta_j^{t}) + \kappa d_2(\btheta_j^{t}, \btheta_j^*) + O(n^{-1/2})  \\
%	& = O(n^{-5/6}) + O(n^{-1/2}),
%\end{align*}
%since $d_2(\btheta^t_j, \tilde\btheta_j^t) = O(n^{-5/6})$ with probability $1-K/n-n^{-2/3}$ and $d_2(\btheta^t_j, \btheta^*_j) = O(n^{-1/2}) $ with probability $1-K/n-(nK)^{-1}$ by Assumption \ref{assump:init}. It leads to 
%\begin{align*}
%	E\|\bar\bmu^{t+1} - \bmu^* \|_2^2 & \leq Ed_2^2(\btheta^t_j, \tilde\btheta_j^t) + Ed_2^2(\btheta^t_j, \btheta^*_j) + O(1/n) + O(K/n) \\
%	%& = O(n^{-5/3}) + O(K/n)+O(n^{-2/3})\\
%	& = O(K/n)+O(n^{-2/3}).
%\end{align*}
%
%For $\|\bmu^{t+1} - \tilde\bmu^t \|_2$, we have %with probability $1-K/n-n^{-2/3}$
%\begin{align*}
%	\|\bmu^{t+1} - \tilde\bmu^t \|_2 & \leq d_2(\btheta_j^{t+1}, \btheta_j^*) + \| \tilde \bmu^t - \bmu^* \|  \\
%	& = O(n^{-5/6})+O_p(n^{-1/2})
%\end{align*}
%and $E\| \bmu^{t+1} - \tilde\bmu^t \|_2^2 = O(K/n) + O(n^{-2/3})$. 

Therefore, we have %with probability $1-K/n - n^{-2/3}$,
\begin{align*}
	\|\bmu' - \bmu^* \|_2 & \leq \|\tilde\bmu^t - \bmu^* \|_2 + \|\bmu^{t+1} - \bmu^* \|_2 \leq O_p(n^{-5/6}) + O_p(n^{-1/2}), \\
	\|\bmu'' - \bmu^* \|_2 & \leq \|\tilde\bmu^t - \bmu^* \|_2 + \|\bar\bmu^{t+1} - \bmu^{t+1} \|_2 + \|\bmu^{t+1} - \bmu^* \|_2 \leq O_p(n^{-5/6}) + O_p(n^{-1/2}), \\
	\|\bmu^{t+1} - \tilde\bmu^t \|_2 & \leq \|\tilde\bmu^t - \bmu^* \|_2 + \|\bmu^{t+1} - \bmu^* \|_2 \leq O_p(n^{-5/6}) + O_p(n^{-1/2}), \\
	\|\bar\bmu^{t+1} - \tilde\bmu^t \|_2 & \leq \|\bmu^{t+1} - \bar\bmu^{t+1} \|_2 + \|\bmu^{t+1} - \bmu^* \|_2 + \|\bmu^* - \tilde\bmu^t \|_2 \leq O_p(n^{-5/6}) + O_p(n^{-1/2}).
\end{align*}

%Recall that $P(n,K) = O(n^{-5/6})$, we have 
%\begin{align*}
%	\| \nabla_{\bmu} & {\widetilde{Q}} ( \bmu^{t+1} | \tilde\btheta^t) \|_2 = (o_p(1)+c) \|\bar\bmu^{t+1} - \bmu^{t+1} \|_2 =  (o_p(1)+c)P(n,K).
%\end{align*}

Next, let's control $P(\epsilon^c)=P((\epsilon_0 \cap \epsilon_1 \cap \epsilon_2)^c)$ where $\epsilon_0= \cap_j \epsilon_{0j}$. 
By Proposition 5.10 in \cite{vershynin2010introduction}, we have $P(\epsilon_{0j}^c)\lesssim \exp\{-n\}$ which leads to $P(\epsilon_0^c)\lesssim K\exp\{-n\} \leq K/n$. To control $P(\epsilon_1^c)=P(\| \nabla_{\bmu\bmu}^2 {\widetilde{Q}} ( \bmu^{t+1} | \tilde\btheta^t) - \nabla_{\bmu\bmu}^2 Q_{\bmu}(\bmu^{t+1} | \btheta^t)  \|_2 > C_1)$, we have
\begin{align}
& \nabla_{\bmu\bmu}^2 {\widetilde{Q}} ( \bmu^{t+1} | \tilde\btheta^t) - \nabla_{\bmu\bmu}^2 Q_{\bmu}(\bmu^{t+1} | \btheta^t) = \nonumber \\
&\frac{1}{Kn}\sum_{i=1}^n\sum_{j=1}^K \{ \nabla_{\bmu\bmu}^2  t(\y_{i1},\tilde\bEta_j^t) h(\y_{i1},\bmu',\tilde\btheta_j^t) - \nabla_{\bmu\bmu}^2t(\y_{i1},\bEta^*_j) h(\y_{i1},\bmu^*,\btheta_j^*)   + \nabla_{\bmu\bmu}^2t(\y_{i1},\bEta^*_j) h(\y_{i1},\bmu^*,\btheta_j^*)  \} \nonumber\\ 
	& -  \frac{1}{Kn}\sum_{i=1}^n\sum_{j=1}^K \{ \nabla_{\bmu\bmu}^2h(\y_{ij},\bmu'',\btheta_j^t) - \nabla_{\bmu\bmu}^2 h(\y_{ij},\bmu^*,\btheta_j^*) + \nabla_{\bmu\bmu}^2 h(\y_{ij},\bmu^*,\btheta_j^*)\}.\nonumber
\end{align}
Therefore, by $\epsilon_0$, Assumption \ref{assump:smooth}, \ref{assump:init} and Lemma S.1, we have 
\begin{align*}
	P(\| \nabla_{\bmu\bmu}^2 {\widetilde{Q}} ( \bmu^{t+1} | \tilde\btheta^t) - \nabla_{\bmu\bmu}^2 Q_{\bmu}(\bmu^{t+1} | \btheta^t)  \|_2 > C_1) = O(K/n) + O(n^{-2/3}).
\end{align*}
As for $P(\epsilon_2^c)=P(\| \nabla_{\bmu} {\widetilde{Q}} ( \bmu^{t+1} | \tilde\btheta^t) \|_2>C_2)$, with the results we obtained when bounding $\| \nabla_{\bmu} {\widetilde{Q}} ( \bmu^{t+1} | \tilde\btheta^t) \|_2$, we have 
%\begin{align*}
%	P(\epsilon_2^c)=P(\| \nabla_{\bmu} {\widetilde{Q}} ( \bmu^{t+1} | \tilde\btheta^t) \|_2>C_2)
%\end{align*}
\begin{align*}
	P(\epsilon_2^c)& =P(\| \nabla_{\bmu} {\widetilde{Q}} ( \bmu^{t+1} | \tilde\btheta^t) \|_2>C_2) \\ 
	& \lesssim \sqrt{(E\|\bmu' - \bmu^* \|^2_2 + \frac{1}{K} \sum_{j=1}^K E\|\tilde\bEta_j^t - \bEta_j^*\|_2^2)\| \bmu^{t+1} - \tilde\bmu^t \|^2_2} \\ 
	& + \sqrt{(E\|\bmu'' - \bmu^* \|_2 + \frac{1}{K} \sum_{j=1}^K E\|\tilde\btheta_j^t - \btheta_j^*\|^2_2)\| \bar\bmu^{t+1} - \tilde\bmu^t \|^2_2} \\ 
	& + \sqrt{E\|\frac{1}{Kn}\sum_{i=1}^n\sum_{j=1}^K \{\nabla_{\bmu\bmu}^2t(\y_{i1},\bEta^*_j) h(\y_{i1},\bmu^*,\btheta_j^*) - \nabla_{\bmu\bmu}^2 h(\y_{ij},\bmu^*,\btheta_j^*) \}\|_2^4 E\|\bmu^{t+1} - \tilde\bmu^t\|_2^4} \\ 
	& + \sqrt{E\|\frac{1}{Kn}\sum_{i=1}^n\sum_{j=1}^K \{\nabla_{\bmu\bmu}^2 h(\y_{ij},\bmu^*,\btheta_j^*) - E\nabla_{\bmu\bmu}^2 h(\y_{ij},\bmu^*,\btheta_j^*) \}\|_2^4 E\|\bar\bmu^{t+1} - \bmu^{t+1}\|_2^4} \\ 
	& + \|\frac{1}{Kn}\sum_{i=1}^n\sum_{j=1}^K E\nabla_{\bmu\bmu}^2 h(\y_{ij},\bmu^*,\btheta_j^*)\|_2^2 E\|\bar\bmu^{t+1} - \bmu^{t+1}\|_2^2 \\
	&  = O(K/n) + O(n^{-2/3}).
\end{align*}
In summary, we have $P(\epsilon)> 1 - K/n - n^{-2/3}$.

Therefore, by combining the above results, with probability at least $1-Kn^{-1}-n^{-2/3}$ we have
\begin{align*}
	\| \tilde\bmu^{t+1} - \bmu^{t+1} \|_2 =
	O(n^{-5/6}).
\end{align*}

Finally, we bound $|\lambda_j^{t+1} - \tilde\lambda_j^{t+1}|$. Recall that
$\lambda_{j}^{t+1}=\frac{1}{n}\sum_{i=1}^n w_{\btheta_j^t}^j(\y_{ij})$ and $\tilde\lambda_{j}^{t+1}=\frac{1}{n}\sum_{i=1}^n w_{\tilde\btheta_j^t}^j(\y_{ij})$, by Assumption \ref{assump:smooth} we have
\begin{align*}
	|\lambda_j^{t+1} - \tilde\lambda_j^{t+1}| & \leq \frac{1}{n}\sum_{i=1}^n  | w_{\btheta_j^t}^j(\y_{ij}) - w_{\tilde\btheta_j^t}^j(\y_{ij}) | \\
	& \leq \frac{1}{n}\sum_{i=1}^n m_3(\y_{ij})d_2( \tilde\btheta_j^t, \btheta_j^t )  = O_p(n^{-5/6}).
\end{align*}
Therefore, with probability at least $1-K/n-n^{-2/3}$ we have 
\begin{align*}
	d_2(\tilde\btheta_j^{t+1}, \btheta_j^{t+1}) \leq \sqrt{2}\| \tilde\bmu^{t+1} - \bmu^{t+1} \|_2 + |\lambda_j^{t+1} - \tilde\lambda_j^{t+1}|= O(n^{-5/6}).
\end{align*}
 This completes the proof.
\end{proof}

\section{Proof of Theorem \ref{coro1}}

\begin{proof}
	From Lemma \ref{lemma1} and Assumption \ref{assump:init} we have with probability at least $1-K/n-n^{-2/3}$
	\begin{align*}
		d_2(\tilde\btheta^t, \btheta^*) &\leq d_2(\tilde\btheta^t, \btheta^t) + d_2(\btheta^t, \btheta^*) \\
		& \leq \sqrt{K}\max_{j}d_2(\tilde\btheta_j^t, \btheta_j^t) + \kappa^td_2(\btheta^0, \btheta^*) + O(\sqrt{K/n}) \\
		& \leq O(n^{-1/3} \sqrt{K/n}) + \kappa^td_2(\btheta^0, \btheta^*) + O(\sqrt{K/n}) \\
		& \leq \kappa^td_2(\btheta^0, \btheta^*) + O(\sqrt{K/n})
	\end{align*}
	since $n^{-1/3} \sqrt{K/n}$ is ignorable compared to $\sqrt{K/n}$. Note that $\kappa<1$, so $\kappa^td_2(\btheta^0, \btheta^*)$ can be dominated by $O(\sqrt{K/n})$ when $t$ is large enough, and we have   
	\begin{align*}
		d_2(\tilde\btheta^t, \btheta^*) \leq c\sqrt{K/n}
	\end{align*}
	with $c$ a positive constant. This completes the proof.
\end{proof}

\section{Proof of Theorem \ref{overall_contraction}}

Consider the contraction region 
\begin{align*}
\B(\btheta^*;c_0,c_1)= & \{\btheta=(\bmu_0,\bmu_1,\lambda_1,\ldots,\lambda_K): \lambda_j\in(c_0,1-c_0), \bmu_0,\bmu_1\in\mathbb{R}^d,\\
                      &\|\bmu_k-\bmu_k^*\|_2(\leq \frac{1}{4\|\bSigma_j\|_2}\|\bmu_0^*-\bmu_1^*\|_2) \leq \frac{M^{3/2}}{4}\Delta_j,k=0,1,\\
                      &(1-c_1)\Delta_j^2<|\delta_0(\bbeta_j)|,|\delta_1(\bbeta_j)|,\sigma^2(\bbeta_j)<(1+c_1)\Delta_j^2  \},
\end{align*} 
based on which we state two technical lemmas needed in the proof of Theorem \ref{overall_contraction}. Definitions of the notations used in defining $\B(\btheta^*;c_0,c_1)$ can be found in Supplementary Material S6. In the following, $c_3, c_4, c_6, C_{\bmu}, c_{\bmu}, c_{\A\B}$ are some functions of $c_0,c_1,M$, and their exact forms can be found in the proof of the two technical lemmas. 
 \begin{lem1}[Contraction on the population iteration]\label{contraction} Suppose $\btheta^* \in\Theta$ and $\btheta\in\B(\btheta^*;c_0,c_1)$, then with 
 \begin{align*}
 	\kappa_1 & = (c_3 \vee C_{\bmu}) \exp(-c_4\Delta_{min}^2), \\
 	\kappa_{2} & =  [ \{ \frac{1}{\sqrt{K}} ( M c_3c_{\A\B} +  c_6 )\} \vee \{ MC_{\bmu} c_{\A\B} +  c_{\bmu}\} ] \exp(-c_4\Delta_{min}^2), \\
 	\kappa_{3} & =  [ \{M c_3c_{\A\B} +  c_6\} \vee \{ MC_{\bmu} c_{\A\B} +  c_{\bmu}\} ] \exp(-c_4\Delta_{min}^2),
 \end{align*}
we have 
 \begin{itemize}
 	\item[1.]  $ |\lambda_j(\btheta)-\lambda_j^*| \leq \kappa_1 d_2(\btheta_j, \btheta_j^*) \leq \kappa_1 d_2(\btheta, \btheta^*)$
 	\item[2.]  $\| \bmu_k(\btheta) - \bmu_k^* \|_2 \leq \kappa_{2}d_2(\btheta,\btheta^*) $ or $
	\| \bmu_k(\btheta) - \bmu_k^* \|_2  
	\leq  \kappa_{3} \frac{1}{K} \sum_{j=1}^K d_2(\btheta_j,\btheta_j^*) $, $k=0,1$.
 \end{itemize}
It implies if $\Delta_{min} > C(c_0,c_1,c_2,M,K)$ with $C(c_0,c_1,c_2,M,K)$ being a positive quantity that depends on $c_0,c_1,c_2, M,$ and $K$, then $\exists\ \kappa = (\sqrt{K}\kappa_1+2\kappa_2) \lesssim \exp(-c\Delta_{min}^2) \in(0,1)$, s.t., $$ d_2(M(\btheta),\btheta^*)\leq \kappa d_2(\btheta,\btheta^*). $$
\end{lem1}
Note that, $\kappa_1$ and $\kappa_3$ are the $\kappa''$ and $\kappa'$ in Theorem \ref{overall_contraction}, respectively.

%Lemma \ref{contraction} established the contraction property of EM estimates at the population level. Recall the two conditions stated in Section \ref{sec:generalTheory} that guarantee the convergence of population estimates, Condition \ref{cond0} can be verified easily as $q(\btheta)$ can be written as a summation of several strongly concave terms. As for Condition \ref{cond1}, note that $\nabla Q(M(\btheta)|\btheta)=\0$, thus it is equivalent to verify $\| \nabla Q(M(\btheta)|\btheta^*) \|_2\leq \gamma \|\btheta - \btheta^*\|_2$. 
%A direct computation can show that $\| \nabla Q(M(\btheta)|\btheta^*) \|_2$ actually measures the distance between the population estimator $M(\btheta)$ and the true coefficient vector $\btheta^*$. Therefore, Lemma \ref{contraction} is essentially a verification of Condition \ref{cond1} in the Gaussian mixture model. Till now, we have already obtained the population contraction and the following lemma allows the extension to sample-based estimator. 

%%%%%%%%%%%%%%%%%%%%%%%%%%%%%
\begin{lem1}[Uniform contraction inequality]\label{concentration} Suppose $\btheta^* \in\Theta$, and Condition \ref{cond2} %\ref{cond2} 
is satisfied, then with probability at least $1-n^{-1}$, 
$$ \sup_{\btheta\in\B(\btheta^*;c_0,c_1)} | \lambda_j^n(\btheta) - \lambda_j(\btheta) | \lesssim \sqrt{\frac{\log(n)}{n}}, $$
and with probability at least $1 - \frac{1}{nK}$ we have
\begin{align*}
	\sup_{\btheta\in \B(\btheta^*;c_0,c_1)} \| \bmu_1^n(\btheta) - \bmu_1(\btheta) \|_2 
	 \lesssim \sqrt{ \frac{\log(nK)}{nK} }. 
\end{align*}
It implies that with probability at least $1-\frac{K}{n}-\frac{1}{nK}$, we have
$$ \sup_{\btheta\in\B(\btheta^*;c_0,c_1)}d_2(M_n(\btheta),M(\btheta))\leq T(n,K) :=  \sqrt{\frac{K\log(n)}{n}} + \sqrt{ \frac{\log(nK)}{nK}}. $$
\end{lem1}

%%%%%%%%%%%%%%%%%%%%%%%%%%%%%%%%%

%\begin{theorem}\label{overall_contraction} Consider the model over the parameter space $\btheta$, under Condition \ref{cond2} %\ref{cond2} 
%and let $\Delta_{min}$ be large enough to satisfy $\Delta_{min} > C(c_0,c_1,M,K)$, also suppose $n$ is large enough to make $T(n,K)$ be $o_p(1)$, then there exists a constant $\kappa = \sqrt{K}\kappa_1 + 2\kappa_2 \lesssim \exp(-c\Delta_{min}^2) \in(0,1)$  such that with probability at least $1-\frac{K}{n}-\frac{1}{nK}$ we have
%$$ d_2( \btheta^{t+1},\btheta^*)\leq \kappa^{t+1} d_2(\btheta^*,\btheta^0) + \frac{1-\kappa^{t+1}}{1-\kappa}T(n,K). $$
%In particular, at the $t+1$-th iteration, we have with probability at least $1-\frac{1}{n}$
%$$ | \lambda_j^{t+1} - \lambda_j^* | \leq \kappa_{1} d_2(\btheta_j^*, \btheta_j^t) + O_p(\sqrt{\frac{\log(n)}{n}}) $$
%and 
%$$ \sum_{k=0}^1 \|  \bmu_k^{t+1} - \bmu_k^* \|_2 \leq \kappa_{3} \frac{1}{K}\sum_{j=1}^K d_2(\btheta_j^*, \btheta_j^t) + O_p(\sqrt{\frac{\log(nK)}{nK}}) $$
%with probability at least $1-\frac{1}{nK}$. Moreover, for the $j$-th site, with probability at least $1- \frac{1}{n} - \frac{1}{nK}$, 
%$$
%d_2( \btheta_j^{t+1},\btheta_j^*) \leq \kappa_3\frac{1}{K}\sum_{m=1}^K d_2( \btheta_m^{t},\btheta_m^*) + \kappa_1 d_2( \btheta_j^{t},\btheta_j^*) + O_p(\sqrt{\frac{\log(n)}{n}}) + O_p(\sqrt{\frac{\log(nK)}{nK}}).
%$$
%where $T(n,K,d)=b_1 \sqrt{\log(nd)} \sqrt{\frac{K(b_2 d + b_3)}{n}}$ with positive constants $b_1,b_2$ and $b_3$.
%\end{theorem}

\begin{proof}[Proof of Theorem \ref{overall_contraction} %\ref{overall_contraction}
]
First we need to verify that Condition \ref{cond2} %\ref{cond2} 
guarantees that the initial estimator $\btheta^0$ is in the contraction region $\B(\btheta^*;c_0,c_1)$. Recall that the parameter space is $$\Theta=\{\btheta=(\lambda_1,\ldots,\lambda_K,\bmu_0,\bmu_1):\forall j\in[K],\lambda_j\in(c_w,1-c_w),\bmu_0,\bmu_1\in \mathbb{R}^d \}$$ with $0<c_w<1$, and the contraction region is
\begin{align*}
\B(\btheta^*;c_0,c_1)= & \{\btheta=(\lambda_1,\ldots,\lambda_K,\bmu_0,\bmu_1): \lambda_j\in(c_0,1-c_0), \bmu_0,\bmu_1\in\mathbb{R}^d,\\
                      &\|\bmu_k-\bmu_k^*\|_2 (\leq \frac{1}{4\|\bSigma_j\|_2}\|\bmu_0^*-\bmu_1^*\|_2) \leq \frac{M^{3/2}}{4}\Delta_j,k=1,2,\\
                      &(1-c_1)\Delta_j^2<|\delta_0(\bbeta_j)|,|\delta_1(\bbeta_j)|,\sigma^2(\bbeta_j)<(1+c_1)\Delta_j^2  \}
\end{align*} 
where 
\begin{align*}
	& \bbeta_j = \bSigma_j^{-1}(\bmu_0-\bmu_1),\ \beta^*_j=\bSigma_j^{-1}(\bmu_0^*-\bmu_1^*) \\
	& \delta_0(\bbeta_j) = \bbeta_j^T(\bmu_0^* - \frac{\bmu_0+\bmu_1}{2}),\ \delta_1(\bbeta_j) = \bbeta_j^T(\bmu_1^* - \frac{\bmu_0+\bmu_1}{2})\\
	& \sigma(\bbeta_j)=\sqrt{\bbeta_j^T\bSigma_j\bbeta_j}=\sqrt{(\bmu_0-\bmu_1)^T\bSigma_j^{-1}(\bmu_0-\bmu_1)},\\
	& \Delta_j=\sqrt{\bbeta_j^{*T}\bSigma_j\bbeta_j^*}=\sqrt{(\bmu_0^*-\bmu_1^*)^T\bSigma_j^{-1}(\bmu_0^*-\bmu_1^*)},
\end{align*}
and constants $c_0,c_1,c_w$ satisfy $0 < c_0 \leq c_w<1, 1/2 < c_1<1$. %Here we can see that $\Delta_j$ can be regarded as the SNR of site $j$ and we further require that $\exists\ \Delta_{min}, \Delta_{max}>0$, s.t., for $\forall j\in[K]$, we have $\Delta_{min}\leq \Delta_j\leq\Delta_{max}$. Thus, $\Delta_{min}$ is the global SNR in our multi-site learning setting.
Condition \ref{cond2} states that the initial estimator $\btheta^0$ satisfies $d_2(\btheta^0,\btheta^*)\leq r\Delta_{min}$, with $$r<\frac{M^{3/2}}{4}\wedge\frac{|c_0-c_w|}{\Delta_{min}}\wedge(\sqrt{\frac{(2c_1-1)}{M}+\frac{4}{M}}-\frac{2}{\sqrt{M}})\wedge(\sqrt{ \frac{c_1}{M} + \frac{1}{4}(M + \frac{1}{M} +2) } - \frac{1}{2}(\sqrt{M} + \frac{1}{\sqrt{M}})).$$
Thus, what we need to do is to use the statements in both Condition \ref{cond2} and $\Theta$ to derive each statement in $\B(\btheta^*;c_0,c_1)$. 

As $\lambda_j^* \in (c_w,1-c_w)$, when $|\lambda_j -\lambda_j^* | \leq r\Delta_{min} \leq |c_0 - c_w |$, we have $\lambda_j \in (c_0, 1-c_0)$. 

For $\delta_0(\bbeta_j)$ and $\delta_1(\bbeta_j)$, we have
\begin{align*}
	&| \Delta_j^2 - \bbeta_j^T(\bmu_0^* - \frac{\bmu_0 + \bmu_1}{2}) | \\ & = | \bbeta_j^{*T}(\bmu_0^* - \bmu_1^*) - \bbeta_j^T (\bmu_0^* - \bmu_1^* + \bmu_1^* - \frac{\bmu_0 + \bmu_1}{2}) | \\
	& \leq \| \bbeta^*_j - \bbeta_j \|_2 \| \bmu_0^* - \bmu_1^* \|_2 + | \bbeta_j^T (\bmu_1^* - \frac{\bmu_0 + \bmu_1}{2}) | \\
	& \leq M \| \bmu_0^* - \bmu_1^* - (\bmu_0 - \bmu_1)  \|_2\| \bmu_0^* - \bmu_1^* \|_2 + | (\bmu_0 - \bmu_1)^T \bOmega_j(2\bmu_1^* - \bmu_0 -\bmu_1) |/2 \\ 
	& \leq 2M^{1/2} r \Delta_j^2 + \frac{M}{2}r^2\Delta_j^2 + \frac{1}{2}\Delta_j^2 \\
	& \leq c_1 \Delta_j^2
\end{align*} 
where the last inequality is because $$r<\sqrt{\frac{(2c_1-1)}{M}+\frac{4}{M}}-\frac{2}{\sqrt{M}}.$$ It leads to $(1-c_1)\Delta_j^2<|\delta_0(\bbeta_j)|,\ |\delta_1(\bbeta_j)|<(1+c_1)\Delta_j^2$. 

For $\sigma(\bbeta_j)$, we have
$$ | \sigma^2(\bbeta_j) - \Delta_j^2 | \leq r\Delta^2_jM(r + \sqrt{M}) + r\sqrt{M}\Delta_j^2 \leq c_1 \Delta_j^2  $$ because
$$r<\sqrt{ \frac{c_1}{M} + \frac{1}{4}(M + \frac{1}{M} +2) } - \frac{1}{2}(\sqrt{M} + \frac{1}{\sqrt{M}}),$$ and it leads to $(1-c_1)\Delta_j^2<\sigma^2(\bbeta_j)<(1+c_1)\Delta_j^2$. 

We also have $ \|\bmu_k - \bmu_k^* \|_2 \leq M^{3/2}\Delta_{j}/4$ as $r\leq \frac{M^{3/2}}{4}$. Thus, Condition \ref{cond2} can guarantee $\btheta^0 \in \B(\btheta^*;c_0,c_1)$.

Next we prove Theorem \ref{overall_contraction} %\ref{overall_contraction} 
use induction
% Consider the model over the parameter space $\btheta$, under condition \ref{cond2} and let $\Delta_{min}$ be large enough to satisfy $\Delta_{min} > C(c_0,c_1,M,K)$, then there exists constants $\kappa\in(0,1)$ such that with probability at least $1-\frac{K}{n}-\frac{1}{nK}$ we have
%$$ d_2( \btheta^{t+1},\btheta^*)\leq \kappa^{t+1} d_2(\btheta^*, \btheta^0) + \frac{1-\kappa^{t+1}}{1-\kappa}T(n,K,d) $$
under Condition \ref{cond2} and the conditions in Lemma S.3 and Lemma S.4. Under Condition \ref{cond2}, we have $\btheta^0 \in \B(\btheta^*;c_0,c_1)$, it follows from Lemma S.3 and Lemma S.4 that
\begin{align*}
	d_2( \btheta^1,\btheta^*) = d_2(M_n( \btheta^0),\btheta^*) & \leq d_2(M( \btheta^0),\btheta^*) + T(n, K) \\ & \leq \kappa d_2( \btheta^0,\btheta^*) + T(n, K) \\ 
	& \leq \kappa r\Delta_{min} + T(n, K)     
\end{align*}
	which also implies $ \btheta^1 \in \B(\btheta^*;c_0,c_1)$ when $n$ is large enough to make $T(n, K) \leq (1-\kappa)r\Delta_{min} $. Now let's assume this property holds at the $t$-th step, i.e.,
	$$ d_2( \btheta^{t},\btheta^*)\leq \kappa^{t} d_2(\btheta^*, \btheta^0) + \frac{1-\kappa^{t}}{1-\kappa}T(n,K) $$ and $ \btheta^t \in \B(\btheta^*;c_0,c_1)$. Then we have
	\begin{align*}
		d_2( \btheta^{t+1},\btheta^*) = d_2(M_n( \btheta^{t}),\btheta^*) &\leq d_2(M( \btheta^{t}),\btheta^*) + T(n, K) \\ 
		& \leq \kappa [\kappa^{t} d_2(\btheta^*, \btheta^0) + \frac{1-\kappa^{t}}{1-\kappa}T(n,K)] + T(n,K) \\ 
		& = \kappa^{t+1} d_2(\btheta^*, \btheta^0) + \frac{1-\kappa^{t+1}}{1-\kappa}T(n,K).
	\end{align*}
	It also leads to $d_2( \btheta^{t+1},\btheta^*)\leq \kappa^{t+1}r\Delta_{min} + \frac{1-\kappa^{t+1}}{1-\kappa}T(n,K)$ and guarantees $ \btheta^{t+1}\in\B(\btheta^*;c_0,c_1)$ when $n$ is large enough to make $T(n, K) \leq (1-\kappa)r\Delta_{min}$. It completes the proof.
\end{proof}

\section{Proof of Lemma S.2}

\begin{proof}
Recall that $\bar\bmu^{t+1}$ is the solution of $\nabla_{\bmu}Q_{\bmu} ( \bmu|\tilde\btheta^t)=\0$, and $\bmu^{t+1}$ is the root of $\nabla_{\bmu}Q_{\bmu} ( \bmu|\btheta^t)=\0$. We use Lemma 6 of \cite{zhang2013communication} to measure the distance between $\bar\bmu^{t+1}$ and $\bmu^{t+1}$. 

Let's first define the following three good events
\begin{align*}
	&\epsilon_{0j} := \{ \frac{1}{n} \sum_{i=1}^n  m_1(\y_{ij}) \leq 2L\},\ j\in [K], \\
	&\epsilon_1 := \{ \| \nabla_{\bmu\bmu}^2 Q_{\bmu} ( \bmu^{t+1} | \tilde\btheta^t) - \nabla_{\bmu\bmu}^2 Q_{\bmu}(\bmu^{t+1} | \btheta^t)  \|_2 \leq C_1 \}, \\
	&\epsilon_2 := \{ \| \nabla_{\bmu} Q_{\bmu} ( \bmu^{t+1} | \tilde\btheta^t) \|_2 \leq C_2 \},
\end{align*}
where %$E(m_1(\y_{ij}))<M$ for all $i=1,\ldots,n$ and $j=1,\ldots, K$, 
$C_1\leq \rho\mu_-/2$ and $C_2 \leq (1-\rho)\mu_-\delta_{\rho}/2$ with $\delta_{\rho}=\min\{\rho, \rho\mu_-/(4L)\}$. 
Then, by using Lemma 6 in \cite{zhang2013communication}, under $\epsilon_0 \cap \epsilon_1 \cap\epsilon_2 $ where $\epsilon_0= \cap_j \epsilon_{0j}$ we have $\|\bar \bmu^{t+1} - \bmu^{t+1}  \|_2 \leq C \| \nabla_{\bmu} Q_{\bmu} ( \bmu^{t+1} | \tilde\btheta^t) \|_2$. Next we control $\| \nabla_{\bmu} Q_{\bmu} ( \bmu^{t+1} | \tilde\btheta^t) \|_2$. 

Since $\nabla_{\bmu} Q_{\bmu} ( \bmu^{t+1} | \btheta^t)= \0 $, with $\check\btheta_j=\alpha_j\btheta_j^t + (1-\alpha_j)\tilde\btheta_j^{t}$ where $\alpha_j\in(0,1)$ we have
\begin{align*}
	\nabla_{\bmu} & Q_{\bmu}(\bmu^{t+1}|\tilde\btheta^t) = \nabla_{\bmu} Q_{\bmu}(\bmu^{t+1}|\tilde\btheta^t) - \nabla_{\bmu} Q_{\bmu} ( \bmu^{t+1} | \btheta^t) \\
	& = \frac{1}{Kn}\sum_{i=1}^n\sum_{j=1}^K \{ \nabla_{\bmu} h(\y_{ij},\bmu^{t+1},\tilde\btheta_j^t) - \nabla_{\bmu} h(\y_{ij}, \bmu^{t+1}, \btheta_j^t) \}\\
	& = \frac{1}{Kn}\sum_{i=1}^n\sum_{j=1}^K \{ \nabla_{\bmu\btheta_j}^2 h(\y_{ij},\bmu^{t+1},\check\btheta_j)(\tilde\btheta_j^{t} - \btheta_j^t) \}\\
	& = \frac{1}{Kn}\sum_{i=1}^n\sum_{j=1}^K \nabla_{\bmu\btheta_j}^2\{  h(\y_{ij},\bmu^{t+1},\check\btheta_j) - h(\y_{ij},\bmu^*,\btheta_j^*) \}(\tilde\btheta_j^{t} - \btheta_j^t)  \\
	& + \frac{1}{Kn}\sum_{i=1}^n\sum_{j=1}^K \nabla_{\bmu\btheta_j}^2\{  h(\y_{ij},\bmu^*,\btheta_j^*) - E h(\y_{ij},\bmu^*,\btheta_j^*) \}(\tilde\btheta_j^{t} - \btheta_j^t) \\
	& + \frac{1}{Kn}\sum_{i=1}^n\sum_{j=1}^K \nabla_{\bmu\btheta_j}^2 E h(\y_{ij},\bmu^*,\btheta_j^*)(\tilde\btheta_j^{t} - \btheta_j^t),  
\end{align*}
and it leads to
\begin{align*}
	\|\nabla_{\bmu} Q_{\bmu}(\bmu^{t+1}|\tilde\btheta^t)\|_2 
		& \leq [ \|\frac{1}{Kn}\sum_{i=1}^n\sum_{j=1}^K \nabla_{\bmu\btheta_j}^2\{  h(\y_{ij},\bmu^{t+1},\check\btheta_j) - h(\y_{ij},\bmu^*,\btheta_j^*) \}\|_2  \\
	& + \|\frac{1}{Kn}\sum_{i=1}^n\sum_{j=1}^K \nabla_{\bmu\btheta_j}^2\{  h(\y_{ij},\bmu^*,\btheta_j^*) - E h(\y_{ij},\bmu^*,\btheta_j^*) \}\|_2 \\
	& + \|\frac{1}{Kn}\sum_{i=1}^n\sum_{j=1}^K \nabla_{\bmu\btheta_j}^2 E h(\y_{ij},\bmu^*,\btheta_j^*) \|_2 ] \|\tilde\btheta_j^{t} - \btheta_j^t\|_2.  
\end{align*}
Therefore, by $\epsilon_0$, Assumption \ref{loc_conc}--\ref{assump:smooth} and Lemma S.1, we have $$ \|\nabla_{\bmu} Q_{\bmu}(\bmu^{t+1}|\tilde\btheta^t)\|_2 = O_p(n^{-5/6}). %(c + o_p(1))\|\tilde\btheta_j^{t} - \btheta_j^t\|_2 = (c + o_p(1))P(n,K). 
$$ 

Next, let's control $P(\epsilon^c)=P((\epsilon_0 \cap \epsilon_1 \cap \epsilon_2)^c)$ where $\epsilon_0= \cap_j \epsilon_{0j}$. 
By Assumption \ref{assump:smooth} and Lemma S.1, %Proposition 5.10 in Vershynin \cite{vershynin2010introduction}, 
we have $P(\epsilon_{0j}^c)\lesssim 1/n$ which leads to $P(\epsilon_0^c)\lesssim  K/n$. To control $P(\epsilon_1^c)$, since
\begin{align*}
 \| \nabla_{\bmu\bmu}^2 Q_{\bmu} ( \bmu^{t+1} | \tilde\btheta^t) & - \nabla_{\bmu\bmu}^2 Q_{\bmu}(\bmu^{t+1} | \btheta^t)  \|_2 \\
	& \leq [ \|\frac{1}{Kn}\sum_{i=1}^n\sum_{j=1}^K \nabla_{\bmu\bmu}^2\{  h(\y_{ij},\bmu^{t+1},\tilde\btheta_j^t) - h(\y_{ij},\bmu^*,\btheta_j^*) \}\|_2  \\
	& + \|\frac{1}{Kn}\sum_{i=1}^n\sum_{j=1}^K \nabla_{\bmu\bmu}^2\{  h(\y_{ij},\bmu^{t+1},\btheta_j^t) - h(\y_{ij},\bmu^*,\btheta_j^*) \}\|_2, 
\end{align*}
we have 
\begin{align*}
P(\epsilon_1^c) & = P(\| \nabla_{\bmu\bmu}^2 Q_{\bmu} ( \bmu^{t+1} | \tilde\btheta^t) - \nabla_{\bmu\bmu}^2 Q_{\bmu}(\bmu^{t+1} | \btheta^t)  \|_2 > C_1) \\
	& \leq P( \|\frac{1}{Kn}\sum_{i=1}^n\sum_{j=1}^K \nabla_{\bmu\bmu}^2\{  h(\y_{ij},\bmu^{t+1},\tilde\btheta_j^t) - h(\y_{ij},\bmu^*,\btheta_j^*) \}\|_2 > C_1/2)  \\
	& + P(\|\frac{1}{Kn}\sum_{i=1}^n\sum_{j=1}^K \nabla_{\bmu\bmu}^2\{  h(\y_{ij},\bmu^{t+1},\btheta_j^t) - h(\y_{ij},\bmu^*,\btheta_j^*) \}\|_2 > C_1/2) \\
	& \leq E\| \bmu^{t+1} - \bmu^*  \|_2^2 + Ed_2^2(\tilde\btheta_j^t, \btheta_j^t) + Ed_2^2(\btheta_j^t, \btheta_j^*) \\ 
	& = O(K/n) + O(n^{-2/3}).
\end{align*}
Finally, we consider $P(\epsilon_2^c)$, and by Markov's inequality and Holder' inequality we have 
\begin{align*}
	P(\epsilon_2^c) & = P(\|\nabla_{\bmu} Q_{\bmu}(\bmu^{t+1}|\tilde\btheta^t)\|_2 > C_2) \\
	& \leq [ E\|\frac{1}{Kn}\sum_{i=1}^n\sum_{j=1}^K \nabla_{\bmu\btheta_j}^2\{  h(\y_{ij},\bmu^{t+1},\check\btheta_j) - h(\y_{ij},\bmu^*,\btheta_j^*) \}\|_2^2 E\|\tilde\btheta_j^{t} - \btheta_j^t\|_2^2 ]^{1/2}  \\
	& + [ E\|\frac{1}{Kn}\sum_{i=1}^n\sum_{j=1}^K \nabla_{\bmu\btheta_j}^2\{  h(\y_{ij},\bmu^*,\btheta_j^*) - E h(\y_{ij},\bmu^*,\btheta_j^*) \}\|_2^4 E\|\tilde\btheta_j^{t} - \btheta_j^t\|_2^4 ]^{1/2} \\
	& + \|\frac{1}{Kn}\sum_{i=1}^n\sum_{j=1}^K \nabla_{\bmu\btheta_j}^2 E h(\y_{ij},\bmu^*,\btheta_j^*) \|_2^2 E\|\tilde\btheta_j^{t} - \btheta_j^t\|_2^2 \\ 
	& = O(K/n) + O(n^{-2/3}).  
\end{align*}
Therefore, we have
\begin{align*}
	P(\epsilon^c)  & \leq P(\epsilon_0^c)+P(\epsilon_1^c)+P(\epsilon_2^c) \\ 
	& = c_1 \frac{K}{n} + c_2 n^{-2/3},
\end{align*}
and with $P(\epsilon) \geq 1 - c_1 \frac{K}{n} + c_2 n^{-2/3}$, we have
$$ \|\bar \bmu^{t+1} - \bmu^{t+1}  \|_2 = O(n^{-5/6}). $$
	
\end{proof}

\section{Gaussian mixture model's updating formula in EM algorithm}

In this section, we briefly go over the updating formulas of EM algorithm and introduce more notations required in the subsequent theoretical analysis. For Gaussian mixture models, the complete log-likelihood is
\begin{align*}
L_c(\btheta') & =\frac{1}{Kn}\sum_{i=1}^n\sum_{j=1}^K Z_{ij} \{ \log f(\y_{ij}|\bmu_1') + \log(\lambda_j') \} +(1-Z_{ij})\{\log f(\y_{ij}|\bmu_0') +\log (1-\lambda_j')\} \\
    & = -\frac{1}{2Kn}\sum_{i=1}^n\sum_{j=1}^K[(1-Z_{ij})(\y_{ij}-\bmu_0')^T\bOmega_j(\y_{ij}-\bmu_0')+Z_{ij}(\y_{ij}-\bmu_1')^T\bOmega_j(\y_{ij}-\bmu_1')] \\
	& + \frac{1}{Kn}\sum_{i=1}^n\sum_{j=1}^K[(1-Z_{ij})\log(1-\lambda_j')+Z_{ij}\log(\lambda_j')]
\end{align*}
where $\bOmega_j=\bSigma_j^{-1}$ and we consider a more general setting with heterogeneous variance-covariance matrix across sites.
With a given parameter $\btheta$, we have   
\[
\gamma_{\btheta}(\y_{ij})=E(Z_{ij}|\y_{ij},\btheta) =\frac{\lambda_j}{\lambda_j + (1-\lambda_j)\exp\{(\bmu_0-\bmu_1)^T\bSigma_j^{-1}(\y_{ij}-\frac{\bmu_0+\bmu_1}{2})\}}.
\]
%and 
%\[
%1-\gamma_{\btheta}(\y_{ij})=P_{\btheta}(Z_{ij}=0|\Y_{ij}=\y_{ij}, S_{ij}=j). 
%\]
Thus, in the E-step we get the Q function
\begin{align*}
& Q(\btheta'|\btheta) \\
	& = -\frac{1}{2Kn}\sum_{i=1}^n\sum_{j=1}^K[(1-\gamma_{\btheta}(\y_{ij}))(\y_{ij}-\bmu_0')^T\bOmega_j(\y_{ij}-\bmu_0')+\gamma_{\btheta}(\y_{ij})(\y_{ij}-\bmu_1')^T\bOmega_j(\y_{ij}-\bmu_1')] \\
	& + \frac{1}{Kn}\sum_{i=1}^n\sum_{j=1}^K[(1-\gamma_{\btheta}(\y_{ij}))\log(1-\lambda_j')+\gamma_{\btheta}(\y_{ij})\log(\lambda_j')],
\end{align*}
and in the M-step we update $\btheta$ by 
$$ M_n(\btheta)=(\bmu_0^n(\btheta), \bmu_1^n(\btheta),\lambda_{1}^n(\btheta),\ldots,\lambda_{K}^n(\btheta)) = \arg\max_{\btheta'} Q(\btheta'|\btheta). $$
The explicit updating formula is
\begin{align*}
	& \lambda_{j}^n(\btheta)=\frac{1}{n}\sum_{i=1}^n \gamma_{\btheta}(\y_{ij}),\ j\in [K]  \\
	& \bmu_{0}^n(\btheta)=[\frac{1}{nK}\sum_{j=1}^K\sum_{i=1}^n (1-\gamma_{\btheta}(\y_{ij})) \bOmega_j ]^{-1} \frac{1}{nK}\sum_{j=1}^K\sum_{i=1}^n (1-\gamma_{\btheta}(\y_{ij})) \bOmega_j \y_{ij},\\
	& \bmu_1^n(\btheta)=[\frac{1}{nK}\sum_{j=1}^K\sum_{i=1}^n \gamma_{\btheta}(\y_{ij}) \bOmega_j ]^{-1} \frac{1}{nK}\sum_{j=1}^K\sum_{i=1}^n \gamma_{\btheta}(\y_{ij}) \bOmega_j \y_{ij}.
\end{align*}
The corresponding population version is
\begin{align}
    & M(\btheta)=(\bmu_0(\btheta), \bmu_1(\btheta),\lambda_{1}(\btheta),\ldots,\lambda_{K}(\btheta)) = \arg\max_{\btheta'} E[Q(\btheta'|\btheta)] \\
	& \lambda_{j}(\btheta)=E\gamma_{\btheta}(Y_{j}),\ j\in [K] \label{eqpop1} \\
	& \bmu_{0}(\btheta)=[\frac{1}{K}\sum_{j=1}^K (1-E\gamma_{\btheta}(Y_{j})) \bOmega_j ]^{-1} \frac{1}{K}\sum_{j=1}^K E[(1-\gamma_{\btheta}(Y_{j})) \bOmega_j Y_{j}], \label{eqpop2}\\
	& \bmu_1(\btheta)=[\frac{1}{K}\sum_{j=1}^K E\gamma_{\btheta}(Y_{j}) \bOmega_j ]^{-1} \frac{1}{K}\sum_{j=1}^K E[\gamma_{\btheta}(Y_{j}) \bOmega_j Y_{j}]. \label{eqpop3}
\end{align}
Here we remark that, for simplicity, in all the population version notations we use $(Y_j,Z_j)$ to represent random variables from site $j$. All the expectations are taken with respect to the true model with parameter $\btheta^*$ except for extra remark. %Similar as before, we let $\btheta_j=(\bmu_0,\bmu_1,\lambda_j)$ and $\bEta_j=(\bmu_0,\bmu_1,\lambda_j,\lambda_1)$. 

Next we define the contraction region $\B(\btheta^*;c_0,c_1)$ as:  
\begin{align*}
\B(\btheta^*;c_0,c_1)= & \{\btheta=(\bmu_0,\bmu_1,\lambda_1,\ldots,\lambda_K): \lambda_j\in(c_0,1-c_0), \bmu_0,\bmu_1\in\mathbb{R}^d,\\
                      &\|\bmu_k-\bmu_k^*\|_2(\leq \frac{1}{4\|\bSigma_j\|_2}\|\bmu_0^*-\bmu_1^*\|_2) \leq \frac{M^{3/2}}{4}\Delta_j,k=0,1,\\
                      &(1-c_1)\Delta_j^2<|\delta_0(\bbeta_j)|,|\delta_1(\bbeta_j)|,\sigma^2(\bbeta_j)<(1+c_1)\Delta_j^2  \},
\end{align*} 
where 
\begin{align*}
	& \bbeta_j = \bSigma_j^{-1}(\bmu_0-\bmu_1),\ \bbeta^*_j=\bSigma_j^{-1}(\bmu_0^*-\bmu_1^*) \\
	& \delta_0(\bbeta_j) = \bbeta_j^T(\bmu_0^* - \frac{\bmu_0+\bmu_1}{2}),\ \delta_1(\bbeta_j) = \bbeta_j^T(\bmu_1^* - \frac{\bmu_0+\bmu_1}{2})\\
	& \sigma(\bbeta_j)=\sqrt{\bbeta_j^T\bSigma_j\bbeta_j}=\sqrt{(\bmu_0-\bmu_1)^T\bSigma_j^{-1}(\bmu_0-\bmu_1)},\\
	& \Delta_j=\sqrt{\bbeta_j^{*T}\bSigma_j\bbeta_j^*}=\sqrt{(\bmu_0^*-\bmu_1^*)^T\bSigma_j^{-1}(\bmu_0^*-\bmu_1^*)},
\end{align*}
and constants $c_0,c_1,c_w$ satisfy $0 < c_0 \leq c_w < 1/2$ and $ 1/2< c_1<1$. The contraction property of EM algorithm will be considered within this region. In the proof of Theorem \ref{overall_contraction}, we will show that any $\btheta \in \Theta$ that satisfies Condition \ref{cond2} guarantees $\btheta \in \B(\btheta^*;c_0,c_1)$.

Since we consider a heterogeneous variance-covariance matrix setting where $\bSigma_i \neq \bSigma_j$ for $i\neq j$, for all $j\in[K]$, we require that there are some positive constants $M$ to make $M^{-1}\leq \lambda_{min}(\bSigma_j)\leq \lambda_{max}(\bSigma_j)\leq M$ where $\lambda_{min}(\cdot)$ and $\lambda_{max}(\cdot)$ denote the smallest and largest eigenvalues, respectively. We use $\Delta_j$ as the signal to noise ratio of site $j$ and we further require that there exist $\Delta_{min}, \Delta_{max}>0$, s.t., $ \Delta_{min}=\sup_{\Delta}\{\Delta \leq \Delta_j, \text{for all}\ j\in [K] \} $ and $ \Delta_{max}=\inf_{\Delta}\{\Delta \geq \Delta_j, \text{for all}\ j\in [K] \} $. Thus, for any $j\in [K]$ we have $\Delta_{min}\leq \Delta_j\leq\Delta_{max}$, and $\Delta_{min}$ can be viewed as the global signal to noise ratio in our multi-site learning setting. In addition, we require $\Delta_{max} = c_2 \Delta_{min}$ with a positive constant $c_2>1$. The conclusions of homogeneous variance-covariance matrix setting presented in the main body of the manuscript can be easily obtained by letting $\Delta=\Delta_{min}=\Delta_{max}$ and $c_2=1$.

\section{Verification of Assumption \ref{assump:smooth} in Gaussian mixture model}

\begin{proof}
For Gaussian mixture model, we have
$$ h(\y;\bmu,\btheta_j')=-\frac{1}{2}w_{\btheta_j'}(\y)(\y-\bmu_1)^T\bSigma_j^{-1}(\y-\bmu_1) -\frac{1}{2}(1-w_{\btheta_j'}(\y))(\y-\bmu_0)^T\bSigma_j^{-1}(\y-\bmu_0), $$
based on which the derivatives can be derived as below:
\begin{align*}
&\nabla_{\bmu} h(\y;\bmu,\btheta_j') =  
\begin{pmatrix}
w^j_{\btheta_j'}(\y)\bSigma_j^{-1}(\y-\bmu_1)  \\
(1-w^j_{\btheta_j'}(\y))\bSigma_j^{-1}(\y-\bmu_0)
\end{pmatrix}, \\
&\nabla^2_{\bmu\bmu} h(\y;\bmu,\btheta_j') =  
\begin{pmatrix}
-w^j_{\btheta_j'}(\y)\bSigma_j^{-1} & \0 \\
\0 & (w^j_{\btheta_j'}(\y)-1)\bSigma_j^{-1}
\end{pmatrix} = 
\begin{pmatrix}
-w^j_{\btheta_j'}(\y) & 0 \\
0 & (w^j_{\btheta_j'}(\y)-1)
\end{pmatrix}\otimes \bSigma_j^{-1},\\
&
\nabla^2_{\bmu\bmu} h(\y;\bmu,\btheta_j') t(\y,\bEta_j')  = 
\begin{pmatrix}
-w^j_{\btheta_j'}(\y)t(\y,\bEta_j') & 0 \\
0 & (w^j_{\btheta_j'}(\y)-1)t(\y,\bEta_j')
\end{pmatrix}\otimes \bSigma_j^{-1}, \\
&\nabla^2_{\bmu\btheta_j} h(\y;\bmu,\btheta_j') =  
\begin{pmatrix}
\bSigma_j^{-1}(\y-\bmu_1) \otimes (\frac{\partial w^j_{\btheta_j'}(\y)}{\partial \btheta_j'})^T  \\
\bSigma_j^{-1}(\y-\bmu_0) \otimes (\frac{(1-\partial w^j_{\btheta_j'}(\y))}{\partial \btheta_j'})^T
\end{pmatrix},
\end{align*}
where the derivative $\nabla^2_{\bmu\btheta_j}h(\y;\bmu,\btheta_j)$ is taken with respect to the $\bmu$ part and $\btheta_j$ part sequentially.  

\subsection{Lipchitz-continuity}

We first handle the derivatives taken with respect to $\bmu$ only. Let's denote $H(\y;\btheta_j')=\nabla^2_{\bmu\bmu} h(\y;\bmu,\btheta_j')$ and $\widetilde H(\y;\bEta_j')=\nabla^2_{\bmu\bmu} h(\y;\bmu,\btheta_j')t(\y,\bEta_j')$, then we have
\begin{align*}
\| H(\y;\bar\btheta_j) - H(\y;\bar\btheta_j') \|_2 &= \left \| 
\begin{pmatrix}
-w^j_{\bar\btheta_j}(\y) + w^j_{\bar\btheta_j'}(\y) & 0 \ \\
0 & w^j_{\bar\btheta_j}(\y) - w^j_{\bar\btheta_j'}(\y)
\end{pmatrix}\otimes \bSigma_j^{-1} \right \|_2 \\
& \leq | w^j_{\bar\btheta_j}(\y) - w^j_{\bar\btheta_j'}(\y) | \cdot d_1(\bSigma_j^{-1}) 
\end{align*}
and 
\begin{align*}
& \|\widetilde H(\y;\bar\bEta_j) - \widetilde H(\y;\bar\bEta_j') \|_2 = \\ & \left \| 
\begin{pmatrix}
-w^j_{\bar\btheta_j}(\y)t(\y,\bar\bEta_j) + w^j_{\bar\btheta_j'}(\y)t(\y,\bar\bEta_j') & 0 \\
0 & (1-w^j_{\bar\btheta_j'}(\y))t(\y,\bar\bEta_j') - (1-w^j_{\bar\btheta_j}(\y))t(\y,\bar\bEta_j)
\end{pmatrix}\otimes \bSigma_j^{-1} \right \|_2 \\
& \leq ( |w^j_{\bar\btheta_j}(\y)t(\y,\bar\bEta_j) - w^j_{\bar\btheta_j'}(\y)t(\y,\bar\bEta_j') | + | t(\y,\bar\bEta_j') - t(\y,\bar\bEta_j) |) \cdot d_1(\bSigma_j^{-1}) 
\end{align*}
where $d_1(\cdot)$ is the largest singular value of a matrix. Thus, we only need to verify 
\begin{align}
 	| w^j_{\bar\btheta_j}(\y) - w^j_{\bar\btheta_j'}(\y) | \leq m_1(\y)\| \bar\btheta_j - \bar\btheta_j' \|_2 \label{lip1}
 \end{align}
and 
\begin{align}
 	( |w^j_{\bar\btheta_j}(\y)t(\y,\bar\bEta_j) - w^j_{\bar\btheta_j'}(\y)t(\y,\bar\bEta_j') | + | t(\y,\bar\bEta_j') - t(\y,\bar\bEta_j) |) \leq m_2(\y)\| \bar\bEta_j - \bar\bEta_j' \|_2 \label{lip2}
 \end{align}
 with $m_1(\y)$ and $m_2(\y)$ satisfy the moment conditions stated in Assumption \ref{assump:smooth}. %being sub-Gaussian distributed random variables.

Let's first handle \eqref{lip1}. Let $\Delta_{\btheta_j}=\btheta_j - \btheta_j'$, and $\btheta_j^u = \btheta_j' + u \Delta_{\btheta_j}$ with $u \in (0,1)$, then 
\begin{align*}
	w^j_{\btheta_j}(\y) - w^j_{\btheta_j'}(\y) & = \int_0^1 \left\langle \frac{\partial w^j_{\btheta_j}(\y)}{\partial \btheta_j}|_{\btheta_j=\btheta_j^u}, \Delta_{\btheta_j}  \right\rangle du \\
	& = \int_0^1 \left\langle \frac{\partial w^j_{\btheta_j}(\y)}{\partial \lambda_j}|_{\btheta_j=\btheta_j^u}, \Delta_{\lambda_j}  \right\rangle du + \sum_{k=0}^1 \int_0^1 \left\langle \frac{\partial w^j_{\btheta_j}(\y)}{\partial \bmu_k}|_{\btheta_j=\btheta_j^u}, \Delta_{\bmu_k}  \right\rangle du.
\end{align*}
From
\begin{align*}
	w^j_{\btheta_j}(\y) & =\frac{\lambda_j}{\lambda_j + (1-\lambda_j)\exp\{(\bmu_0-\bmu_1)^T\bSigma_j^{-1}(\y-\frac{\bmu_0+\bmu_1}{2})\}} \\
	&= \frac{\lambda_j}{\lambda_j + (1-\lambda_j)\exp\{\bbeta_j^T(\y-\frac{\bmu_0+\bmu_1}{2})\}}, 
\end{align*}
we have 
\begin{align*}
	\frac{\partial w^j_{\btheta_j}(\y)}{\partial \lambda_j} 
	&= \frac{\exp\{\bbeta_j^T(\y-\frac{\bmu_0+\bmu_1}{2})\}}{[\lambda_j + (1-\lambda_j)\exp\{\bbeta_j^T(\y-\frac{\bmu_0+\bmu_1}{2})\}]^2},  \\
	& = \frac{\exp(t)}{\{ \lambda_j + (1-\lambda_j)\exp(t) \}^2} \ \text{(using}\ t= \bbeta_j^T(\y-\frac{\bmu_0+\bmu_1}{2}) ) \\
	& = \frac{1}{\{ \lambda_j\exp(-t/2) + (1-\lambda_j)\exp(t/2) \}^2} \\
	& \leq \frac{1}{4\lambda_j(1-\lambda_j)} \ \text{(using results in Section \ref{sectool} and}\ t \in \mathbb{R}).
	\end{align*}
In addition, we have 	
\begin{align*}
	\frac{\partial w^j_{\btheta_j}(\y)}{\partial \bmu_k}&= \lambda_j(1-\lambda_j)\frac{\partial w^j_{\btheta_j}(\y)}{\partial \lambda_j}\bOmega_j(\y-\bmu_k)(-1)^{1(k=1)},\ k=0,1, 
\end{align*}
which leads to 
\begin{align*}
	\left\| \frac{\partial w^j_{\btheta_j}(\y)}{\partial \bmu_k} \right\|_2 \leq \frac{1}{4}\| \bOmega_j(\y-\bmu_k) \|_2. 
\end{align*}
Therefore, we have verified \eqref{lip1} with $m_1(\y)=\sqrt{3}[\{4\lambda_j(1-\lambda_j)\}^{-1} + \| \bOmega_j(\y-\bmu_1) \|_2/4 + \| \bOmega_j(\y-\bmu_0) \|_2/4 ]=c_1 + c_2 \|\z \|_2 $ where $\z\sim N_d(\0,\I_d)$. % and $\|\z\|_2$ is sub-Gaussian distributed. 

Let's then handle \eqref{lip2}. By plugging in the form of $w^j_{\btheta_j}(\y)$ and $t(\y,\bEta_j)$, we have $w^j_{\btheta_j}(\y)t(\y,\bEta_j)=w^1_{\btheta_1}(\y)\lambda_j/\lambda_1$ and
\begin{align*}
	\frac{\partial w^j_{\btheta_j}(\y)t(\y,\bEta_j)}{\partial \lambda_1} & = \frac{\lambda_j}{\lambda_1}\frac{\partial w^1_{\btheta_1}(\y)}{\partial \lambda_1} - \frac{\lambda_j}{\lambda_1^2}w^1_{\btheta_1}(\y) \\
	 \frac{\partial w^j_{\btheta_j}(\y)t(\y,\bEta_j)}{\partial \lambda_j} & =\frac{1}{\lambda_1} w^1_{\btheta_1}(\y) \\
	 \frac{\partial w^j_{\btheta_j}(\y)}{\partial \bmu_k}&= \frac{\lambda_j}{\lambda_1} \lambda_1(1-\lambda_1)\frac{\partial w^1_{\btheta_1}(\y)}{\partial \lambda_1}\bOmega_1(\y-\bmu_k)(-1)^{1(k=1)},\ k=0,1.
	 \end{align*}
Let $\Delta_{\bEta_j}=\bEta_j - \bEta_j'$, and $\bEta_j^u = \bEta_j' + u \Delta_{\bEta_j}$ with $u \in (0,1)$, then 
\begin{align*}
	& w^j_{\btheta_j}(\y)t(\y,\bEta_j) - w^j_{\btheta_j'}(\y)t(\y,\bEta_j') \\
	& = \int_0^1 \left\langle \frac{\partial w^j_{\btheta_j}(\y)t(\y,\bEta_j)}{\partial \bEta_j}|_{\bEta_j=\bEta_j^u}, \Delta_{\bEta_j}  \right\rangle du \\
	& = \int_0^1 \left\langle \frac{\partial w^j_{\btheta_j}(\y)}{\partial \lambda_j}|_{\bEta_j=\bEta_j^u}, \Delta_{\lambda_j}  \right\rangle du + \int_0^1 \left\langle \frac{\partial w^j_{\btheta_j}(\y)}{\partial \lambda_j}|_{\bEta_j=\bEta_j^u}, \Delta_{\lambda_j}  \right\rangle du \\ & + \sum_{k=0}^1 \int_0^1 \left\langle \frac{\partial w^j_{\btheta_j}(\y)}{\partial \bmu_k}|_{\bEta_j=\bEta_j^u}, \Delta_{\bmu_k}  \right\rangle du.
\end{align*}
For $\btheta \in  U_{\btheta^*}(\rho)$, we have the following results:
\begin{align*}
	& \left| \frac{\partial w^j_{\btheta_j}(\y)t(\y,\bEta_j)}{\partial \lambda_1} \right| \leq \frac{\lambda_j}{\lambda_1} \left| \frac{\partial w^j_{\btheta_j}(\y)}{\partial \lambda_1} \right| + \frac{\lambda_j}{\lambda_1^2}  \leq \frac{\lambda_j}{4\lambda_1^2(1-\lambda_j)} + \frac{\lambda_j}{\lambda_1^2} := c_1, \\
	& \left|\frac{\partial w^j_{\btheta_j}(\y)t(\y,\bEta_j)}{\partial \lambda_j} \right|=\frac{1}{\lambda_1} w^1_{\btheta_1}(\y) \leq \frac{1}{\lambda_1} := c_2, \\
	& \left\| \frac{\partial w^j_{\btheta_j}(\y)}{\partial \bmu_k}\right\|_2= \frac{\lambda_j}{\lambda_1} \lambda_1(1-\lambda_1)\frac{\partial w^1_{\btheta_1}(\y)}{\partial \lambda_1}\bOmega_1(\y-\bmu_k) \leq \frac{\lambda_j}{4\lambda_1} \| \bOmega_1(\y-\bmu_k) \|_2:=m_{2k}(\y). 
\end{align*}
Therefore, with $d_2(\bEta_j,\bEta_j')=|\lambda_j-\lambda_j'|+|\lambda_1 - \lambda_1'| + \|\bmu_0 - \bmu_0'\|_2 + \|\bmu_1 - \bmu_1'\|_2$, we have
\begin{align*}
	|w^j_{\btheta_j}(\y)t(\y,\bEta_j) - w^j_{\btheta_j'}(\y)t(\y,\bEta_j')  & \leq c_1 |\lambda_1 - \lambda_1'| + c_2 |\lambda_j - \lambda_j'| + \sum_{k=0}^1 m_{2k}(\y) \| \bmu_k -\bmu_k' \|_2 \\ 
	& \leq (c_1 + c_2 + m_{20}(\y) + m_{21}(\y))d_2(\bEta_j,\bEta_j') \\ 
	& \leq \sqrt{4}(c_1 + c_2 + m_{20}(\y) + m_{21}(\y))\|\bEta_j-\bEta_j'\|_2 \\ 
	& := m_{23}(\y)\|\bEta_j-\bEta_j'\|_2.
\end{align*}
	%where $m_{23}(\y)$ is sub-Gaussian. 
	Follow the same technique, we can also obtain
	\begin{align*}
		|t(\y,\bEta_j) - t(\y,\bEta_j') | \leq m_{24}(\y)\|\bEta_j-\bEta_j'\|_2, 
	\end{align*}
	%with $m_{24}(\y)$ be sub-Gaussian. 
	and it completes the verification of \eqref{lip2} with $m_2(\y)=m_{23}(\y)+m_{24}(\y)$ which can be written as $c_1 + c_2 \|\z\|_2$ with $\z \sim N_d(\0,\I_d)$.  
	
As for the derivative that also involves $\btheta_j$, we have
\begin{align*}
	\|\nabla^2_{\bmu\btheta_j} h(\y;\bmu,\btheta_j') - \nabla^2_{\bmu\btheta_j} h(\y;\bar\bmu,\btheta_j'')\|_2 \leq \|\nabla^2_{\bmu\btheta_j} & h(\y;\bmu,\btheta_j')  - \nabla^2_{\bmu\btheta_j} h(\y;\bar\bmu,\btheta_j')\|_2 + \\ 
	& \|\nabla^2_{\bmu\btheta_j} h(\y;\bar\bmu,\btheta_j') - \nabla^2_{\bmu\btheta_j} h(\y;\bar\bmu,\btheta_j'')\|_2,
\end{align*}
	and we want to show that 
	\begin{align}
		& \|\nabla^2_{\bmu\btheta_j} h(\y;\bmu,\btheta_j') - \nabla^2_{\bmu\btheta_j} h(\y;\bar\bmu,\btheta_j')\|_2 \leq m_{41}(\y)\| \bmu -\bar\bmu \|_2, \label{derive1}\\
		& \|\nabla^2_{\bmu\btheta_j} h(\y;\bar\bmu,\btheta_j') - \nabla^2_{\bmu\btheta_j} h(\y;\bar\bmu,\btheta_j'')\|_2 \leq m_{42}(\y)\|\btheta_j' - \btheta_j'' \|_2 \label{derive2}
	\end{align}
	where $m_4(\y)= m_{41}(\y)\vee m_{42}(\y)$ satisfies some moment conditions as stated in Assumption \ref{assump:smooth}. For \eqref{derive1}, using some results we obtained before, we have
	\begin{align*}
		\|\nabla^2_{\bmu\btheta_j} h(\y;\bmu,\btheta_j') & - \nabla^2_{\bmu\btheta_j} h(\y;\bar\bmu,\btheta_j')\|_2 =   
\left\| \begin{pmatrix}
\bSigma_j^{-1}(\bmu_1-\bar\bmu_1) \otimes (\frac{\partial w^j_{\btheta_j'}(\y)}{\partial \btheta_j'})^T  \\
\bSigma_j^{-1}(\bmu_0-\bar\bmu_0) \otimes (\frac{(1-\partial w^j_{\btheta_j'}(\y))}{\partial \btheta_j'})^T
\end{pmatrix}\right\|_2 \\
&\leq \| \bSigma_j^{-1}(\bmu_1-\bar\bmu_1) \otimes (\frac{\partial w^j_{\btheta_j'}(\y)}{\partial \btheta_j'})^T \|_2 + \|\bSigma_j^{-1}(\bmu_0-\bar\bmu_0) \otimes (\frac{(1-\partial w^j_{\btheta_j'}(\y))}{\partial \btheta_j'})^T \|_2 \\
& \leq \| \bSigma_j^{-1}(\bmu_1-\bar\bmu_1)\|_2 \| (\frac{\partial w^j_{\btheta_j'}(\y)}{\partial \btheta_j'})^T \|_2 + \|\bSigma_j^{-1}(\bmu_0-\bar\bmu_0)\|_2 \| (\frac{(1-\partial w^j_{\btheta_j'}(\y))}{\partial \btheta_j'})^T \|_2 \\
& \leq (c_1 + c_2 \|\z\|_2)\| \bmu- \bmu' \|_2.
	\end{align*}  
Similarly, for \eqref{derive2}, we have
\begin{align*}
	\|\nabla^2_{\bmu\btheta_j} h(\y;\bar\bmu,\btheta_j') - \nabla^2_{\bmu\btheta_j} h(\y;\bar\bmu,\btheta_j'')\|_2 \leq (c_1\|\z\|_2^3 + c_2 \|\z\|_2^2 + c_3\|\z\|_2 + c_4)\| \btheta_j'- \btheta_j'' \|_2.
\end{align*}
Therefore, with the above results, we have
\begin{align*}
	\|\nabla^2_{\bmu\btheta_j} h(\y;\bmu,\btheta_j') - \nabla^2_{\bmu\btheta_j} h(\y;\bar\bmu,\btheta_j'')\|_2 \leq (c_1\|\z\|_2^3 + c_2 \|\z\|_2^2 + c_3\|\z\|_2 + c_4)(\|\bmu -\bar\bmu\|_2 + \| \btheta_j'- \btheta_j'' \|_2)
\end{align*}
	with $m_4(\y)=c_1\|\z\|_2^3 + c_2 \|\z\|_2^2 + c_3\|\z\|_2 + c_4$. Notice that all $m_k(\y)$, $k=1,2,3,4$ can be written as functions of $\|\z\|_2$ which follows a $\chi_d$ distribution. It can be easily verified that the moment conditions in Assumption \ref{assump:smooth} can be satisfied.

\subsection{Moment bounds}
Recall the form of Hessian matrices, we have 
\begin{align*}
	\nabla^2_{\bmu\bmu} h(\y;\bmu^*,\btheta_j^*) &=   
\begin{pmatrix}
-w^j_{\btheta_j^*}(\y) & 0 \\
0 & (w^j_{\btheta_j^*}(\y)-1)
\end{pmatrix}\otimes \bSigma_j^{-1} 
\end{align*}
and 
	\begin{align*}
	E\nabla^2_{\bmu\bmu} h(\y;\bmu^*,\btheta_j^*) &=   
\begin{pmatrix}
-\lambda_j^* & 0 \\
0 & -(1-\lambda_j^*)
\end{pmatrix}\otimes \bSigma_j^{-1}. 
\end{align*}
Since $w^j_{\btheta_j^*}(\y) \in [0,1]$ and all other terms are also bounded when $\btheta \in U_{\btheta^*}(\rho)$, we conclude that $ E( \| \nabla^2_{\bmu\bmu} h(\y;\bmu^*,\btheta_j^*) - E \nabla^2_{\bmu\bmu} h(\y;\bmu^*,\btheta_j^*)  \|_2^8 ) $ is bounded. Similarly, since all the terms in
\begin{align*}
	\nabla^2_{\bmu\bmu} h(\y;\bmu^*,\btheta_j^*)t(\y,\bEta_j^*) &=   
\begin{pmatrix}
- \frac{\lambda_j^*}{\lambda_1^*} w^1_{\btheta_1^*}(\y) & 0 \\
0 & \frac{1-\lambda_j^*}{1-\lambda_1^*}(w^1_{\btheta_1^*}(\y)-1)
\end{pmatrix}\otimes \bSigma_j^{-1} 
\end{align*}
and 
\begin{align*}
	E\nabla^2_{\bmu\bmu} h(\y;\bmu^*,\btheta_j^*)t(\y,\bEta^*_j) &=   
\begin{pmatrix}
-\lambda_j^* & 0 \\
0 & -(1-\lambda_j^*)
\end{pmatrix}\otimes \bSigma_j^{-1} 
\end{align*}
are bounded, we have $ E( \| \nabla^2_{\bmu\bmu} h(\y;\bmu^*,\btheta_j^*)t(\y,\bEta^*_j) - E \nabla^2_{\bmu\bmu} h(\y;\bmu^*,\btheta_j^*)t(\y,\bEta^*_j)  \|_2^8 ) $ is bounded.

Then, as for the gradient, we have
\begin{align*}
	\|\nabla_{\bmu} h(\y;\bmu^*,\btheta_j^*) \|_2&=  \left\|
\begin{pmatrix}
w^j_{\btheta_j*}(\y)\bSigma_j^{-1}(\y-\bmu_1^*)  \\
(1-w^j_{\btheta_j^*}(\y))\bSigma_j^{-1}(\y-\bmu_0^*)
\end{pmatrix} \right\|_2 \\ 
& \leq \| \bSigma_j^{-1}(\y-\bmu_1^*) \|_2 + \| \bSigma_j^{-1}(\y-\bmu_0^*) \|_2 \\
& \leq c_1 + c_2\|\z\|_2
\end{align*}
with $\z \sim N_d(\0,\I_d)$. Therefore, it can be verified that $E(\|\nabla_{\bmu} h(\y;\bmu^*,\btheta_j^*) \|_2^8)$ is bounded. 

Finally, since 
\begin{align*}
	\nabla^2_{\bmu\btheta_j} h(\y;\bmu,\btheta_j) =  
\begin{pmatrix}
\bSigma_j^{-1}(\y-\bmu_1) \otimes (\frac{\partial w^j_{\btheta_j}(\y)}{\partial \btheta_j})^T  \\
\bSigma_j^{-1}(\y-\bmu_0) \otimes (\frac{(1-\partial w^j_{\btheta_j}(\y))}{\partial \btheta_j})^T
\end{pmatrix},
\end{align*} 
we have 
\begin{align*}
	\|\nabla^2_{\bmu\btheta_j} h(\y;\bmu,\btheta_j)\|_2 & \leq \|\bSigma_j^{-1}(\y-\bmu_1)\|_2 \|(\frac{\partial w^j_{\btheta_j}(\y)}{\partial \btheta_j})\|_2 +  \|\bSigma_j^{-1}(\y-\bmu_0) \|_2 \|(\frac{(1-\partial w^j_{\btheta_j}(\y))}{\partial \btheta_j})^T \|_2 \\
	& \leq c_1\|\z\|_2^2 + c_2\|\z\|_2 + c_3.
\end{align*}
Thus, the moment condition can be verified as before.

\end{proof}

%We consider the theoretical performance of the EM algorithm on the pooled data over the parameter space $$\btheta=\{\btheta=(\bmu_0,\bmu_1,\lambda_1,\ldots,\lambda_K):\forall j\in[K],\lambda_j\in(c_w,1-c_w),\bmu_0,\bmu_1\in \mathbb{R}^d \}$$ with $0<c_w<1$.

\section{Proof of Lemma S.3}
\begin{proof}%[Proof of Lemma S.3]
 Lemma S.3 is about the contraction on the population iteration, and we need to use the population updating formulas \eqref{eqpop1}-\eqref{eqpop3}. Here we divide the whole proof into several parts.

%\begin{lem2}[Contraction for the mixing proportion] 
%	With $\kappa_1 = c_3 \exp(-c_4\Delta_{min}^2) \vee C_{\bmu} \exp(-c_4\Delta_{min}^2) $ where $C_{\bmu} = c_3 (\sqrt{M}\Delta_{max}/8 + \sqrt{M}\Delta_{max}(M^2/2+\sqrt{1+c_1})/8 + \frac{\sqrt{M}}{4\sqrt{1-c_1}\Delta_{min}})$, then we have
%$$ |E(\gamma_{\btheta}(Y_j))-E(\gamma_{\btheta^*}(Y_j))| \leq \kappa_1 d_2(\btheta_j, \btheta_j^*) \leq \kappa_1 d_2(\btheta, \btheta^*). $$
%\end{lemma}
%
%
%\begin{lemma}[Contraction for the mean parameter]
%With $$ \kappa_{2} :=  [ \{ \frac{1}{\sqrt{K}} ( M c_3c_{\A\B} +  c_6 )\} \vee \{ MC_{\bmu} c_{\A\B} +  c_{\bmu}\} ] \exp(-c_4\Delta_{min}^2) $$
%where , we have 
%\begin{align*}
%	\| \bmu_1(\btheta) - \bmu_1(\btheta^*)  \|_2 
%	& \leq \kappa_{2}d_2(\btheta,\btheta^*). 
%\end{align*}
%We can also get 
%\begin{align*}
%	\| \bmu_1(\btheta) - \bmu_1(\btheta^*)  \|_2  
%	& \leq \kappa_{3} \frac{1}{K} \sum_{j=1}^K d_2(\btheta_j,\btheta_j^*)
%\end{align*}
%with $$ \kappa_{3} :=  [ \{M c_3c_{\A\B} +  c_6\} \vee \{ MC_{\bmu} c_{\A\B} +  c_{\bmu}\} ] \exp(-c_4\Delta_{min}^2). $$
%\end{lemma}

\subsection*{Goal and Self-consistency}
	Our goal is to find a $\kappa\in(0,1)$ to make
	\begin{align*}
		&|\lambda_j(\btheta)-\lambda_j^*|\leq \kappa d_2(\btheta,\btheta^*),\ j\in[K] \\
		&\|\bmu_0(\btheta)-\bmu_0^*\|_2 \leq \kappa d_2(\btheta,\btheta^*), \\
		&\|\bmu_1(\btheta)-\bmu_1^*\|_2 \leq \kappa d_2(\btheta,\btheta^*).
\end{align*}
Firstly, let's verify the self-consistency property $M(\btheta^*)=\btheta^*$. For each $j\in[K]$, we have  
\begin{align*}
 \lambda_{j}(\btheta^*)&=E\gamma_{\btheta^*}(Y_j)=P_{\btheta^*}(Z_j=1|Y_j)=\lambda_j^*,\ j\in [K] \\
 \bmu_0(\btheta^*)&= [\frac{1}{K}\sum_{j=1}^K (1-E\gamma_{\btheta^*}(Y_j)) \bOmega_j ]^{-1} \frac{1}{K}\sum_{j=1}^K E[(1-\gamma_{\btheta^*}(Y_j)) \bOmega_j Y_j] \\
 &= [\frac{1}{K}\sum_{j=1}^K (1-E\gamma_{\btheta^*}(Y_j)) \bOmega_j ]^{-1} \frac{1}{K}\sum_{j=1}^K E[P_{\btheta^*}(Z_j=0|Y_j) \bOmega_j Y_j] \\
 &= [\frac{1}{K}\sum_{j=1}^K (1-E\gamma_{\btheta^*}(Y_j)) \bOmega_j ]^{-1} \frac{1}{K}\sum_{j=1}^K E[E[1(Z_j=0)|Y_j] \bOmega_j Y_j] \\
 &= [\frac{1}{K}\sum_{j=1}^K (1-E\gamma_{\btheta^*}(Y_j)) \bOmega_j ]^{-1} \frac{1}{K}\sum_{j=1}^K E[E[\bOmega_j Y_j1(Z_j=0)|Y_j]] \\
 &= [\frac{1}{K}\sum_{j=1}^K (1-E\gamma_{\btheta^*}(Y_j)) \bOmega_j ]^{-1} \frac{1}{K}\sum_{j=1}^K E[\bOmega_j Y_j1(Z_j=0)] \\
 &= [\frac{1}{K}\sum_{j=1}^K (1-E\gamma_{\btheta^*}(Y_j)) \bOmega_j ]^{-1} \frac{1}{K}\sum_{j=1}^K \bOmega_j E[ 1-\gamma_{\btheta^*}(Y_j)]\bmu_0^* = \bmu_0^*.
\end{align*}
Similarly, we can verify $\bmu_1(\btheta^*)=\bmu_1^*$. Thus, we only need to prove 
\begin{align}
		&|\lambda_j(\btheta)-\lambda_j(\btheta^*)|\leq \kappa d_2(\btheta,\btheta^*),\ j\in[K] \label{eqpop4} \\
		&\|\bmu_0(\btheta)-\bmu_0(\btheta^*)\|_2 \leq \kappa d_2(\btheta,\btheta^*), \label{eqpop5}\\
		&\|\bmu_1(\btheta)-\bmu_1(\btheta^*)\|_2 \leq \kappa d_2(\btheta,\btheta^*).\label{eqpop6}
	\end{align}

\subsection*{Some tools from Cai et al. (2019)}\label{sectool}
We need some tools to help derive the above bounds. For the reader's convenience, we take the following results from the supplementary material C.1.2 of \cite{cai2019chime}. For functions $ f_1(t)=\frac{1}{[we^t+(1-w)e^{-t}]^2}$, $f_2(t)=\frac{t}{[we^t+(1-w)e^{-t}]^2}$ and $f_3(t)=\frac{t^2-b^2}{[we^t+(1-w)e^{-t}]^2}$, we have 
\begin{align*}
	&f_1(t)\leq\frac{1}{4w(1-w)}\leq\frac{1}{4\min\{w,1-w\}},\  \text{for all}\ t\in\mathbb{R}, \\
	&\sup_{t\in[a,\infty]}f_1(t)\leq \frac{1}{\min\{w,1-w\}^2}\exp(-2a),\  \text{for all}\ a\geq 0, \\
	&|f_2(t)|\leq\frac{|t|e^{-|t|}}{\min\{w,1-w\}^2}\leq\frac{1}{4\min\{w^2,(1-w)^2\}},\  \text{for all}\ t\in\mathbb{R}, \\
	&\sup_{t\in[a,\infty]}f_2(t)\leq \frac{1}{\min\{w,1-w\}^2}\exp(-3a/2),\  \text{for all}\ a\geq 0, \\
	&|f_3(t)|\leq\frac{|t^2-b^2|e^{-|t|}}{\min\{w,1-w\}^2}\leq\frac{1+b^2}{\min\{w^2,(1-w)^2\}},\  \text{for all}\ t\in\mathbb{R}, \\
	&\sup_{t\in[a,\infty]}f_3(t)\leq \frac{1+b^2}{\min\{w,1-w\}^2}\exp(-a),\  \text{for all}\ a\geq 0.
\end{align*}

\subsection*{Taylor expansion of $\lambda_j(\btheta)$ and $\bmu_1(\btheta)$}
We need to verify
\begin{align}
		|\lambda_j(\btheta)-\lambda_j(\btheta^*)| &= |E(\gamma_{\btheta}(Y_j))-E(\gamma_{\btheta^*}(Y_j))| \leq \kappa_1 d_2(\btheta,\btheta^*),\label{eqpop7} \\
		\|\bmu_1(\btheta)-\bmu_1(\btheta^*)\|_2 &=\|[\frac{1}{K}\sum_{j=1}^K E\gamma_{\btheta}(Y_j) \bOmega_j ]^{-1} \frac{1}{K}\sum_{j=1}^K E[\gamma_{\btheta}(Y_j) \bOmega_j Y_j]\nonumber \\
		&- [\frac{1}{K}\sum_{j=1}^K E\gamma_{\btheta^*}(Y_j) \bOmega_j ]^{-1} \frac{1}{K}\sum_{j=1}^K E[\gamma_{\btheta^*}(Y_j) \bOmega_j Y_j]\|_2\nonumber \\ 
		& \leq \kappa_2 d_2(\btheta,\btheta^*).\label{eqpop8}
	\end{align}
	where the two constants $\kappa_1$ and $\kappa_2$ are to be determined. 
Let $\Delta_{\btheta}=\btheta-\btheta^*$, and $\btheta_u=\btheta^*+u\Delta_{\btheta}$ with $u\in(0,1)$. Then we have
\begin{align}
	E(\gamma_{\btheta}(Y_j))-E(\gamma_{\btheta^*}(Y_j)) &= E[\int_0^1 \langle \frac{d \gamma_{\btheta}(Y_j)}{d \btheta}|_{\btheta=\btheta_u},\Delta_{\btheta} \rangle du ] \nonumber \\
	& = E[ \int_0^1 \langle \frac{\partial \gamma_{\btheta}(Y_j)}{\partial \lambda_j}|_{\btheta=\btheta_u},\Delta_{\lambda_j} \rangle du  ] + \sum_{k=0}^1 E[ \int_0^1 \langle \frac{\partial \gamma_{\btheta}(Y_j)}{\partial \bmu_k}|_{\btheta=\btheta_u},\Delta_{\bmu_k} \rangle du  ] \label{eqlambj}
\end{align}
and 
\begin{align}
	\bmu_1(\btheta)-\bmu_1(\btheta^*) &= \int_0^1 (\frac{d \bmu_1(\btheta)}{d \btheta}|_{\btheta=\btheta_u})\Delta_{\btheta} du  \nonumber \\
	& = \sum_{j=1}^K\int_0^1 (\frac{\partial \bmu_1(\btheta)}{\partial \lambda_j}|_{\btheta=\btheta_u})\Delta_{\lambda_j} du  + \sum_{k=0}^1 \int_0^1  (\frac{\partial \bmu_1(\btheta)}{\partial \bmu_k}|_{\btheta=\btheta_u})\Delta_{\bmu_k}  du \label{eqlammu}
\end{align}
where $\Delta_{\lambda_j}=\lambda_j-\lambda_j^*$ and $\Delta_{\bmu_k}=\bmu_k-\bmu_k^*$.
Let's first deal with \eqref{eqlambj} whose key part is $\partial \gamma_{\btheta}(Y_j)$. Recall the form of $\gamma_{\btheta}(Y_j)$ is 
\begin{align*}
	\gamma_{\btheta}(Y_j) & =\frac{\lambda_j}{\lambda_j + (1-\lambda_j)\exp\{(\bmu_0-\bmu_1)^T\bSigma_j^{-1}(Y_j-\frac{\bmu_0+\bmu_1}{2})\}} \\
	&= \frac{\lambda_j}{\lambda_j + (1-\lambda_j)\exp\{\bbeta_j^T(Y_j-\frac{\bmu_0+\bmu_1}{2})\}} 
\end{align*}
and the partial derivatives of $\gamma_{\btheta}(Y_j)$ with respect to each parameter in $\btheta_j$ are 
\begin{align}
	\frac{\partial \gamma_{\btheta}(Y_j)}{\partial \lambda_j} &= \frac{\exp\{\bbeta_j^T(Y_j-\frac{\bmu_0+\bmu_1}{2})\}}{[\lambda_j + (1-\lambda_j)\exp\{\bbeta_j^T(Y_j-\frac{\bmu_0+\bmu_1}{2})\}]^2},\ j\in [K] \label{eqdlam} \\
	\frac{\partial \gamma_{\btheta}(Y_j)}{\partial \bmu_k}&=\frac{\lambda_j(1-\lambda_j)\exp\{\bbeta_j^T(Y_j-\frac{\bmu_0+\bmu_1}{2})\}}{[\lambda_j + (1-\lambda_j)\exp\{\bbeta_j^T(Y_j-\frac{\bmu_0+\bmu_1}{2})\}]^2}\bOmega_j(Y_j-\bmu_k)(-1)^{1(k=1)} \nonumber\\
	& = \lambda_j(1-\lambda_j)\frac{\partial \gamma_{\btheta}(Y_j)}{\partial \lambda_j}\bOmega_j(Y_j-\bmu_k)(-1)^{1(k=1)},\ k=0,1. \label{eqdmu}
\end{align}
As the term $\bbeta_j^T(Y_j-\frac{\bmu_0+\bmu_1}{2})$ appears in both \eqref{eqdlam} and \eqref{eqdmu}, next we write it as a one-dimensional normal random variable and then to obtain probabilistic bounds of the expectations of \eqref{eqdlam} and \eqref{eqdmu}. Let $\tilde Y_j=\bOmega_j^{1/2}\{Y_j - \frac{\bmu_0^* + \bmu_1^*}{2} \} $, then 
$$ \tilde Y_j \sim (1-\lambda_j^*)N_d(\bOmega_j^{1/2}\frac{\bmu_0^* - \bmu_1^*}{2}, \I_d ) + \lambda_j^*N_d(\bOmega_j^{1/2}\frac{\bmu_1^* - \bmu_0^*}{2}, \I_d ):= \Psi_j + Z_N  $$
where $ \Psi_j \sim (1-\lambda_j^*)\bOmega_j^{1/2}\frac{\bmu_0^* - \bmu_1^*}{2} +  \lambda_j^*\bOmega_j^{1/2}\frac{\bmu_1^* - \bmu_0^*}{2}$ and $Z_n \sim N_d(\0,\I_d)$. Also, we have $Y_j=\bSigma_j^{1/2}\tilde Y_j + (\bmu_0^* + \bmu_1^*)/2 $. For simplicity, let's adopt the following notations:
\begin{align*}
	& \Delta_{\bmu}=(\bmu_0 + \bmu_1 - \bmu_0^* - \bmu_1^*)/2, \\ %\Delta_{\bbeta_j}=\bbeta_j-\bbeta_j^*,\ 
	&\delta_0(\bbeta_j)=\bbeta_j^T(\bmu_0^* - \frac{\bmu_0 + \bmu_1}{2}), \ \delta_1(\bbeta_j)=\bbeta_j^T(\bmu_1^* - \frac{\bmu_0 + \bmu_1}{2}),\ \sigma(\bbeta_j)=\sqrt{\bbeta_j^T\bSigma_j\bbeta_j}, \\
	& \delta_{\bbeta_j}=\bbeta_j^T\bSigma_j^{1/2}\Psi_j - \bbeta_j^T\Delta_{\bmu}\ \text{with}\ P(\delta_{\bbeta_j}=\delta_0(\bbeta_j))=1-\lambda_j^*=1-P(\delta_{\bbeta_j}=\delta_1(\bbeta_j)).
\end{align*} 
Then, we have
\begin{align*}
	\bbeta_j^T(Y_j - \frac{\bmu_0 + \bmu_1}{2}) \stackrel{d}{=} \delta_{\bbeta_j} + \sigma(\bbeta_j)Z_{N1} 
\end{align*}
where $Z_{N1}\sim N(0,1)$. 

\subsection*{Contraction for the mixing proportion}
Follow the same reasoning as the proof of Lemma 3.1 of \cite{cai2019chime} (see the derivation of (C.12) in the supplemental material of CHIME), using the results in \ref{sectool} we can obtain 
\begin{align}
	E[\frac{\partial \gamma_{\btheta}(Y_j)}{\partial \lambda_j}] \leq c_3 \exp(-c_4 \Delta_j^2) \label{eqwj}
\end{align}
where $c_3:=2/c_0^2$ and $c_4=\frac{1-c_1}{2} \wedge \frac{(1-c_1)^2}{8(1+c_1)}$ with $c_0,c_1$ defined in $\B(\btheta^*;c_0,c_1)$. Next we deal with $|\langle \frac{\partial}{\partial \bmu_1}E \gamma_{\btheta}(Y_j),\Delta\bmu_1 \rangle|$. We have
\begin{align}
	& |\langle \frac{\partial}{\partial \bmu_1}E \gamma_{\btheta}(Y_j),\Delta\bmu_1 \rangle| \nonumber \\ 
	& =|\langle \lambda_j(1-\lambda_j)E\frac{\partial \gamma_{\btheta}(Y_j)}{\partial \lambda_j}\bOmega_j(\bmu_1-Y_j), \Delta_{\bmu_1} \rangle  | \nonumber \\
	& \leq |\lambda_j(1-\lambda_j)|\cdot |\langle E\frac{\partial \gamma_{\btheta}(Y_j)}{\partial \lambda_j}\bOmega_j(\bSigma_j^{1/2}\tilde\Y_j + (\bmu_0^*+\bmu_1^*)/2- \bmu_1), \Delta_{\bmu_1} \rangle  | \nonumber \\
	& = |\lambda_j(1-\lambda_j)|\cdot |\langle E\frac{\partial \gamma_{\btheta}(Y_j)}{\partial \lambda_j}\bOmega_j(\bSigma_j^{1/2}(\Psi_j + Z_N) + (\bmu_0^*+\bmu_1^*)/2- \bmu_1), \Delta_{\bmu_1} \rangle  | \nonumber \\
	& \leq \frac{1}{4}\{ |\langle E\frac{\partial \gamma_{\btheta}(Y_j)}{\partial \lambda_j}\bOmega_j^{1/2}\Psi_j, \Delta_{\bmu_1} \rangle  | \label{eqLamConc1} \\ 
	& + |\langle E\frac{\partial \gamma_{\btheta}(Y_j)}{\partial \lambda_j}\bOmega_j^{1/2}Z_N, \Delta_{\bmu_1} \rangle  | \label{eqLamConc2} \\ 
	& + |\langle E\frac{\partial \gamma_{\btheta}(Y_j)}{\partial \lambda_j}\bOmega_j(\frac{\bmu_0^*+\bmu_1^*}{2}- \bmu_1), \Delta_{\bmu_1} \rangle  | \}. \label{eqLamConc3}
\end{align} 

For \eqref{eqLamConc1}, we have
\begin{align*}
	|\langle E\frac{\partial \gamma_{\btheta}(Y_j)}{\partial \lambda_j}\bOmega_j^{1/2}\Psi_j, \Delta_{\bmu_1} \rangle | &\leq  c_3\exp(-c_4\Delta_j^2)\|\bOmega_j^{1/2}\Psi_j\|_2\|\Delta_{\bmu_1}\|_2 \\
	& \leq \frac{c_3}{2}\sqrt{M}\Delta_j\exp(-c_4\Delta_j^2)\|\Delta_{\bmu_1}\|_2. 
\end{align*}
where the last inequality is due to $\Psi_j \sim (1-\lambda_j^*)\bOmega_j^{1/2}\frac{\bmu_0^* - \bmu_1^*}{2} +  \lambda_j^*\bOmega_j^{1/2}\frac{\bmu_1^* - \bmu_0^*}{2}$, thus $\|\Psi_j\|_2 \leq \| \bOmega_j^{1/2}\frac{\bmu_0^*-\bmu_1^*}{2} \|_2$.

For \eqref{eqLamConc3}, we have
\begin{align*}
	|\langle E\frac{\partial \gamma_{\btheta}(Y_j)}{\partial \lambda_j}\bOmega_j(\frac{\bmu_0^*+\bmu_1^*}{2}- \bmu_1), \Delta_{\bmu_1} \rangle  | 
	&\leq c_3\exp(-c_4\Delta_j^2)| \langle \bOmega_j(\frac{\bmu_0^*+\bmu_1^*}{2}- \bmu_1), \Delta_{\bmu_1} \rangle   | \\ 
	& \leq c_3\exp(-c_4\Delta_j^2)\|\bOmega_j(\frac{\bmu_0^*+\bmu_1^*}{2}- \bmu_1)\|_2\|\Delta_{\bmu_1}\|_2 \\
	& \leq c_3\exp(-c_4\Delta_j^2)\sqrt{M}\Delta_j(M^2/2+\sqrt{1+c_1})\|\Delta_{\bmu_1}\|_2/2
\end{align*}
by the fact that $\btheta \in \B(\btheta^*;c_0,c_1)$.

Finally we work on \eqref{eqLamConc2}. Let $\alpha_j=\bSigma_j^{1/2}\bbeta_j$, $\H$ be an orthogonal matrix with first row be $\alpha_j^T/\|\alpha_j\|_2$. Then it follows that
$$ \H\alpha_j = \|\alpha_j\|_2\e_1=\sigma(\bbeta_j)\e_1 $$
where $\e_1$ is the first canonical basis vector and $$E[\frac{\partial \gamma_{\btheta}(Y_j)}{\partial \lambda_j}\bSigma^{-1/2}_j Z_N]= \bSigma^{-1/2}_j \H^T E[\frac{\partial \gamma_{\btheta}(Y_j)}{\partial \lambda_j}\H Z_N].$$ Then, we have 
\begin{align*}
	E[\frac{\partial \gamma_{\btheta}(Y_j)}{\partial \lambda_j}\H Z_N]
	& = E[\frac{\exp\{\bbeta_j^T(Y_j-\frac{\bmu_0+\bmu_1}{2})\}}{[\lambda_j + (1-\lambda_j)\exp\{\bbeta_j^T(Y_j-\frac{\bmu_0+\bmu_1}{2})\}]^2}\H Z_N] \\
	& = E[\frac{\exp\{\delta_{\bbeta_j} + \bbeta_j^T\bSigma_j^{1/2}Z_{N} \}}{[\lambda_j + (1-\lambda_j)\exp\{\delta_{\bbeta_j} + \bbeta_j^T\bSigma_j^{1/2}Z_{N}\}]^2}\H Z_N] \\ 
	& = E[\frac{\exp\{\delta_{\bbeta_j} + \alpha_j^T\H^T\H Z_{N} \}}{[\lambda_j + (1-\lambda_j)\exp\{\delta_{\bbeta_j} + \alpha_j^T\H^T\H Z_{N}\}]^2}\H Z_N] \\
	& = E[\frac{\exp\{\delta_{\bbeta_j} + \|\alpha_j\|_2Y_1 \}}{[\lambda_j + (1-\lambda_j)\exp\{\delta_{\bbeta_j} + \|\alpha_j\|_2Y_1\}]^2} Y]\ \text{by}\ Y=\H Z_N\sim N_d(0,\I_d) \\
	& = E[\frac{\exp\{\delta_{\bbeta_j} + \sigma(\bbeta_j)Z_{N1} \}}{[\lambda_j + (1-\lambda_j)\exp\{\delta_{\bbeta_j} + \sigma(\bbeta_j)Z_{N1}\}]^2} Z_{N1}\e_1].
\end{align*}  
Thus, 
\begin{align*}
	E[\frac{\partial \gamma_{\btheta}(Y_j)}{\partial \lambda_j}\Sigma^{-1/2}_j Z_N] 
	& = E[\frac{\exp\{\delta_{\bbeta_j} + \sigma(\bbeta_j)Z_{N1} \}}{[\lambda_j + (1-\lambda_j)\exp\{\delta_{\bbeta_j} + \sigma(\bbeta_j)Z_{N1}\}]^2} Z_{N1}]\Sigma^{-1/2}_j\H^T\e_1 \\ 
	& = E[\frac{\exp\{\delta_{\bbeta_j} + \sigma(\bbeta_j)Z_{N1} \}}{[\lambda_j + (1-\lambda_j)\exp\{\delta_{\bbeta_j} + \sigma(\bbeta_j)Z_{N1}\}]^2} Z_{N1}] \Sigma^{-1/2}_j \frac{\alpha_j}{\|\alpha_j\|_2} \\
	& = E[\frac{\exp\{\delta_{\bbeta_j} + \sigma(\bbeta_j)Z_{N1} \}}{[\lambda_j + (1-\lambda_j)\exp\{\delta_{\bbeta_j} + \sigma(\bbeta_j)Z_{N1}\}]^2} Z_{N1}] \Sigma^{-1/2}_j \frac{\alpha_j}{\sigma(\bbeta_j)} \\
	& = E[\frac{\exp\{\delta_{\bbeta_j} + \sigma(\bbeta_j)Z_{N1} \}}{[\lambda_j + (1-\lambda_j)\exp\{\delta_{\bbeta_j} + \sigma(\bbeta_j) Z_{N1}\}]^2} \sigma(\bbeta_j) Z_{N1}]\bbeta_j/\sigma^2(\bbeta_j).
\end{align*}
By writing $E[\frac{\exp\{\delta_{\bbeta_j} + \sigma(\bbeta_j)Z_{N1} \}}{[\lambda_j + (1-\lambda_j)\exp\{\delta_{\bbeta_j} + \sigma(\bbeta_j) Z_{N1}\}]^2} \sigma(\bbeta_j) Z_{N1}]$ as 
$$E[\frac{\exp\{\delta_{\bbeta_j} + \sigma(\bbeta_j)Z_{N1} \}}{[\lambda_j + (1-\lambda_j)\exp\{\delta_{\bbeta_j} + \sigma(\bbeta_j) Z_{N1}\}]^2} (\sigma(\bbeta_j) Z_{N1} + \delta_{\bbeta_j} - \delta_{\bbeta_j} )] $$ and deal with the two terms separately with the tools we introduced in \ref{sectool}, we get $$ E[\frac{\exp\{\delta_{\bbeta_j} + \sigma(\bbeta_j)Z_{N1} \}}{[\lambda_j + (1-\lambda_j)\exp\{\delta_{\bbeta_j} + \sigma(\bbeta_j) Z_{N1}\}]^2} \sigma(\bbeta_j) Z_{N1}] \leq c_3 \exp(-c_4 \Delta_j^2), $$
and it follows that
\begin{align*}
	|\langle E\frac{\partial \gamma_{\btheta}(Y_j)}{\partial \lambda_j}\bOmega_j^{1/2} Z_N , \Delta_{\bmu_1} \rangle  | & \leq c_3 \exp(-c_4\Delta_j^2) |\langle \bbeta_j,\Delta_{\bmu_1} \rangle|/\sigma^2(\bbeta_j) \\
	& \leq \frac{\sqrt{M}}{\sqrt{1-c_1}\Delta_j} c_3 \exp(-c_4\Delta_j^2) \| \Delta_{\bmu_1} \|_2
\end{align*}
as $ \|\bbeta_j\|_2 \leq \|\bOmega_j^{1/2}\|_2\| \bOmega_j^{1/2}(\bmu_0 - \bmu_1) \|_2 \leq \sqrt{M}\sigma(\bbeta_j)$ and $\sigma(\bbeta_j)\geq \sqrt{1-c_1}\Delta_j$. 

Combine the above results, we have
\begin{align*}
	|\langle \frac{\partial}{\partial \bmu_1}E \gamma_{\btheta}(Y_j),\Delta\bmu_1 \rangle| 
	& \leq  \frac{c_3}{8}\sqrt{M}\Delta_j\exp(-c_4\Delta_j^2)\|\Delta_{\bmu_1}\|_2 \\ 
	& + c_3\exp(-c_4\Delta_j^2)\sqrt{M}\Delta_j(M^2/2+\sqrt{1+c_1})\|\Delta_{\bmu_1}\|_2/8 \\
	& + \frac{\sqrt{M}}{4\sqrt{1-c_1}\Delta_j} c_3 \exp(-c_4\Delta_j^2) \| \Delta_{\bmu_1} \|_2 \\
	& \leq C^j_{\bmu}\exp(-c_4\Delta_j^2)\| \Delta_{\bmu_1} \|_2 
\end{align*}
with $ C^j_{\bmu} = c_3(\sqrt{M}\Delta_j/8 + \sqrt{M}\Delta_j(M^2/2+\sqrt{1+c_1})/8 + \frac{\sqrt{M}}{4\sqrt{1-c_1}\Delta_j})$. Similarly, we can also verify $$ |\langle \frac{\partial}{\partial \bmu_0}E \gamma_{\btheta}(Y_j),\Delta\bmu_0 \rangle| \leq C^j_{\bmu}\exp(-c_4\Delta_j^2)\| \Delta_{\bmu_0} \|_2. $$
To summarize, we have 
$$ |E(\gamma_{\btheta}(Y_j))-E(\gamma_{\btheta^*}(Y_j))| \leq c_3 \exp(-c_4\Delta_j^2) |\Delta_{\lambda_j}| + C^j_{\bmu}\exp(-c_4\Delta_j^2)(\| \Delta_{\bmu_0} \|_2+\| \Delta_{\bmu_1} \|_2). $$
If we let $\kappa_1 = c_3 \exp(-c_4\Delta_{min}^2) \vee C_{\bmu} \exp(-c_4\Delta_{min}^2) $ with $C_{\bmu} = c_3 (\sqrt{M}\Delta_{max}/8 + \sqrt{M}\Delta_{max}(M^2/2+\sqrt{1+c_1})/8 + \frac{\sqrt{M}}{4\sqrt{1-c_1}\Delta_{min}})$, then we have
$$ |E(\gamma_{\btheta}(Y_j))-E(\gamma_{\btheta^*}(Y_j))| \leq \kappa_1 d_2(\btheta_j, \btheta_j^*) \leq \kappa_1 d_2(\btheta, \btheta^*). $$
Note that, $\kappa_1$ is the $\kappa''$ in Theorem \ref{overall_contraction} of the manuscript.

\subsection*{Contraction for the mean}
Recall \eqref{eqpop3} and \eqref{eqlammu}
\begin{align*}
& \bmu_1(\btheta)=[\frac{1}{K}\sum_{j=1}^K E\gamma_{\btheta}(Y_j) \bOmega_j ]^{-1} \frac{1}{K}\sum_{j=1}^K E[\gamma_{\btheta}(Y_j) \bOmega_j Y_j] \\	
&\bmu_1(\btheta)-\bmu_1(\btheta^*)  = \sum_{j=1}^K\int_0^1 (\frac{\partial \bmu_1(\btheta)}{\partial \lambda_j}|_{\btheta=\btheta_u})\Delta_{\lambda_j} du  + \sum_{k=0}^1 \int_0^1  (\frac{\partial \bmu_1(\btheta)}{\partial \bmu_k}|_{\btheta=\btheta_u})\Delta_{\bmu_k}  du.
\end{align*}
In order to bound $\|\bmu_1(\btheta)-\bmu_1(\btheta^*)\|_2$, we need to deal with each term separately on the right hand side. For simplicity, let's denote
$$\A=\frac{1}{K}\sum_{j=1}^K E\gamma_{\btheta}(Y_j) \bOmega_j,\ \B=\frac{1}{K}\sum_{j=1}^K E[\gamma_{\btheta}(Y_j) \bOmega_j (Y_j-\bmu_1)],$$ and it follows that 
\begin{align*}
	\partial \bmu_1(\btheta) = -\A^{-1}\partial \A \A^{-1} \B + \A^{-1} \partial \B.
\end{align*}
Thus, 
\begin{align*}
    \| (\frac{\partial \bmu_1(\btheta)}{\partial \lambda_j})\Delta_{\lambda_j} \|_2  
    & \leq \|\frac{\partial \bmu_1(\btheta)}{\partial \lambda_j}\|_2 |\Delta_{\lambda_j}|
	 \\ 
	 & \leq (\|\A^{-1}\|_2 \|\frac{\partial \A}{\partial \lambda_j}\|_2 \|\A^{-1} \B\|_2 + \|\A^{-1}\|_2 \|\frac{\partial \B}{\partial \lambda_j}\|_2)|\Delta_{\lambda_j}|.
\end{align*}
We have
\begin{align*}
	\|\frac{\partial \A}{\partial \lambda_j}\|_2 & \leq \frac{1}{K} \| \frac{\partial}{\partial \lambda_j} E\gamma_{\btheta}(Y_j) \bOmega_j\|_2 \\ 
	& \leq \frac{1}{K} |\frac{\partial}{\partial \lambda_j} E\gamma_{\btheta}(Y_j) | \| \bOmega_j\|_2 \\
	& \leq c_3\exp(-c_4\Delta_j^2)\frac{1}{K} \|\bOmega_j\|_2 \\
	& \leq \frac{M}{K} c_3\exp(-c_4\Delta_j^2).
\end{align*}
{As for $ \|\A^{-1}\|_2 $, due to the fact that $\bOmega_j$ are positive definite and $0<E\gamma_{\btheta}(Y_j)<1$ for all $j\in[K]$, $\A$ is also positive definite. Thus, $\lambda_{min}(\A) \neq 0$ and $\|\A^{-1}\|_2=\lambda_{min}(\A)^{-1}$ is upper bounded where $\lambda_{min}(\A)$ is the smallest eigenvalue of $\A$.} Next, we deal with
\begin{align*}
	\|\frac{\partial \B}{\partial \lambda_j}\|_2 \leq \frac{1}{K} \|E \frac{\partial \gamma_{\btheta}(Y_j)}{\partial \lambda_j}\bOmega_j(Y_j - \bmu_1) \|_2
\end{align*}
where $Y_j= \bSigma_j^{1/2} Z_N + \bSigma_j^{1/2} \Psi_j + (\bmu_0^* + \bmu_1^*)/2 $. For each $j$, we have
\begin{align*}
 & \|E \frac{\partial \gamma_{\btheta}(Y_j)}{\partial \lambda_j}\bOmega_j(Y_j - \bmu_1) \|_2 \\
 & \leq  \|E \frac{\partial \gamma_{\btheta}(Y_j)}{\partial \lambda_j}\bOmega_j^{1/2}Z_N \|_2  + |E \frac{\partial \gamma_{\btheta}(Y_j)}{\partial \lambda_j}|\cdot\|\bOmega_j(\bSigma_j^{1/2} \Psi_j + (\bmu_0^* + \bmu_1^*)/2 - \bmu_1) \|_2 \\
 & \leq \frac{\sqrt{M}}{\sqrt{1-c_1}\Delta_j}c_3\exp(-c_4\Delta_j^2)+c_3\exp(-c_4\Delta_j^2)\|\bOmega_j\|_2\|\bSigma_j^{1/2} \Psi_j + (\bmu_0^* + \bmu_1^*)/2 - \bmu_1\|_2. 
\end{align*}
Note that $ \bSigma_j^{1/2} \Psi_j + (\bmu_0^* + \bmu_1^*)/2 \sim (1-\lambda_j^*)\bmu_0^* + \lambda_j^*\bmu_1^*$, thus
\begin{align*}
	& \|\bSigma_j^{1/2} \Psi_j + (\bmu_0^* + \bmu_1^*)/2 - \bmu_1\|_2 \\ 
	& \leq \| \bmu_0^* - \bmu_1 \|_2 + \| \bmu_1^* - \bmu_1 \|_2 \\
	& \leq \| \bmu_0^* - \bmu_0 + \bmu_0 - \bmu_1 \|_2 + \| \bmu_1^* - \bmu_1 \|_2 \\
	& \leq \frac{M}{2}\|\bOmega_j^{-1/2}\|_2\| \bOmega_j^{1/2} ( \bmu_0^* - \bmu_1^* )\|_2 + \|\bOmega_j^{-1/2}\|_2\| \bOmega_j^{1/2}( \bmu_0 - \bmu_1 )\|_2 \\
	& \leq M^{3/2}\Delta_j/2 + \sqrt{M(1+c_1)}\Delta_j.
\end{align*}
Thus, 
\begin{align*}
	& \|E \frac{\partial \gamma_{\btheta}(Y_j)}{\partial \lambda_j}\bOmega_j(Y_j - \bmu_1) \|_2 \\ 
	& \leq \frac{\sqrt{M}}{\sqrt{1-c_1}\Delta_j}c_3\exp(-c_4\Delta_j^2)+c_3\exp(-c_4\Delta_j^2)M(M^{3/2}\Delta_j/2 + \sqrt{M(1+c_1)}\Delta_j) \\
	& \leq c_3(\frac{\sqrt{M}}{\sqrt{1-c_1}\Delta_{min}}+M^{5/2}\Delta_{max}/2 + M^{3/2}\sqrt{1+c_1}\Delta_{max})\exp(-c_4\Delta_{min}^2) \\
	& := c_6 \exp(-c_4\Delta_{min}^2).
\end{align*}
It follows that $$ \|\frac{\partial \B}{\partial \lambda_j}\|_2 \leq c_6 \frac{1}{K} \exp(-c_4\Delta_{min}^2). $$
Next, we work on providing an upper bound for $\| \A^{-1}\B \|_2$. It can be verified that
\begin{align*}
	\A^{-1}\B &= [\frac{1}{K}\sum_{j=1}^K E\gamma_{\btheta}(Y_j) \bOmega_j]^{-1} [\frac{1}{K}\sum_{j=1}^K E\gamma_{\btheta^*}(Y_j) \bOmega_j] [\frac{1}{K}\sum_{j=1}^K E\gamma_{\btheta^*}(Y_j) \bOmega_j]^{-1}\cdot  \\
	& \left\{\frac{1}{K}\sum_{j=1}^K E[\gamma_{\btheta}(Y_j) \bOmega_j (Y_j-\bmu_1)] - \frac{1}{K}\sum_{j=1}^K E[\gamma_{\btheta^*}(Y_j) \bOmega_j (Y_j-\bmu_1^*)]  \right\}.
\end{align*}
Let $g(\btheta)=\frac{1}{K}\sum_{j=1}^K E[\gamma_{\btheta}(Y_j) \bOmega_j (Y_j-\bmu_1)]$, thus 
\begin{align*}
	\frac{1}{K}\sum_{j=1}^K & E[\gamma_{\btheta}(Y_j) \bOmega_j (Y_j-\bmu_1)] - \frac{1}{K}\sum_{j=1}^K E[\gamma_{\btheta^*}(Y_j) \bOmega_j (Y_j-\bmu_1^*)] \\ 
	& = g(\btheta) - g(\btheta^*) \\
	& = \int_0^1 \frac{d g(\btheta)}{d \btheta}|_{\btheta=\btheta_u}(\btheta - \btheta^*) du \\
	& = \sum_{j=1}^K \int_0^1 \frac{\partial g(\btheta)}{\partial \lambda_j}|_{\btheta=\btheta_u}\Delta_{\lambda_j} du + \sum_{k=1}^2 \int_0^1 \frac{\partial g(\btheta)}{\partial \bmu_k}|_{\btheta=\btheta_u}\Delta_{\bmu_k} du
\end{align*}
with $\btheta_u=\btheta + u (\btheta^* - \btheta)$ with $u\in(0,1)$. In particular, 
\begin{align*}
	& \frac{\partial g(\btheta)}{\partial \lambda_j} = \frac{1}{K}E[ \frac{\partial \gamma_{\btheta}(Y_j)}{\partial \lambda_j}\bOmega_j (Y_j-\bmu_1)] \\
	& \frac{\partial g(\btheta)}{\partial \bmu_0}=\frac{1}{K}\sum_{j=1}^K E(\bOmega_j(Y_j -\bmu_1)[ \frac{\partial \gamma_{\btheta}(Y_j)}{\partial \bmu_0} ]^T) \\
	& \frac{\partial g(\btheta)}{\partial \bmu_1}=\frac{1}{K}\sum_{j=1}^K E(\bOmega_j(Y_j -\bmu_1)[ \frac{\partial \gamma_{\btheta}(Y_j)}{\partial \bmu_1} ]^T) + \frac{1}{K}\sum_{j=1}^K E(\gamma_{\btheta}(Y_j)\bOmega_j). 
\end{align*}
Thus, we have 
\begin{align*}
	\|\A^{-1}\B\|_2 & \lesssim \sum_{j=1}^K \| \frac{\partial g(\btheta)}{\partial \lambda_j}  \|_2 |\Delta_{\lambda_j}| + \sum_{k=0}^1 \|\frac{\partial g(\btheta)}{\partial \bmu_k}\|_2\|\Delta_{\bmu_k}\|_2 \\ 
	& \lesssim c_6\exp(-c_4\Delta_{min}^2) +  \frac{M^{3/2}}{4}\Delta_{max} + \frac{M^{3/2}}{4}\|\frac{1}{K}\sum_{j=1}^K E(\bOmega_j(Y_j -\bmu_1)[ \frac{\partial \gamma_{\btheta}(Y_j)}{\partial \bmu_1} ]^T)\|_2\Delta_{max}
\end{align*}
as $|\Delta_{\lambda_j}|\leq 1-c_0-c_w$ and $\|\Delta_{\bmu_k}\|_2 \leq M^{3/2}\Delta_{max}/4 $.
The term $ \|\frac{1}{K}\sum_{j=1}^K E(\bOmega_j(Y_j -\bmu_1)[ \frac{\partial \gamma_{\btheta}(Y_j)}{\partial \bmu_1} ]^T)\|_2 $ will be bounded later. Now, let's focus on $\frac{\partial \bmu_1(\btheta)}{\partial \bmu_1}$. Similarly, we need to take derivative on $\A$ and $\B$ separately with respect to $\bmu_1$, and $\partial \A/\partial \bmu_1$ will be a tensor of dimension $d\times d\times d$ and we can rearrange the elements into a big matrix of dimension $d \times d^2$, i.e., by writing $\bmu_1$ as $(\bmu_{21},\bmu_{22},\ldots,\bmu_{2d})^T$ we have
$$ \frac{\partial \A}{\partial \bmu_1} = [ \frac{\partial \A}{\partial \bmu_{21}},\ldots, \frac{\partial \A}{\partial \bmu_{2d}}  ]. $$
Thus, 
\begin{align*}
\|\frac{\partial \bmu_1(\btheta)}{\partial \bmu_1}\|_2 &= \left\| -\A^{-1}\frac{\partial \A}{\partial \bmu_1}(\I_d \otimes \A^{-1}\B) + \A^{-1} \frac{\partial \B}{\partial \bmu_1} \right\|_2 \\
& \leq \left\|\A^{-1}\right\|_2 \left\|\frac{\partial \A}{\partial \bmu_1}\right\|_2 \left\|\A^{-1}\B \right\|_2 + \left\|\A^{-1}\right\|_2 \left\| \frac{\partial \B}{\partial \bmu_1} \right\|_2.
\end{align*}
As rearranging the columns of a matrix will not change its operator norm, it can be verified that we can also write   
 $$ \frac{\partial \A}{\partial \bmu_1} = \frac{1}{K} \sum_{j=1}^K \bOmega_j \otimes \frac{\partial}{\partial \bmu_1}E\gamma_{\btheta}(Y_j), $$
and it leads to
\begin{align*}
	\| \frac{\partial \A}{\partial \bmu_1} \|_2 & \leq \frac{1}{K} \sum_{j=1}^K \|\bOmega_j\|_2 \| \frac{\partial}{\partial \bmu_1}E\gamma_{\btheta}(Y_j) \|_2 \\
	& \leq MC_{\bmu}\exp(-c_4\Delta_{min}^2).
\end{align*}
Also, we have $$ \frac{\partial \B}{\partial \bmu_1} = \frac{1}{K} \sum_{j=1}^K E[ \bOmega_j(Y_j - \bmu_1)(\frac{\partial \gamma_{\btheta}(Y_j)}{\partial \bmu_1})^T ],  $$
and $$ \| \frac{\partial \B}{\partial \bmu_1} \|_2 \leq \frac{1}{K} \sum_{j=1}^K \|E[ \bOmega_j(Y_j - \bmu_1)(\frac{\partial \gamma_{\btheta}(Y_j)}{\partial \bmu_1})^T ]\|_2:=\frac{1}{K}\sum_{j=1}^K \| \frac{\partial \B_j}{\partial \bmu_1} \|_2 $$
For each $j$,  
\begin{align*}
\| \frac{\partial \B_j}{\partial \bmu_1} \|_2=\|E[ \bOmega_j(Y_j - \bmu_1)(\frac{\partial \gamma_{\btheta}(Y_j)}{\partial \bmu_1})^T ]\|_2 & \leq \| \lambda_j(1-\lambda_j) E\frac{\partial \gamma_{\btheta}(Y_j)}{\partial \lambda_j}\bOmega_j(Y_j-\bmu_1)(Y_j-\bmu_1)^T\bOmega_j  \|_2. 
\end{align*}
As $Y_j = \bSigma_j^{1/2}Z_N + \bSigma_j^{1/2}\Psi_j + (\bmu_0^* + \bmu_1^*)/2 $, we have
\begin{align*}
	(Y_j-\bmu_1)(Y_j-\bmu_1)^T 
	& = \bSigma_j^{1/2}Z_N(\bSigma_j^{1/2}Z_N)^T - 2 \bSigma_j^{1/2}Z_N (\bSigma_j^{1/2}\Psi_j + (\bmu_0^* + \bmu_1^*)/2 - \bmu_1)^T  \\ 
	& + (\bSigma_j^{1/2}\Psi_j + (\bmu_0^* + \bmu_1^*)/2- \bmu_1)(\bSigma_j^{1/2}\Psi_j + (\bmu_0^* + \bmu_1^*)/2- \bmu_1)^T.
\end{align*}
It follows that
\begin{align}
	& \frac{1}{\lambda_j(1-\lambda_j)}\frac{\partial \B_j}{\partial \bmu_1} \nonumber \\ 
	& = E\frac{\partial \gamma_{\btheta}(Y_j)}{\partial \lambda_j}\bSigma_j^{-1/2}Z_N(\bSigma_j^{-1/2}Z_N)^T \label{eqDevMu1} \\
	& - 2E\frac{\partial \gamma_{\btheta}(Y_j)}{\partial \lambda_j}\bSigma_j^{-1/2}Z_N (\bSigma_j^{1/2}\Psi_j + (\bmu_0^* + \bmu_1^*)/2 - \bmu_1)^T\bSigma_j^{-1} \label{eqDevMu2} \\
	& +  E\frac{\partial \gamma_{\btheta}(Y_j)}{\partial \lambda_j}\bSigma_j^{-1}(\bSigma_j^{1/2}\Psi_j + (\bmu_0^* + \bmu_1^*)/2- \bmu_1)(\bSigma_j^{1/2}\Psi_j + (\bmu_0^* + \bmu_1^*)/2- \bmu_1)^T\bSigma_j^{-1}. \label{eqDevMu3}
\end{align}
For \eqref{eqDevMu2}, 
\begin{align*}
 & \| E\frac{\partial \gamma_{\btheta}(Y_j)}{\partial \lambda_j}\bSigma_j^{-1/2}Z_N (\bSigma_j^{1/2}\Psi_j + (\bmu_0^* + \bmu_1^*)/2 - \bmu_1)^T\bSigma_j^{-1} \|_2 \\ 
 & \leq \| E\frac{\partial \gamma_{\btheta}(Y_j)}{\partial \lambda_j}\bSigma_j^{-1/2}Z_N \|_2 \| \bSigma_j^{1/2}\Psi_j + (\bmu_0^* + \bmu_1^*)/2 - \bmu_1\|_2 \|\bSigma_j^{-1} \|_2 \\ 
 & \leq \frac{M^2(M\Delta_j/2 + \sqrt{1+c_1}\Delta_j)}{\sqrt{1-c_1}\Delta_j}c_3\exp(-c_4\Delta_j^2).
\end{align*}
For \eqref{eqDevMu3},
\begin{align*}
	& \|E\frac{\partial \gamma_{\btheta}(Y_j)}{\partial \lambda_j}\bSigma_j^{-1}(\bSigma_j^{1/2}\Psi_j + (\bmu_0^* + \bmu_1^*)/2- \bmu_1)(\bSigma_j^{1/2}\Psi_j + (\bmu_0^* + \bmu_1^*)/2- \bmu_1)^T\bSigma_j^{-1} \|_2 \\
	& \leq |E\frac{\partial \gamma_{\btheta}(Y_j)}{\partial \lambda_j}| \|\bSigma_j^{-1} \|_2^2 \| \bSigma_j^{1/2}\Psi_j + (\bmu_0^* + \bmu_1^*)/2- \bmu_1 \|_2^2 \\ 
	& \leq M^3(M\Delta_j/2 + \sqrt{1+c_1}\Delta_j)^2c_3\exp(-c_4\Delta_j^2).
\end{align*}
For \eqref{eqDevMu1}, using the same technique as before, let $\alpha_j=\bSigma_j^{1/2}\bbeta_j$ and $\H$ be an orthogonal matrix whose first row is $\alpha_j/\|\alpha_j\|_2$ with $\H\alpha=\|\alpha_j\|_2\e_1$. Then follow the proof in page 19--20 of the supplementary material of \cite{cai2019chime}, we have 
$$ \| E\frac{\partial \gamma_{\btheta}(Y_j)}{\partial \lambda_j}\bSigma_j^{-1/2}Z_N(\bSigma_j^{-1/2}Z_N)^T \|_2 \leq \frac{4M^3c_{7j}}{1-c_1}\exp(-c_4\Delta_j^2) + c_3 M^3 \exp(-c_4\Delta_j^2),$$
with $$ c_{7j} = \frac{2[1+2(1+c_1)\Delta_j^2 + 2(1+c_1)^2\Delta_j^4]}{c_0^2(1-c_1)\Delta_j^2}. $$
Thus, 
\begin{align*}
	\|\frac{\partial \B_j}{\partial \bmu_1}\|_2 & \leq \frac{1}{4}\{ \frac{4M^3c_{7j}}{1-c_1}\exp(-c_4\Delta_j^2) + c_3 M^3 \exp(-c_4\Delta_j^2) \\
	& +  \frac{2M^2(M\Delta_j/2 + \sqrt{1+c_1}\Delta_j)}{\sqrt{1-c_1}\Delta_j}c_3\exp(-c_4\Delta_j^2) \\
	& +  M^3(M\Delta_j/2 + \sqrt{1+c_1}\Delta_j)^2c_3\exp(-c_4\Delta_j^2) \} \\
	& \leq \frac{1}{4}\{ \frac{4M^3c_{7}}{1-c_1}\exp(-c_4\Delta_{min}^2) + c_3 M^3 \exp(-c_4\Delta_{min}^2) \\
	& +  \frac{2M^2(M\Delta_{max}/2 + \sqrt{1+c_1}\Delta_{max})}{\sqrt{1-c_1}\Delta_{min}}c_3\exp(-c_4\Delta_{min}^2) \\
	& +  M^3(M\Delta_{max}/2 + \sqrt{1+c_1}\Delta_{max})^2c_3\exp(-c_4\Delta_{min}^2) \} \\
	& := c_{\bmu}\exp(-c_4\Delta_{min}^2)
\end{align*}
with $$ c_{7} = \frac{2[1+2(1+c_1)\Delta_{max}^2 + 2(1+c_1)^2\Delta_{max}^4]}{c_0^2(1-c_1)\Delta_{min}^2} $$ and 
$$ c_{\bmu}= \frac{1}{4}\left\{ \frac{4M^3c_{7}}{1-c_1} + c_3 M^3 
	 +  \frac{2M^2(M\Delta_{max}/2 + \sqrt{1+c_1}\Delta_{max})}{\sqrt{1-c_1}\Delta_{min}}c_3
	 +  M^3(M\Delta_{max}/2 + \sqrt{1+c_1}\Delta_{max})^2c_3\right\}.$$ Thus, we have 
$$ \|\frac{\partial \B}{\partial \bmu_1}\|_2 \leq c_{\bmu}\exp(-c_4\Delta_{min}^2), $$
%and due to the symmetry we can also verify  
%$$ \|\frac{\partial \B}{\partial \bmu_0}\|_2 \leq c_{\bmu}\exp(-c_4\Delta_{min}^2). $$
% With this results, we can get 
%$$ \|\frac{\partial g(\btheta)}{\partial \bmu_0} \|_2 = \|\frac{\partial \B}{\partial \bmu_0}\|_2 \leq c_{\bmu}\exp(-c_4\Delta_{min}^2), $$
which also leads to 
\begin{align*}
	\|\A^{-1}\B\|_2 &  \lesssim c_6\exp(-c_4\Delta_{min}^2) + \frac{M^{3/2}}{4}\Delta_{max} + \frac{M^{3/2}}{4}c_{\bmu}\exp(-c_4\Delta_{min}^2)\Delta_{max} := c_{\A\B}.
\end{align*}
Thus, we have 
\begin{align*}
    \| \frac{\partial \bmu_1(\btheta)}{\partial \lambda_j}\Delta_{\lambda_j} \|_2  
	 & \leq (\|\A^{-1}\|_2 \|\frac{\partial \A}{\partial \lambda_j}\|_2 \|\A^{-1} \B\|_2 + \|\A^{-1}\|_2 \|\frac{\partial \B}{\partial \lambda_j}\|_2)|\Delta_{\lambda_j}| \\
	 & \lesssim \frac{1}{K} [M c_3\exp(-c_4\Delta_{min}^2)c_{\A\B} +  c_6 \exp(-c_4\Delta_{min}^2)] |\Delta_{\lambda_j}|.
\end{align*}
as 
\begin{align*}
	\|\frac{\partial \A}{\partial \lambda_j}\|_2 & \leq M \frac{1}{K} c_3\exp(-c_4\Delta_{min}^2).
\end{align*}
and 
$$ \|\frac{\partial \B}{\partial \lambda_j}\|_2 \leq c_6 \frac{1}{K} \exp(-c_4\Delta_{min}^2). $$
Also, as  $ \| \frac{\partial \A}{\partial \bmu_1} \|_2 \leq MC_{\bmu}\exp(-c_4\Delta_{min}^2)$ and 
$ \|\frac{\partial \B}{\partial \bmu_1}\|_2 \leq c_{\bmu}\exp(-c_4\Delta_{min}^2)$ we have
\begin{align*}
\|\frac{\partial \bmu_1(\btheta)}{\partial \bmu_1}\|_2 
&  \leq \left\|\A^{-1}\right\|_2 \left\|\frac{\partial \A}{\partial \bmu_1}\right\|_2 \left\|\A^{-1}\B \right\|_2 + \left\|\A^{-1}\right\|_2 \left\| \frac{\partial \B}{\partial \bmu_1} \right\|_2 \\ 
& \lesssim   MC_{\bmu} c_{\A\B}\exp(-c_4\Delta_{min}^2) +  c_{\bmu}\exp(-c_4\Delta_{min}^2).
\end{align*}
where $$C_{\bmu} = (\sqrt{M}\Delta_{max}/8 + \sqrt{M}\Delta_{max}(M^2/2+\sqrt{1+c_1})/8 + \frac{\sqrt{M}}{4\sqrt{1-c_1}\Delta_{min}})c_3 $$ and
$$ c_{\bmu}= \frac{1}{4}\left\{ \frac{4M^3c_{7}}{1-c_1} + c_3 M^3 
	 +  \frac{2M^2(M\Delta_{max}/2 + \sqrt{1+c_1}\Delta_{max})}{\sqrt{1-c_1}\Delta_{min}}c_3
	 +  M^3(M\Delta_{max}/2 + \sqrt{1+c_1}\Delta_{max})^2c_3\right\}.$$ 
Thus, $$ \| \frac{\partial \bmu_1(\btheta)}{\partial \bmu_1}\Delta_{\bmu_1} \|_2 \lesssim (MC_{\bmu} c_{\A\B}\exp(-c_4\Delta_{min}^2) +  c_{\bmu}\exp(-c_4\Delta_{min}^2))\| \Delta_{\bmu_1} \|_2. $$
Due to the symmetry, we also have 
$$ \| \frac{\partial \bmu_1(\btheta)}{\partial \bmu_0}\Delta_{\bmu_0} \|_2 \lesssim (MC_{\bmu} c_{\A\B}\exp(-c_4\Delta_{min}^2) +  c_{\bmu}\exp(-c_4\Delta_{min}^2))\| \Delta_{\bmu_0} \|_2. $$
It follows that 
\begin{align*}
	\| \bmu_1(\btheta) - \bmu_1(\btheta^*)  \|_2 & \leq \sum_{j=1}^K \| \frac{\partial \bmu_1(\btheta)}{\partial \lambda_j}\|_2|\Delta_{\lambda_j}|  + \sum_{k=0}^1 \| \frac{\partial \bmu_1(\btheta)}{\partial \bmu_k} \Delta_{\bmu_k}\|_2 \\ 
	& \leq \frac{1}{K}[M c_3\exp(-c_4\Delta_{min}^2)c_{\A\B} +  c_6 \exp(-c_4\Delta_{min}^2)] \sum_{j=1}^K |\Delta_{\lambda_j}| \\ 
	& +  (MC_{\bmu} c_{\A\B}\exp(-c_4\Delta_{min}^2) +  c_{\bmu}\exp(-c_4\Delta_{min}^2))(\| \Delta_{\bmu_0} \|_2 +\| \Delta_{\bmu_1} \|_2) \\ 
	& \leq \frac{1}{\sqrt{K}}[M c_3\exp(-c_4\Delta_{min}^2)c_{\A\B} +  c_6 \exp(-c_4\Delta_{min}^2)] \|\Delta_{\Lambda} \|_2 \\ 
	& +  (MC_{\bmu} c_{\A\B}\exp(-c_4\Delta_{min}^2) +  c_{\bmu}\exp(-c_4\Delta_{min}^2))(\| \Delta_{\bmu_0} \|_2 +\| \Delta_{\bmu_1} \|_2) \\
	& \leq \kappa_{2}d_2(\btheta,\btheta^*)
\end{align*}
with $$ \kappa_{2} :=  [ \{ \frac{1}{\sqrt{K}} ( M c_3c_{\A\B} +  c_6 )\} \vee \{ MC_{\bmu} c_{\A\B} +  c_{\bmu}\} ] \exp(-c_4\Delta_{min}^2). $$
We can also get 
\begin{align*}
	\| \bmu_1(\btheta) - \bmu_1(\btheta^*)  \|_2  
	& \leq \frac{1}{K}[M c_3\exp(-c_4\Delta_{min}^2)c_{\A\B} +  c_6 \exp(-c_4\Delta_{min}^2)] \sum_{j=1}^K |\Delta_{\lambda_j}| \\ 
	& +  (MC_{\bmu} c_{\A\B}\exp(-c_4\Delta_{min}^2) +  c_{\bmu}\exp(-c_4\Delta_{min}^2))(\| \Delta_{\bmu_0} \|_2 +\| \Delta_{\bmu_1} \|_2) \\ 
	& \leq \kappa_{3} \frac{1}{K} \sum_{j=1}^K d_2(\btheta_j,\btheta_j^*)
\end{align*}
with $$ \kappa_{3} :=  [ \{M c_3c_{\A\B} +  c_6\} \vee \{ MC_{\bmu} c_{\A\B} +  c_{\bmu}\} ] \exp(-c_4\Delta_{min}^2). $$
Note that, $\kappa_3$ is the $\kappa'$ in Theorem \ref{overall_contraction} of the manuscript.

\subsection*{Combine the results}

Recall that with $\kappa_1 = c_3 \exp(-c_4\Delta_{min}^2) \vee C_{\bmu}\exp(-c_4\Delta_{min}^2) $ where $C_{\bmu} = (\sqrt{M}\Delta_{max}/8 + \sqrt{M}\Delta_{max}(M^2/2+\sqrt{1+c_1})/8 + \frac{\sqrt{M}}{4\sqrt{1-c_1}\Delta_{min}})c_3$, we have
$$ |E(\gamma_{\btheta}(Y_j))-E(\gamma_{\btheta^*}(Y_j))| \leq \kappa_1 d_2(\btheta, \btheta^*). $$
Thus, 
\begin{align*}
	d_2( M(\btheta), \btheta^* ) & = \sqrt{\sum_{j=1}^K | E(\gamma_{\btheta}(Y_j))-E(\gamma_{\btheta^*}(Y_j)) |^2} + \sum_{k=0}^1 \| \bmu_k(\btheta) - \bmu_k(\btheta^*) \|_2  \\ 
	& \leq (\sqrt{K}\kappa_1+2\kappa_2)d_2(\btheta,\btheta^*) \\ 
	& := \kappa d_2(\btheta,\btheta^*).
\end{align*}
Based on the assumption that $\Delta_{max} = c_2\Delta_{min}$, we can write $\kappa$ as $\kappa=\kappa_0\exp(-c_4\Delta_{min}^2)$ where $\kappa_0 = poly(\Delta_{min};M,c_0,c_1,c_2,K)+poly(1/\Delta_{min};M,c_0,c_1,c_2,K)$. Therefore, we can find a quantity $C(c_0,c_1,c_2,M,K)$ decided by $c_0, c_1,c_2, M$ and $K$ such that when the global SNR $\Delta_{min}$ is large enough to satisfy $\Delta_{min}>C(c_0,c_1,c_2,M,K)$,  there exists a $\kappa \in (0,1)$ such that 
$$ d_2( M(\btheta), \btheta^* ) \leq  \kappa d_2(\btheta,\btheta^*).$$ 
\end{proof}

\section{Proof of Lemma S.4}

\begin{proof}%[Proof of Lemma S.4]
We divide the proof into the derivation of the concentration inequalities for the estimates of $\lambda_{j}$ and $\bmu_k$ separately.  

\subsection{Concentration of the mixing proportion}
For each $j\in[K]$, we have
\[
\lambda^n_j(\btheta) = \frac{1}{n}\sum_{i=1}^n \frac{\lambda_j}{\lambda_j + (1-\lambda_j)\exp\{(\bmu_0-\bmu_1)^T\bSigma_j^{-1}(Y_{ij}-\frac{\bmu_0+\bmu_1}{2})\}}.
\]
Let's define $$ Z_{\lambda_j} = \sup_{\btheta\in\B(\btheta^*;c_0,c_1)}|\lambda^n_j(\btheta)-\lambda_j(\btheta)|, $$
and let $\{\epsilon_1,\ldots,\epsilon_n\}$ be a sequence of i.i.d. Rademacher random variables. Then, for any $\lambda>0$, by using a standard symmetrization result for empirical processes, we have
$$ E(\exp(\lambda Z_{\lambda_j}))\leq E[\exp( 2\lambda \sup_{\btheta\in\B(\btheta^*;c_0,c_1)}|  \frac{1}{n}\sum_{i=1}^n \epsilon_i \frac{\lambda_j}{\lambda_j + (1-\lambda_j)\exp\{\bbeta_j^T(Y_{ij}-\frac{\bmu_0+\bmu_1}{2})\}}|)]. $$
 We can check that $\psi(x)=\frac{\lambda_j}{\lambda_j+(1-\lambda_j)\exp(x)}-\lambda_j$ is Lipschitz with constant $\frac{1-\lambda_j}{\lambda_j}\leq \frac{1-c_0}{c_0}$ and $\psi(0)=0$. Then by applying Lemma C.1 of \cite{cai2019chime}, i.e., the Ledoux-Talagrand  contraction for Rademacher processes with $g(\cdot)=1$, we have
 \begin{align}
 	&E(\exp(\lambda Z_{\lambda_j})) \nonumber
 	\\
 	& \leq E[\exp( 2\lambda \sup_{\btheta\in\B(\btheta^*;c_0,c_1)}|  \frac{1}{n}\sum_{i=1}^n \epsilon_i ( \frac{\lambda_j}{\lambda_j + (1-\lambda_j)\exp\{\bbeta_j^T(Y_{ij}-\frac{\bmu_0+\bmu_1}{2})\}} - \lambda_j ) |)] \nonumber \\ 
 	& \cdot E(\exp(\frac{2\lambda}{n}\sup_{\btheta\in\B(\btheta^*;c_0,c_1)}|  \frac{1}{n}\sum_{i=1}^n \epsilon_i \lambda_j|)) \nonumber \\ 
 	& \leq E[\exp( 2\lambda \sup_{\btheta\in\B(\btheta^*;c_0,c_1)}|  \frac{1}{n}\sum_{i=1}^n \epsilon_i \frac{1-c_0}{c_0} \bbeta_j^T(Y_{ij}-\frac{\bmu_0+\bmu_1}{2})|)]\cdot \exp(\frac{4\lambda^2(1-c_0)^2}{n})  \nonumber\\ 
 	& \leq E[\exp( 2\lambda \sup_{\btheta\in\B(\btheta^*;c_0,c_1)}|  \frac{1}{n}\sum_{i=1}^n \epsilon_i \frac{1-c_0}{c_0} \bbeta_j^T(Y_{ij} - \bmu_j^* + \bmu_j^* -\frac{\bmu_0+\bmu_1}{2})|)]\cdot \exp(\frac{4\lambda^2(1-c_0)^2}{n}) \nonumber\\ 
 	& \leq E[\exp( 2\lambda \sup_{\btheta\in\B(\btheta^*;c_0,c_1)}|  \frac{1}{n}\sum_{i=1}^n \epsilon_i \frac{1-c_0}{c_0} \bbeta_j^T(Y_{ij} - \bmu_j^*)|)] \label{eqConc1} \\ & \cdot E[\exp( 2\lambda \sup_{\btheta\in\B(\btheta^*;c_0,c_1)}|  \frac{1}{n}\sum_{i=1}^n \epsilon_i \frac{1-c_0}{c_0} \bbeta_j^T  (\bmu_j^* -\frac{\bmu_0+\bmu_1}{2})|)] \label{eqConc2} \\ 
 	    & \cdot \exp(\frac{4\lambda^2(1-c_0)^2}{n}) \nonumber
 \end{align}
where $\bmu_j^*=(1-\lambda_j^*)\bmu_0^* + \lambda_j^*\bmu_1^*$ and the second inequality is due to the property of sub-Gaussian norm of bounded random variables. We first handle \eqref{eqConc2}, as $\btheta \in \B(\btheta^*;c_0,c_1)$ we have
\begin{align*}
	\sup_{\btheta\in\B(\btheta^*;c_0,c_1)}|   \bbeta_j^T  (\bmu_j^* -\frac{\bmu_0+\bmu_1}{2})| = \sup_{\btheta\in\B(\btheta^*;c_0,c_1)}| (1-\lambda_j^*)\delta_0(\bbeta_j) + \lambda_j^* \delta_1(\bbeta_j) | \leq (1+c_1)\Delta_j^2.
\end{align*}
Thus, we have 
$$ E[\exp( 2\lambda \sup_{\btheta\in\B(\btheta^*;c_0,c_1)}|  \frac{1}{n}\sum_{i=1}^n \epsilon_i \frac{1-c_0}{c_0} \bbeta_j^T  (\bmu_j^* -\frac{\bmu_0+\bmu_1}{2})|)] \leq \exp(\frac{4\lambda^2}{n}(\frac{1-c_0}{c_0})^2(1+c_1)^2\Delta_j^4). $$
As for \eqref{eqConc1}, let $Y^{N}_{ij}= Y_{ij}-\bmu_j^*$ be a centered random variable, then 
\begin{align*}
	E[\exp( & 2\lambda \sup_{\btheta\in\B(\btheta^*;c_0,c_1)}|  \frac{1}{n}\sum_{i=1}^n \epsilon_i \frac{1-c_0}{c_0} \bbeta_j^TY_{ij}^N|)] \\ 
	& \leq E(\exp(\frac{2\lambda}{n}
\frac{1-c_0}{c_0} \sup_{\btheta\in\B(\btheta^*;c_0,c_1)} \| \bbeta_j\|_1 \sup_{l \in [d] } | \sum_{i=1}^n \epsilon_i Y_{ij,l}^N | )) \\ 
 & \leq E(\exp(\frac{2\lambda}{n}
\frac{1-c_0}{c_0} \sqrt{dM(1+c_1)}\Delta_j\sup_{l \in [d] } | \sum_{i=1}^n \epsilon_i Y_{ij,l}^N |))  \\ 
& \lesssim \sum_{l=1}^d \exp( \frac{4\lambda^2}{n}
(\frac{1-c_0}{c_0})^2 dM(1+c_1)\Delta^2_j) \\
& \leq \exp( \frac{4\lambda^2}{n}
(\frac{1-c_0}{c_0})^2 dM(1+c_1)\Delta^2_j + \log(d))
\end{align*}
where $Y_{ij,l}^N$ is the $l$-th element of $Y_{ij}^N$, the second inequality is due to 
\begin{align*}
\|\bbeta_j\|_1\leq \sqrt{d}\|\bbeta_j \|_2=\sqrt{d} \| \bSigma_j^{-1/2}\bSigma_j^{1/2}\bbeta_j \|_2\leq \sqrt{dM}\sigma(\bbeta_j)\leq \sqrt{dM(1+c_1)}\Delta_j	
\end{align*}
and the third inequality is by using the property of sub-Gaussian norm of Gaussian random variables and the assumption that all variances are bounded. By combining the above two results, we have 
\begin{align*}
	E(\exp(\lambda Z_{\lambda_j} )) & \leq \exp( \frac{4\lambda^2}{n}
(\frac{1-c_0}{c_0})^2 dM(1+c_1)\Delta^2_j + \log(d)) \cdot  \exp(\frac{4\lambda^2}{n}(\frac{1-c_0}{c_0})^2(1+c_1)^2\Delta_j^4) \\ 
& \cdot \exp(\frac{4\lambda^2(1-c_0)^2}{n}) \\  
& \leq \exp(\frac{\lambda^2}{n}(c_1d + c_2)+ \log(d))
\end{align*}
with $a_1$ and $a_2$ are two constants that do not rely on $(n,d)$. Then, using the Chernoff's approach, let $$ \lambda=\sqrt{\frac{n(\log(n) + \log(d))}{a_1 d + a_2}},\ t=2\sqrt{\log(nd)}\sqrt{\frac{a_1 d + a_2}{n}},  $$
then we have 
\begin{align*}	P(\sup_{\btheta\in\B(\btheta^*;c_0,c_1)} | \lambda_j^n(\btheta) - \lambda_j(\btheta) | > t )  = P(Z_{\lambda_j} > t) \leq e^{-\lambda t} E(e^{\lambda Z_{\lambda_j}}) \leq \frac{1}{n}.  
\end{align*}
It implies that, with probability at least $1-n^{-1}$, 
$$ \sup_{\btheta\in\B(\btheta^*;c_0,c_1)} | \lambda_j^n(\btheta) - \lambda_j(\btheta) | \leq 2\sqrt{\log(nd)}\sqrt{\frac{a_1 d + a_2}{n}} $$
and with probability at least $1-Kn^{-1}$, for all $j \in [K]$, we have 
$$ \sup_{\btheta\in\B(\btheta^*;c_0,c_1)} | \lambda_j^n(\btheta) - \lambda_j(\btheta) | \lesssim \sqrt{\frac{d}{n}} $$
if ignore the logarithmic term. 

\subsection{Concentration of the mean}
Here we only provide the derivation of concentration results of $\bmu_1(\btheta)$, the results of $\bmu_0(\btheta)$ can be similarly obtained. Recall that 
\begin{align*}
	& \bmu_1^n(\btheta)=[\frac{1}{nK}\sum_{j=1}^K\sum_{i=1}^n \gamma_{\btheta}(Y_{ij}) \bOmega_j ]^{-1} \frac{1}{nK}\sum_{j=1}^K\sum_{i=1}^n \gamma_{\btheta}(Y_{ij}) \bOmega_j Y_{ij},\\
	& \bmu_1(\btheta)=[\frac{1}{K}\sum_{j=1}^K E\gamma_{\btheta}(Y_j) \bOmega_j ]^{-1} \frac{1}{K}\sum_{j=1}^K E[\gamma_{\btheta}(Y_j) \bOmega_j Y_j]. %\label{eqpop3}
\end{align*}
Thus,
\begin{align*}
	\bmu_1^n(\btheta) - \bmu_1(\btheta) &= [\frac{1}{nK}\sum_{j=1}^K\sum_{i=1}^n \gamma_{\btheta}(Y_{ij}) \bOmega_j ]^{-1} \frac{1}{nK}\sum_{j=1}^K\sum_{i=1}^n \gamma_{\btheta}(Y_{ij}) \bOmega_j (Y_{ij} -\frac{\bmu_0 + \bmu_1}{2} )  \\ 
	& -  [\frac{1}{K}\sum_{j=1}^K E\gamma_{\btheta}(Y_j) \bOmega_j ]^{-1} \frac{1}{K}\sum_{j=1}^K E[\gamma_{\btheta}(Y_j) \bOmega_j (Y_j -\frac{\bmu_0 + \bmu_1}{2} )] \\ 
	& = [\frac{1}{nK}\sum_{j=1}^K\sum_{i=1}^n \gamma_{\btheta}(Y_{ij}) \bOmega_j ]^{-1} \{ \frac{1}{nK}\sum_{j=1}^K\sum_{i=1}^n \gamma_{\btheta}(Y_{ij}) \bOmega_j (Y_{ij} -\frac{\bmu_0 + \bmu_1}{2} )   \\
	& - \frac{1}{K}\sum_{j=1}^K E[\gamma_{\btheta}(Y_j) \bOmega_j (Y_j -\frac{\bmu_0 + \bmu_1}{2} )] \} + \{ [\frac{1}{nK}\sum_{j=1}^K\sum_{i=1}^n \gamma_{\btheta}(Y_{ij}) \bOmega_j ]^{-1} \\ 
	& - [\frac{1}{K}\sum_{j=1}^K E\gamma_{\btheta}(Y_j) \bOmega_j ]^{-1}  \} \frac{1}{K}\sum_{j=1}^K E\gamma_{\btheta}(Y_j) \bOmega_j (Y_j -\frac{\bmu_0 + \bmu_1}{2} ).
\end{align*}
Let 
\begin{align*}
	W^{(\bmu)}&=\sup_{\btheta\in \B(\btheta^*;c_0,c_1)} \| \frac{1}{nK}\sum_{j=1}^K\sum_{i=1}^n \gamma_{\btheta}(Y_{ij}) \bOmega_j (Y_{ij} -\frac{\bmu_0 + \bmu_1}{2} )    - \frac{1}{K}\sum_{j=1}^K E[\gamma_{\btheta}(Y_j) \bOmega_j (Y_j -\frac{\bmu_0 + \bmu_1}{2} )] \|_2 \\ 
	& = \sup_{\btheta\in \B(\btheta^*;c_0,c_1)} \| \frac{1}{nK}\sum_{j=1}^K\sum_{i=1}^n \{\gamma_{\btheta}(Y_{ij}) \bOmega_j (Y_{ij} -\frac{\bmu_0 + \bmu_1}{2} )    - E[\gamma_{\btheta}(Y_j) \bOmega_j (Y_j -\frac{\bmu_0 + \bmu_1}{2} )] \}\|_2,
\end{align*}
and \begin{align*}
	W^{(\bmu)}_u
	& = \sup_{\btheta\in \B(\btheta^*;c_0,c_1)} \langle \frac{1}{nK}\sum_{j=1}^K\sum_{i=1}^n \{\gamma_{\btheta}(Y_{ij}) \bOmega_j (Y_{ij} -\frac{\bmu_0 + \bmu_1}{2} )    - E[\gamma_{\btheta}(Y_j) \bOmega_j (Y_j -\frac{\bmu_0 + \bmu_1}{2} )] \}, u \rangle ,
\end{align*}
with $u \in \mathbb{S}^{d-1}=\{ u\in\mathbb{R}^d:\|u\|_2=1 \}$. We have that $ W^{(\bmu)}= \sup_{u \in \mathbb{S}^{d-1}} W^{(\bmu)}_u$. Let $\{u_1,\ldots, u_{M_{net}}\}$ denote a $1/2$-net of the space $\mathbb{S}^{d-1}$ (we have $\log(M_{net})\leq 2d$ by \cite{balakrishnan2017statistical}). This means that for any $v \in \mathbb{S}^{d-1}$ there exists some index $j \in [M_{net}]$, s.t., $\|v - u_j \|_2\leq 1/2$. We have
$$ W^{(\bmu)} = \sup_{v \in \mathbb{S}^{d-1}} W^{(\bmu)}_v \leq \max_{j \in [M_{net}]} W^{(\bmu)}_{u_j} + \frac{1}{2} W^{(\bmu)}, $$
which leads to $$ W^{(\bmu)} \leq 2  \max_{j \in [M_{net}]} W^{(\bmu)}_{u_j}. $$	
Thus, next we only need to bound $W_{u}^{(\bmu)}$ for a fixed $u$. Let $\{\epsilon_{ij}\},\ i=1,\ldots,n;\ j=1,\ldots,K$ denote a sequence of i.i.d. Rademacher random variables, for any $\lambda>0$, we have
\begin{align}
  &	E(e^{\lambda W^{(\bmu)}_u }) \nonumber \\ 
  & \leq E[\exp( 2\lambda \sup_{\btheta\in\B(\btheta*;c_0,c_1)} | \frac{1}{nK}\sum_{j=1}^K\sum_{i=1}^n  \epsilon_{ij} \frac{\lambda_j \langle  \bOmega_j (Y_{ij} -\frac{\bmu_0 + \bmu_1}{2} ), u \rangle}{\lambda_j+ (1-\lambda_j ) \exp(\bbeta_j^T(Y_{ij} - \frac{\bmu_0 + \bmu_1}{2})) } | ) ] \nonumber \\
  & \leq E[\exp( 2\lambda \sup_{\btheta\in\B(\btheta*;c_0,c_1)} | \frac{1}{nK}\sum_{j=1}^K\sum_{i=1}^n  \epsilon_{ij} \{ \frac{\lambda_j }{\lambda_j+ (1-\lambda_j ) \exp(\bbeta_j^T(Y_{ij} - \frac{\bmu_0 + \bmu_1}{2})) }  -\lambda_j \} \langle  \bOmega_j (Y_{ij} \nonumber \\ 
  & -\frac{\bmu_0 + \bmu_1}{2} ), u \rangle |  ) ]\cdot E[\exp( 2\lambda \sup_{\btheta\in\B(\btheta*;c_0,c_1)} | \frac{1}{nK}\sum_{j=1}^K\sum_{i=1}^n  \epsilon_{ij} \lambda_j \langle  \bOmega_j (Y_{ij} -\frac{\bmu_0 + \bmu_1}{2} ), u \rangle |)]. \label{eqCon2.5}
\end{align}
Let's first look at the second term at the right hand side of \eqref{eqCon2.5}  
\begin{align}
	E[ & \exp( 2\lambda \sup_{\btheta\in\B(\btheta*;c_0,c_1)} | \frac{1}{nK}\sum_{j=1}^K\sum_{i=1}^n  \epsilon_{ij} \lambda_j \langle  \bOmega_j (Y_{ij} -\frac{\bmu_0 + \bmu_1}{2} ), u \rangle |)] \nonumber \\
	& = E[\exp( 2\lambda \sup_{\btheta\in\B(\btheta*;c_0,c_1)} | \frac{1}{nK}\sum_{j=1}^K\sum_{i=1}^n  \epsilon_{ij} \lambda_j \langle  \bOmega_j (Y_{ij} - \bmu_j^* + \bmu_j^* -\frac{\bmu_0 + \bmu_1}{2} ), u \rangle |)] \nonumber \\
	& \leq E[\exp( 2\lambda \sup_{\btheta\in\B(\btheta*;c_0,c_1)} | \frac{1}{nK}\sum_{j=1}^K\sum_{i=1}^n  \epsilon_{ij} \lambda_j \langle  \bOmega_j (Y_{ij} - \bmu_j^*), u \rangle |)] \label{eqCon3} \\
	& \cdot E[\exp( 2\lambda \sup_{\btheta\in\B(\btheta*;c_0,c_1)} | \frac{1}{nK}\sum_{j=1}^K\sum_{i=1}^n  \epsilon_{ij} \lambda_j \langle  \bOmega_j (\bmu_j^* -\frac{\bmu_0 + \bmu_1}{2} ), u \rangle |)] \label{eqCon4}.
\end{align}
For \eqref{eqCon4}, as we have $ \| \bmu_j^* - (\bmu_0 + \bmu_1)/2 \|_2 \leq (M/2+1)\sqrt{M}\Delta_{max}$, $\lambda_j\leq 1-c_0$, and $\|\bOmega_j\|_2\leq M$, we have
$$ E[\exp( 2\lambda \sup_{\btheta\in\B(\btheta*;c_0,c_1)} | \frac{1}{nK}\sum_{j=1}^K\sum_{i=1}^n  \epsilon_{ij} \lambda_j \langle  \bOmega_j (\bmu_j^* -\frac{\bmu_0 + \bmu_1}{2} ), u \rangle |)] \leq \exp( \frac{4\lambda^2}{nK}(1-c_0)^2M^3(\frac{M}{2} + 1)^2\Delta_{max}^2 )  $$
by using the property of sub-Gaussian norm of bounded random variables. As for \eqref{eqCon3}, we have 
\begin{align*}
	E[ & \exp( 2\lambda \sup_{\btheta\in\B(\btheta*;c_0,c_1)} | \frac{1}{nK}\sum_{j=1}^K\sum_{i=1}^n  \epsilon_{ij} \lambda_j \langle  \bOmega_j (Y_{ij} - \bmu_j^*), u \rangle |)] \\
	& \leq E[ \exp( 2\lambda \sup_{\btheta\in\B(\btheta*;c_0,c_1)} | \frac{1}{nK}\sum_{j=1}^K\sum_{i=1}^n  \epsilon_{ij} \lambda_j \langle \tilde Y^N_{ij}, u \rangle |)] \ \text{with}\ \tilde Y^N_{ij}\sim N_d(\0,\bOmega_j) \\
	& \leq E[ \exp( 2\lambda (1-c_0) \sup_{\btheta\in\B(\btheta*;c_0,c_1)} | \frac{1}{nK}\sum_{j=1}^K\sum_{i=1}^n  \epsilon_{ij} \langle \tilde Y^N_{ij}, u \rangle |)]  \\ 
	& \leq \exp( \frac{4\lambda^2}{nK}(1-c_0)^2M )
\end{align*}
where the last inequality is due to the fact that $ \langle \tilde Y^N_{ij}, u \rangle \sim N(0, u^T\bOmega_j u) $. Thus, combine the bounds of \eqref{eqCon3} and \eqref{eqCon4}, we have
\begin{align*}
	E[ & \exp( 2\lambda \sup_{\btheta\in\B(\btheta*;c_0,c_1)} | \frac{1}{nK}\sum_{j=1}^K\sum_{i=1}^n  \epsilon_{ij} \lambda_j \langle  \bOmega_j (Y_{ij} -\frac{\bmu_0 + \bmu_1}{2} ), u \rangle |)] \\ 
	& \leq \exp( \frac{4\lambda^2}{nK}(1-c_0)^2M )\exp( \frac{4\lambda^2}{nK}(1-c_0)^2M^3(\frac{M}{2} + 1)^2\Delta_{max}^2 )\\
	& := \exp( \frac{\lambda^2}{nK}c_3).
\end{align*}
Next, we deal with the first term at the right hand side of \eqref{eqCon2.5}. 
\begin{align*}
E[ & \exp( \sup_{\btheta\in\B(\btheta*;c_0,c_1)} | \frac{2\lambda }{nK}\sum_{j=1}^K\sum_{i=1}^n  \epsilon_{ij} \{ \frac{\lambda_j }{\lambda_j+ (1-\lambda_j ) \exp(\bbeta_j^T(Y_{ij} - \frac{\bmu_0 + \bmu_1}{2})) }  -\lambda_j \} \langle  \bOmega_j (Y_{ij} -\frac{\bmu_0 + \bmu_1}{2} ), u \rangle |  ) ] \\
& \leq E[\exp( \sup_{\btheta\in\B(\btheta*;c_0,c_1)} | \frac{2\lambda }{nK}\sum_{j=1}^K\sum_{i=1}^n  \epsilon_{ij} \frac{1-c_0}{c_0} \bbeta_j^T(Y_{ij}- \frac{\bmu_0 + \bmu_1}{2}) \langle  \bOmega_j (Y_{ij} -\frac{\bmu_0 + \bmu_1}{2} ), u \rangle |  ) ]\\
& \leq E[\exp( \sup_{\btheta\in\B(\btheta*;c_0,c_1)} | \frac{2\lambda }{nK}\sum_{j=1}^K\sum_{i=1}^n  \epsilon_{ij} \frac{1-c_0}{c_0} \bbeta_j^T(Y_{ij}- \bmu_j^* + \bmu_j^* - \frac{\bmu_0 + \bmu_1}{2}) \\ & \cdot \langle  \bOmega_j (Y_{ij} - \bmu_j^* + \bmu_j^* -\frac{\bmu_0 + \bmu_1}{2} ), u \rangle |  ) ],
\end{align*}
and it can be further expanded as
\begin{align}
	& E[\exp( \sup_{\btheta\in\B(\btheta*;c_0,c_1)} | \frac{2\lambda }{nK}\sum_{j=1}^K\sum_{i=1}^n  \epsilon_{ij} \frac{1-c_0}{c_0} \bbeta_j^T(Y_{ij}- \bmu_j^*) \cdot \langle  \bOmega_j (Y_{ij} - \bmu_j^*), u \rangle |  ) ] \label{eqCon5} \\
	& \cdot E[\exp( \sup_{\btheta\in\B(\btheta*;c_0,c_1)} | \frac{2\lambda }{nK}\sum_{j=1}^K\sum_{i=1}^n  \epsilon_{ij} \frac{1-c_0}{c_0} \bbeta_j^T(\bmu_j^* - \frac{\bmu_0 + \bmu_1}{2}) \cdot \langle  \bOmega_j (Y_{ij} - \bmu_j^*), u \rangle |  ) ] \label{eqCon6} \\
	& \cdot E[\exp( \sup_{\btheta\in\B(\btheta*;c_0,c_1)} | \frac{2\lambda }{nK}\sum_{j=1}^K\sum_{i=1}^n  \epsilon_{ij} \frac{1-c_0}{c_0} \bbeta_j^T(Y_{ij}- \bmu_j^*) \cdot \langle  \bOmega_j (\bmu_j^* -\frac{\bmu_0 + \bmu_1}{2} ), u \rangle |  ) ] \label{eqCon7}\\
	& \cdot E[\exp( \sup_{\btheta\in\B(\btheta*;c_0,c_1)} | \frac{2\lambda }{nK}\sum_{j=1}^K\sum_{i=1}^n  \epsilon_{ij} \frac{1-c_0}{c_0} \bbeta_j^T(\bmu_j^* - \frac{\bmu_0 + \bmu_1}{2}) \cdot \langle  \bOmega_j (\bmu_j^* -\frac{\bmu_0 + \bmu_1}{2} ), u \rangle |  ) ].\label{eqCon8}
\end{align} 
For \eqref{eqCon6}, as $| \bbeta_j^T(\bmu_j^* - \frac{\bmu_0 + \bmu_1}{2}) | \leq (1+c_1)\Delta_{max}^2$
\begin{align*}
E[ & \exp( \sup_{\btheta\in\B(\btheta*;c_0,c_1)} | \frac{2\lambda }{nK}\sum_{j=1}^K\sum_{i=1}^n  \epsilon_{ij} \frac{1-c_0}{c_0} \bbeta_j^T(\bmu_j^* - \frac{\bmu_0 + \bmu_1}{2}) \cdot \langle  \bOmega_j (Y_{ij} - \bmu_j^*), u \rangle |  ) ] \\ 
& \leq E[ \exp( 2\lambda \frac{1-c_0}{c_0} (1+c_1)\Delta_{max}^2 | \frac{1}{nK}\sum_{j=1}^K\sum_{i=1}^n  \epsilon_{ij}  \langle  \bOmega_j (Y_{ij} - \bmu_j^*), u \rangle |  ) ] \\ 
& \leq E[ \exp( 2\lambda \frac{1-c_0}{c_0} (1+c_1)\Delta_{max}^2 | \frac{1}{nK}\sum_{j=1}^K\sum_{i=1}^n  \epsilon_{ij}  \langle  \tilde Y^N_{ij}, u \rangle |  ) ] \ \text{with}\ \tilde Y^N_{ij}\sim N(0, \bOmega_j) \\  
& \leq \exp( \frac{4\lambda^2}{nK}(\frac{1-c_0}{c_0})^2(1+c_1)^2\Delta_{max}^4M)
\end{align*}
by using the property of sub-Gaussian norm of normal random variables. 

For \eqref{eqCon7}, 
\begin{align*}
	E[ & \exp( \sup_{\btheta\in\B(\btheta*;c_0,c_1)} | \frac{2\lambda }{nK}\sum_{j=1}^K\sum_{i=1}^n  \epsilon_{ij} \frac{1-c_0}{c_0} \bbeta_j^T(Y_{ij}- \bmu_j^*) \cdot \langle  \bOmega_j (\bmu_j^* -\frac{\bmu_0 + \bmu_1}{2} ), u \rangle |  ) ] \\ 
	& \leq E[ \exp(2\lambda  \frac{1-c_0}{c_0} \sup_{\btheta\in\B(\btheta*;c_0,c_1)} | \frac{1}{nK}\sum_{j=1}^K |\langle  \bOmega_j (\bmu_j^* -\frac{\bmu_0 + \bmu_1}{2} ), u \rangle | \sum_{i=1}^n  \epsilon_{ij} \bbeta_j^T Y^N_{ij}  |  ) ] \\ 
	& \leq E[ \exp(2\lambda  \frac{1-c_0}{c_0} M^{3/2} (\frac{M}{2}+1)\Delta_{max} \sup_{\btheta\in\B(\btheta*;c_0,c_1)} | \frac{1}{nK}\sum_{j=1}^K \sum_{i=1}^n  \epsilon_{ij} \bbeta_j^T Y^N_{ij}  |  ) ],
\end{align*}
within which we have
\begin{align*}
	\sup_{\btheta\in\B(\btheta*;c_0,c_1)} | \frac{1}{nK}\sum_{j=1}^K \sum_{i=1}^n  \epsilon_{ij} \bbeta_j^T Y^N_{ij}  | & = \sup_{\btheta\in\B(\btheta*;c_0,c_1)} | \frac{1}{nK}\sum_{j=1}^K \sum_{i=1}^n  \epsilon_{ij} (\bmu_0 - \bmu_1)^T \bOmega_j Y^N_{ij}  | \\ 
	& = \sup_{\btheta\in\B(\btheta*;c_0,c_1)} | \frac{1}{nK}\sum_{j=1}^K \sum_{i=1}^n  \epsilon_{ij} (\bmu_0 - \bmu_1)^T \tilde Y^N_{ij}| \\ 
	& \leq \sup_{\btheta\in\B(\btheta*;c_0,c_1)} \|\bmu_0 - \bmu_1 \|_1 \max_{l \in [d]}  | \frac{1}{nK}\sum_{j=1}^K \sum_{i=1}^n  \epsilon_{ij}  \tilde Y^N_{ij,l}| \\ 
	& \leq \sup_{\btheta\in\B(\btheta*;c_0,c_1)} \sqrt{d} \|\bmu_0 - \bmu_1 \|_2 \max_{l \in [d]}  | \frac{1}{nK}\sum_{j=1}^K \sum_{i=1}^n  \epsilon_{ij}  \tilde Y^N_{ij,l}| \\ 
	& \leq \sqrt{dM(1+c_1)}\Delta_{max} \max_{l \in [d]}  | \frac{1}{nK}\sum_{j=1}^K \sum_{i=1}^n  \epsilon_{ij}  \tilde Y^N_{ij,l}|. 
\end{align*}
Thus, 
\begin{align*}
	E[ & \exp( \sup_{\btheta\in\B(\btheta*;c_0,c_1)} | \frac{2\lambda }{nK}\sum_{j=1}^K\sum_{i=1}^n  \epsilon_{ij} \frac{1-c_0}{c_0} \bbeta_j^T(Y_{ij}- \bmu_j^*) \cdot \langle  \bOmega_j (\bmu_j^* -\frac{\bmu_0 + \bmu_1}{2} ), u \rangle |  ) ] \\ 
	& \leq d\cdot E[ \exp(2\lambda  \frac{1-c_0}{c_0} M^{2} (\frac{M}{2}+1) \sqrt{d(1+c_1)}\Delta^2_{max} | \frac{1}{nK}\sum_{j=1}^K \sum_{i=1}^n  \epsilon_{ij}  \tilde Y^N_{ij,l}|  ) ] \\
	& \leq \exp( \frac{4\lambda^2}{nK} (\frac{1-c_0}{c_0})^2 M^{4} (\frac{M}{2}+1)^2 d(1+c_1)\Delta^4_{max} + \log(d) ).
\end{align*}

For \eqref{eqCon8}, as $$ \sup_{\btheta\in\B(\btheta*;c_0,c_1)} | \bbeta_j^T(\bmu_j^* - \frac{\bmu_0 + \bmu_1}{2}) \cdot \langle  \bOmega_j (\bmu_j^* -\frac{\bmu_0 + \bmu_1}{2} ), u \rangle | \leq (1+c_1)\Delta_{max}^3 M^{3/2} (\frac{M}{2}+1) $$ we have
\begin{align*}
	E[ & \exp( \sup_{\btheta\in\B(\btheta*;c_0,c_1)} | \frac{2\lambda }{nK}\sum_{j=1}^K\sum_{i=1}^n  \epsilon_{ij} \frac{1-c_0}{c_0} \bbeta_j^T(\bmu_j^* - \frac{\bmu_0 + \bmu_1}{2}) \cdot \langle  \bOmega_j (\bmu_j^* -\frac{\bmu_0 + \bmu_1}{2} ), u \rangle |  ) ] \\
	& \leq \exp( \frac{4\lambda^2}{nK} (\frac{1-c_0}{c_0})^2  (1+c_1)^2\Delta_{max}^6 M^{3} (\frac{M}{2}+1)^2).
\end{align*}
Combine the upper bounds of \eqref{eqCon6}-\eqref{eqCon8}, an upper bound of the product of these three terms is
\begin{align*}
	\exp( & \frac{4\lambda^2}{nK}(\frac{1-c_0}{c_0})^2(1+c_1)^2\Delta_{max}^4M) \\ 
	& \cdot \exp( \frac{4\lambda^2}{nK} (\frac{1-c_0}{c_0})^2 M^{4} (\frac{M}{2}+1)^2 d(1+c_1)\Delta^4_{max} + \log(d) )\\ 
	& \cdot \exp( \frac{4\lambda^2}{nK} (\frac{1-c_0}{c_0})^2  (1+c_1)^2\Delta_{max}^6 M^{3} (\frac{M}{2}+1)^2) \\ 
	& := \exp( \frac{\lambda^2}{nK} (c_5 d +c_6) + \log(d) ).
\end{align*}
Finally, we proceed to control \eqref{eqCon5}.
\begin{align*}
	E[ & \exp( \sup_{\btheta\in\B(\btheta*;c_0,c_1)} | \frac{2\lambda }{nK}\sum_{j=1}^K\sum_{i=1}^n  \epsilon_{ij} \frac{1-c_0}{c_0} \bbeta_j^T(Y_{ij}- \bmu_j^*) \cdot \langle  \bOmega_j (Y_{ij} - \bmu_j^*), u \rangle |  ) ]  \\ 
	& \leq E[ \exp( \sup_{\btheta\in\B(\btheta*;c_0,c_1)} | \frac{2\lambda }{nK}\sum_{j=1}^K\sum_{i=1}^n  \epsilon_{ij} \frac{1-c_0}{c_0} \langle \bmu_0 - \bmu_1, \tilde Y^N_{ij} \rangle \langle  \tilde Y^N_{ij}, u \rangle |  ) ] \\ 
	& \leq E[ \exp( \sup_{\btheta\in\B(\btheta*;c_0,c_1)} | \frac{2\lambda }{nK}\sum_{j=1}^K\sum_{i=1}^n  \epsilon_{ij} \frac{1-c_0}{c_0} \langle \bmu_0 - \bmu_1 - (\bmu_0^* - \bmu_1^*), \tilde Y^N_{ij} \rangle \langle  \tilde Y^N_{ij}, u \rangle |  ) ]\\ 
	& \cdot E[ \exp( \sup_{\btheta\in\B(\btheta*;c_0,c_1)} | \frac{2\lambda }{nK}\sum_{j=1}^K\sum_{i=1}^n  \epsilon_{ij} \frac{1-c_0}{c_0} \langle \bmu_0^* - \bmu_1^*, \tilde Y^N_{ij} \rangle \langle  \tilde Y^N_{ij}, u \rangle |  ) ].
\end{align*}
Let's define $$ \tilde W_{\tilde u, u} = \langle \tilde u, \frac{1}{nK} \sum_{i=1}^n \sum_{j=1}^K \epsilon_{ij} \tilde Y_{ij}^N \tilde Y_{ij}^{NT} u  \rangle, $$
and
\begin{align*}
	\tilde W_u & = \sup_{\tilde u \in \mathbb{S}^{d-1}} \langle \tilde u, \frac{1}{nK} \sum_{i=1}^n \sum_{j=1}^K \epsilon_{ij} \tilde Y_{ij}^N \tilde Y_{ij}^{NT} u  \rangle =  \sup_{\tilde u \in \mathbb{S}^{d-1}} \tilde W_{\tilde u, u}.
\end{align*}
We have
\begin{align*}
	\sup_{\btheta\in\B(\btheta*;c_0,c_1)} & \frac{1}{nK}\sum_{j=1}^K\sum_{i=1}^n  \epsilon_{ij} \langle \bmu_0 - \bmu_1 - (\bmu_0^* - \bmu_1^*), \tilde Y^N_{ij} \rangle \langle  \tilde Y^N_{ij}, u \rangle \\
	& \leq (\sup_{\btheta\in\B(\btheta*;c_0,c_1)} \| \bmu_0 - \bmu_1 - (\bmu_0^* - \bmu_1^*) \|_2 ) \tilde W_u \\ 
	& \leq \frac{1}{2 } M^{3/2} \Delta_{max} \tilde W_u \\
	& \leq M^{3/2} \Delta_{max} \max_{l \in [M_{net}]} \tilde W_{\tilde u_l, u}
\end{align*}
by using the covering net. Then for a fixed $\tilde u$ we have 
$$\tilde W_{\tilde u, u} = \langle \tilde u, \frac{1}{nK} \sum_{i=1}^n \sum_{j=1}^K \epsilon_{ij} \tilde Y_{ij}^N \tilde Y_{ij}^{NT} u  \rangle = \frac{1}{nK} \sum_{i=1}^n \sum_{j=1}^K \langle \tilde u, \tilde Y_{ij}^N\rangle \langle \tilde Y_{ij}^N, u \rangle,   $$
where we can use 
\begin{align*}
	\| \langle \tilde u, \tilde Y_{ij}^N\rangle \langle \tilde Y_{ij}^N, u \rangle \|_{\psi_1} \leq c_{\psi} \max\{ \| \langle \tilde u, \tilde Y_{ij}^N\rangle\|^2_{\psi_2}, \| \langle u, \tilde Y_{ij}^N\rangle\|^2_{\psi_2} \} \leq c_{\psi} M
\end{align*}
because if $x \sim N(0, \sigma^2)$, then $ \| x \|_{\psi_2} \leq c_{\psi}\sigma $.
Similarly, we can get
\begin{align*}
	\| \langle \bmu_0^* - \bmu_1^*, \tilde Y^N_{ij} \rangle \langle  \tilde Y^N_{ij}, u \rangle  \|_{\psi_1} \leq c_{\psi} \max\{ \|\bmu_0^* - \bmu_1^*\|_2M, M \} \leq c_{\psi} (M^{3/2}\Delta_{max} + M).
\end{align*}
Thus, for sufficiently small $\lambda$, we have 
\begin{align*}
	E[ & \exp( \sup_{\btheta\in\B(\btheta*;c_0,c_1)} | \frac{2\lambda }{nK}\sum_{j=1}^K\sum_{i=1}^n  \epsilon_{ij} \frac{1-c_0}{c_0} \bbeta_j^T(Y_{ij}- \bmu_j^*) \cdot \langle  \bOmega_j (Y_{ij} - \bmu_j^*), u \rangle |  ) ] \\ 
	& \leq \exp[ \frac{4\lambda^2}{nK} (\frac{1-c_0}{c_0})^2 c_{\psi}^2 \{ M^5\Delta_{max}^2 + (M^{3/2}\Delta_{max} + M)^2\} + \log(M_{net}) ] \\
	& \leq \exp( \frac{\lambda^2}{nK} c_7 + 2d ).
\end{align*}
Putting all the pieces together, we have 
\begin{align*}
	E[ \exp(\lambda W^{(\bmu)} ) ] & \leq E(\exp (2\lambda \max_{l\in [M_{net}]} W^{(\bmu)}_{\bmu_l} ) ) \\ 
	& \leq \sum_{l=1}^{M_{net}} E(\exp (2\lambda W^{(\bmu)}_{\bmu_l} ) ) \\
	& \leq M_{net} \exp( \frac{\lambda^2}{nK} c_3)\cdot \exp( \frac{\lambda^2}{nK} (c_5 d +c_6) + \log(d) )\cdot \exp( \frac{\lambda^2}{nK} c_7 + 2d )\\
	& \leq \exp( \frac{\lambda^2}{nK}( c_8 d + c_9 ) + 4d + \log(d) ).  
\end{align*}
Using Chernoff's approach, if we let 
\begin{align*}
	t & = \sqrt{ \frac{ 2(c_8 d + c_9 )(4d + \log(d) + \log(nK) ) }{nK} } \\
	\lambda & = \sqrt{ \frac{(4d + \log(d) + \log(nK))nK }{c_8 d + c_9} }
\end{align*}
then with probability at least $1-\frac{1}{nK}$, we have 
$$ W^{(\bmu)} \lesssim \sqrt{ \frac{ (c_8 d + c_9 )(4d + \log(d) + \log(nK) ) }{nK} }.  $$
It follows that, with probability at least $1 - \frac{1}{nK}$ we have
\begin{align*}
	\sup_{\btheta\in \B(\btheta^*;c_0,c_1)} & \| \bmu_1^n(\btheta) - \bmu_1(\btheta) \|_2  \\
	&= \sup_{\btheta\in \B(\btheta^*;c_0,c_1)} \|  [\frac{1}{nK}\sum_{j=1}^K\sum_{i=1}^n \gamma_{\btheta}(Y_{ij}) \bOmega_j ]^{-1} \{ \frac{1}{nK}\sum_{j=1}^K\sum_{i=1}^n \gamma_{\btheta}(Y_{ij}) \bOmega_j (Y_{ij} -\frac{\bmu_0 + \bmu_1}{2} )   \\
	& - \frac{1}{K}\sum_{j=1}^K E[\gamma_{\btheta}(Y_j) \bOmega_j (Y_j -\frac{\bmu_0 + \bmu_1}{2} )] \} + \{ [\frac{1}{nK}\sum_{j=1}^K\sum_{i=1}^n \gamma_{\btheta}(Y_{ij}) \bOmega_j ]^{-1} \\ 
	& - [\frac{1}{K}\sum_{j=1}^K E\gamma_{\btheta}(Y_j) \bOmega_j ]^{-1}  \} \frac{1}{K} \sum_{j=1}^K E\gamma_{\btheta}(Y_j) \bOmega_j (Y_j -\frac{\bmu_0 + \bmu_1}{2} ) \|_2 \\
	& \lesssim \sqrt{ \frac{ (c_8 d + c_9 )(4d + \log(d) + \log(nK) ) }{nK} } %+  \sqrt{\log(nd)}\sqrt{\frac{a_1 d + a_2}{n}} \\
	%& = O_p ( \sqrt{\frac{\log(nd)d}{n}} )
\end{align*}
because
\begin{align}
	 \|\{ [\frac{1}{nK} & \sum_{j=1}^K\sum_{i=1}^n \gamma_{\btheta}(Y_{ij}) \bOmega_j ]^{-1} - [\frac{1}{K}\sum_{j=1}^K E\gamma_{\btheta}(Y_j) \bOmega_j ]^{-1}  \} \frac{1}{K} \sum_{j=1}^K E\gamma_{\btheta}(Y_j) \bOmega_j (Y_j -\frac{\bmu_0 + \bmu_1}{2} )\|_2 \nonumber  \\
	& = \|[\frac{1}{nK}\sum_{j=1}^K\sum_{i=1}^n \gamma_{\btheta}(Y_{ij}) \bOmega_j ]^{-1} \{ \frac{1}{K}\sum_{j=1}^K E\gamma_{\btheta}(Y_j) \bOmega_j -  \frac{1}{nK}\sum_{j=1}^K\sum_{i=1}^n \gamma_{\btheta}(Y_{ij}) \bOmega_j \} \nonumber \\ 
	& \cdot [\frac{1}{K}\sum_{j=1}^K E\gamma_{\btheta}(Y_j) \bOmega_j ]^{-1} \frac{1}{K} \sum_{j=1}^KE \gamma_{\btheta}(Y_j) \bOmega_j (Y_j -\frac{\bmu_0 + \bmu_1}{2} )\|_2 \nonumber \\ 
	& \leq \| [\frac{1}{nK}\sum_{j=1}^K\sum_{i=1}^n \gamma_{\btheta}(Y_{ij}) \bOmega_j ]^{-1} \|_2 \| \frac{1}{K}\sum_{j=1}^K E\gamma_{\btheta}(Y_j) \bOmega_j -  \frac{1}{nK}\sum_{j=1}^K\sum_{i=1}^n \gamma_{\btheta}(Y_{ij}) \bOmega_j  \|_2 \label{eqCon9} \\ 
	& \cdot \{ \| [\frac{1}{K}\sum_{j=1}^K E\gamma_{\btheta}(Y_j) \bOmega_j ]^{-1} \frac{1}{K} \sum_{j=1}^K E\gamma_{\btheta}(Y_j) \bOmega_j (Y_j -\bmu_1 ) \|_2 \label{eqCon10} \\
	& + \| [\frac{1}{K}\sum_{j=1}^K E\gamma_{\btheta}(Y_j) \bOmega_j ]^{-1} \frac{1}{K} \sum_{j=1}^K E\gamma_{\btheta}(Y_j) \bOmega_j (\bmu_1 -\frac{\bmu_0 + \bmu_1}{2} ) \|_2  \} \label{eqCon11}
\end{align}
where we have \eqref{eqCon9} $\lesssim O_p(1/\sqrt{nK})$, \eqref{eqCon10} $=\|\A^{-1}\B\|_2 \lesssim c_{\A\B}=c_6\exp(-c_4\Delta_{min}^2) + \frac{M^{3/2}}{4}\Delta_{max} + \frac{M^{3/2}}{4}c_{\bmu}\exp(-c_4\Delta_{min}^2)\Delta_{max}$ and \eqref{eqCon11} $=\|\bmu_1 -\frac{\bmu_0 + \bmu_1}{2}\|_2 \leq \sqrt{M(1+c_1)}\Delta_{max}/2$, and it implies that when the global SNR $\Delta_{min}$ is large enough, then the rate of $\sup_{\btheta\in \B(\btheta^*;c_0,c_1)}\| \bmu_1^n(\btheta) - \bmu_1(\btheta) \|_2$ will be as claimed.

To summarize, with probability at least $1-\frac{K}{n}-\frac{1}{nK}$ we have
$$ \sup_{\btheta\in\B(\btheta^*;c_0,c_1)}d_2(M_n(\btheta),M(\btheta)) = O_p(\sqrt{ \frac{ d^2 + d\log(dnK) }{nK} }) + O_p(\sqrt{\log(nd)} \sqrt{\frac{Kd}{n}}). $$

\end{proof}

\section{Additional simulation results}

Figure \ref{add_fig1}--\ref{add_fig3} display the estimation results when $n=3,000$.

\begin{figure}
	\centering
	\includegraphics[width = 1\linewidth]{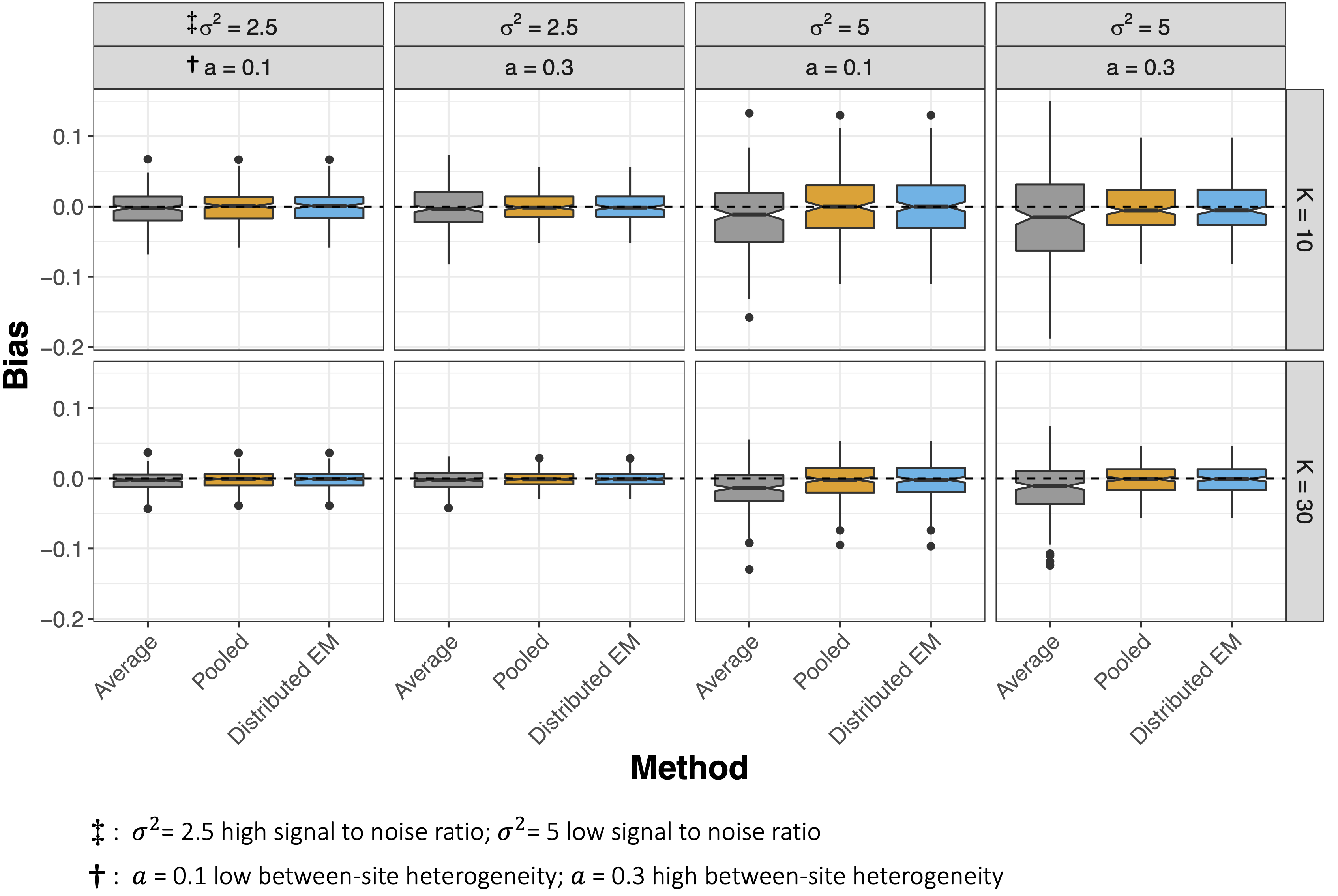}
	\caption{Empirical bias and variances of estimates of $\mu_{01}$ from the average estimator, the pooled estimator, and our distributed EM estimator, when $n=3,000$ under different settings of number of sites ($K$), signal to noise ratio ($\sigma^2$) and heterogeneity level ($a$).}
\label{add_fig1}
\end{figure}

% \begin{figure}
% 	\centering
% 	\includegraphics[width = 1\linewidth]{BIAS_3_C2_mu01_n3000}
% 	\caption{Estimation bias of $\bmu_{01}$ when $n=3000$.}
% \label{fig2}
% \end{figure}

\begin{figure}
	\centering
	\includegraphics[width = 1\linewidth]{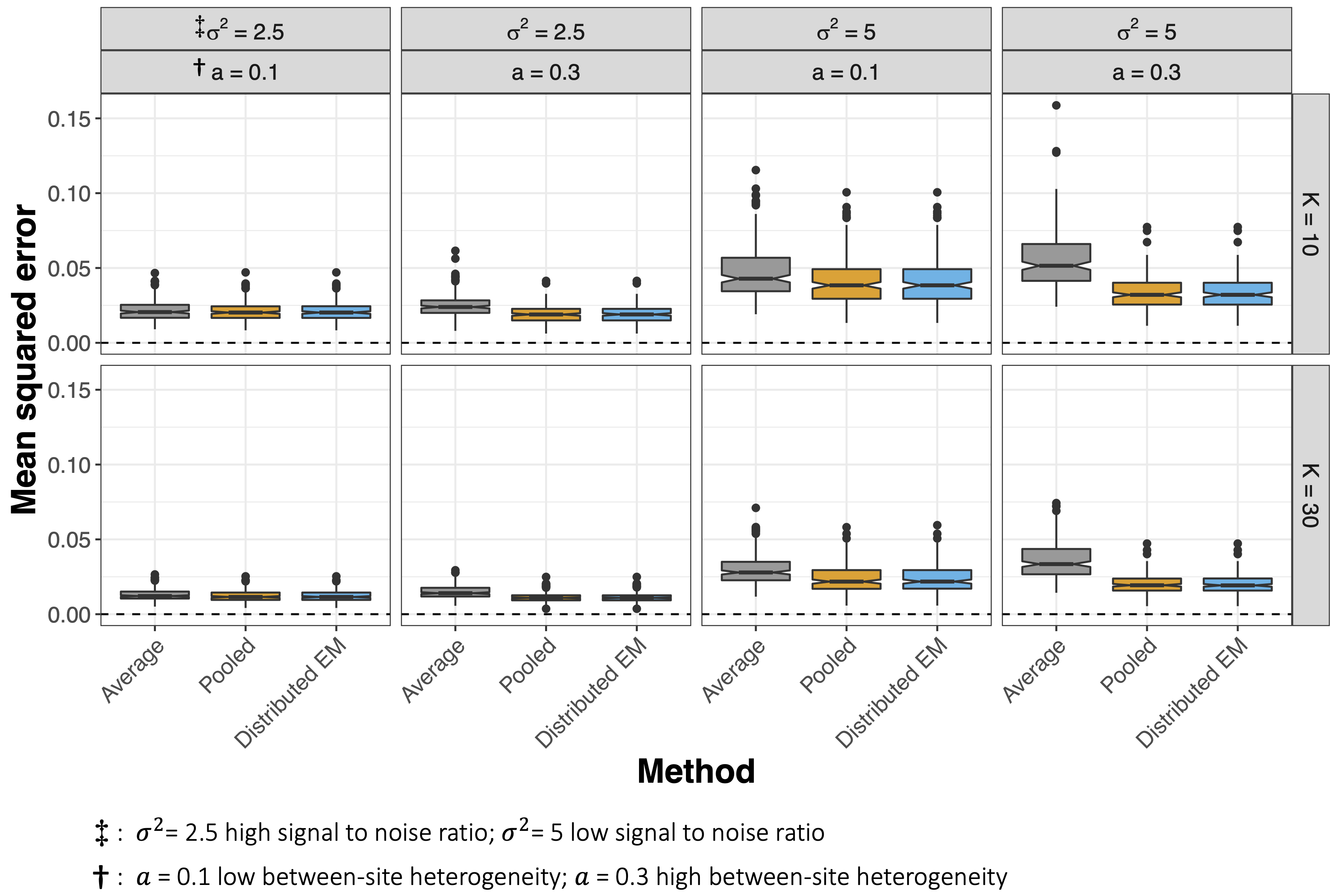}
	\caption{Mean squared error of estimates of $\bmu$ from the average estimator, the pooled estimator, and our distributed EM estimator, when $n=3,000$ under different settings of number of sites ($K$), signal to noise ratio ($\sigma^2$) and heterogeneity level ($a$).}
\label{add_fig3}
\end{figure}

%\newpage

%\section{Proofs}

% %\clearpage 
% \vskip 0.2in
% \bibliography{LCM}

\end{document}